\documentclass{article}
\usepackage{latexsym,calc}
\usepackage{jheppub}
\usepackage{hyperref}
\usepackage[none]{hyphenat}
\usepackage[utf8]{inputenc}
\usepackage[english]{babel}
\usepackage{latexsym}
\usepackage{physics}
\usepackage{pdflscape}
\usepackage{amsfonts}
\usepackage{accents}
\usepackage{amsthm}
\usepackage{amssymb}
\usepackage{mathtools}
\usepackage{scalerel}
\usepackage[dvipsnames]{xcolor}
\usepackage{calc}
\usepackage{pgfmath}
\usepackage{bbm}
\usepackage{setspace}
\usepackage{pgffor}
\usepackage{multicol}
\usepackage{tikz}
\usetikzlibrary{calc}
\usetikzlibrary{hobby}
\usetikzlibrary{decorations.markings}
\usetikzlibrary{decorations.pathreplacing}
\usepackage[numbers]{natbib}
\usepackage{array, makecell}

\allowdisplaybreaks

\newcommand\biopencrossl{%
  \mathrel{\scalerel*{>\kern-.4\LMpt\joinrel\blacktriangleleft}{x}}}
\newcommand\biopencrossr{%
  \mathrel{\scalerel*{\blacktriangleright\joinrel\kern-.4\LMpt<}{x}}}
\newcommand\opencrossl{%
  \mathrel{\scalerel*{>\kern-.4\LMpt\joinrel\vartriangleleft}{x}}}
\newcommand\opencrossr{%
  \mathrel{\scalerel*{\kern-.4\LMpt\joinrel\vartriangleright\joinrel\kern-.4\LMpt<}{x}}}
\newcommand\doublecross{%
  \mathrel{\scalerel*{\kern-.4\LMpt\joinrel\vartriangleright\joinrel\mathrel{\vartriangleleft}}{x}}}
  
\newcolumntype{M}[1]{>{\centering\arraybackslash}m{#1}}

\numberwithin{figure}{section}
\numberwithin{equation}{section}

\def\rlbicross{{\triangleright\!\!\!\blacktriangleleft}}
\def\lrbicross{{\blacktriangleright\!\!\!\triangleleft}}
\def\dcross{{\bowtie}}

\def\rbiprod{{\cdot\kern-.33em\triangleright\!\!\!<}}
\def\lbiprod{{>\!\!\!\triangleleft\kern-.33em\cdot}}

\def\calB{{\mathcal B}}

\def\calH{{\mathcal H}}

\def\calP{{\mathcal P}}

\newcommand{\C}{\mbox{${\mathbb C}$}}

\newcommand{\N}{\mbox{${\mathbb N}$}}
\newcommand{\Z}{\mbox{${\mathbb Z}$}}

\renewcommand{\o}{{}_{\scriptscriptstyle(1)}}
\renewcommand{\t}{{}_{\scriptscriptstyle(2)}}
\newcommand{\thr}{{}_{\scriptscriptstyle(3)}}
\newcommand{\fo}{{}_{\scriptscriptstyle(4)}}

\newcommand{\sso}{{}^{\scriptscriptstyle(1)}}
\newcommand{\sst}{{}^{\scriptscriptstyle(2)}}

\newcommand{\id}{{\rm id}}

\newcommand{\eps}{{\epsilon}}
\newcommand{\tens}{\mathop{\otimes}}
\newcommand{\la}{{\triangleright}}
\newcommand{\ra}{{\triangleleft}}

\newcommand{\<}{\left\langle}
\renewcommand{\>}{\right\rangle}

\newcommand\rcs{0.6} 

\newcommand{\harrow}[5]{
    \def\px{#1} 
    \def\w{#4};
    \def\flip{#5}
    \pgfmathsetmacro{\sl}{4*\w*\px};

    \coordinate (origin) at #2;
    \pgfpointanchor{origin}{center}
    \pgfgetlastxy{\xo}{\yo}
    
    \coordinate (target) at #3;
    \pgfpointanchor{target}{center}
    \pgfgetlastxy{\xt}{\yt}

    \ifnum\flip=0
        \coordinate (p1) at ($(origin) + (-\w*\px, \w*\px)$);
        \coordinate (p2) at ($(target) + (\w*\px, \w*\px)$);
        \coordinate (p3) at ($(target) + (-\w*\px, -\w*\px)$);
        \coordinate (p4) at ($(origin) + (\w*\px, -\w*\px)$);
    \else
        \coordinate (p1) at ($(origin) + (\w*\px, \w*\px)$);
        \coordinate (p2) at ($(target) + (-\w*\px, \w*\px)$);
        \coordinate (p3) at ($(target) + (\w*\px, -\w*\px)$);
        \coordinate (p4) at ($(origin) + (-\w*\px, -\w*\px)$);
    \fi

    \coordinate (center) at ($0.5*(target) + 0.5*(origin)$);
    \ifdim\xo<\xt
        \coordinate (v1) at ($(center) + (\sl, 0)$);
        \coordinate (v2) at ($(center) + ({-0.5*\sl}, {-(sqrt(3)/2)*\sl})$);
        \coordinate (v3) at ($(center) + ({-0.5*\sl}, {(sqrt(3)/2)*\sl})$);
    \else
        \coordinate (v1) at ($(center) + (-\sl, 0)$);
        \coordinate (v2) at ($(center) + ({0.5*\sl}, {(sqrt(3)/2)*\sl})$);
        \coordinate (v3) at ($(center) + ({0.5*\sl}, {-(sqrt(3)/2)*\sl})$);
    \fi
    
    \begin{scope}
        \fill[gray] (p1) -- (p2) -- (p3) -- (p4) -- (p1);        
        \fill[gray] (v1) -- (v2) -- (v3) -- (v1);  
    \end{scope}
}

\newcommand{\varrow}[5]{
    \def\px{#1} 
    \def\w{#4};
    \def\flip{#5}
    \pgfmathsetmacro{\sl}{4*\w*\px};

    \coordinate (origin) at #2;
    \pgfpointanchor{origin}{center}
    \pgfgetlastxy{\xo}{\yo}
    
    \coordinate (target) at #3;
    \pgfpointanchor{target}{center}
    \pgfgetlastxy{\xt}{\yt}

    \ifnum\flip=0
        \coordinate (p1) at ($(origin) + (-\w*\px, \w*\px)$);
        \coordinate (p2) at ($(target) + (-\w*\px, -\w*\px)$);
        \coordinate (p3) at ($(target) + (\w*\px, \w*\px)$);
        \coordinate (p4) at ($(origin) + (\w*\px, -\w*\px)$);
    \else
        \coordinate (p1) at ($(origin) + (-\w*\px, -\w*\px)$);
        \coordinate (p2) at ($(target) + (-\w*\px, \w*\px)$);
        \coordinate (p3) at ($(target) + (\w*\px, -\w*\px)$);
        \coordinate (p4) at ($(origin) + (\w*\px, \w*\px)$);
    \fi

    \coordinate (center) at ($0.5*(target) + 0.5*(origin)$);
    \ifdim\yo<\yt
        \coordinate (v1) at ($(center) + (0, \sl)$);
        \coordinate (v2) at ($(center) + ({-(sqrt(3)/2)*\sl}, {-0.5*\sl})$);
        \coordinate (v3) at ($(center) + ({(sqrt(3)/2)*\sl}, {-0.5*\sl})$);
    \else
        \coordinate (v1) at ($(center) + (0, -\sl)$);
        \coordinate (v2) at ($(center) + ({-(sqrt(3)/2)*\sl}, {0.5*\sl})$);
        \coordinate (v3) at ($(center) + ({(sqrt(3)/2)*\sl}, {0.5*\sl})$);
    \fi
    
    \begin{scope}
        \fill[gray] (p1) -- (p2) -- (p3) -- (p4) -- (p1);        
        \fill[gray] (v1) -- (v2) -- (v3) -- (v1);  
    \end{scope}
}

\newcommand{\LMt}[1]{\tilde{L}_{-}^{#1}}
\newcommand{\TPt}[1]{\tilde{T}_{+}^{#1}}
\newcommand{\TMt}[1]{\tilde{T}_{-}^{#1}}

\newcommand{\TP}[1]{T_{+}^{#1}}
\newcommand{\TM}[1]{T_{-}^{#1}}
\newcommand{\PP}[2]{M_{+}^{#1\otimes #2}}
\newcommand{\PM}[2]{M_{-}^{#1\otimes #2}}

\newcommand{\Si}{S^{-1}}
\newcommand{\pr}[1]{#1^\prime}
\newcommand{\cop}[2]{#1_{(#2)}}
\newcommand{\copp}[3]{#1_{(#2)(#3)}}
\renewcommand{\(}{\left(}
\renewcommand{\)}{\right)}
\newcommand{\rdir}[2]{F^{#1 \otimes #2}\left(\tau_{R}\right)}
\newcommand{\ldir}[2]{F^{#1 \otimes #2}\left(\tau_{L}\right)}
\newcommand{\rdur}[2]{F^{#1 \otimes #2}\left(\tilde{\tau}_{R}\right)}
\newcommand{\ldur}[2]{F^{#1 \otimes #2}\left(\tilde{\tau}_{L}\right)}
\newcommand{\rop}[3][]{F^{#2 \otimes #3}\left(\tau_{#1}\right)}

\newtheorem{thm}{Theorem}[section]
\newtheorem{defn}[thm]{Definition}
\newtheorem{prop}[thm]{Proposition}

\newtheorem{lemma}[thm]{Lemma}

\def\be{\begin{equation}}
\def\ee{\end{equation}}

\newcommand{\bes}{\begin{eqnarray}}
\newcommand{\ees}{\end{eqnarray}}

\definecolor{grey}{HTML}{DDDDDD}
\definecolor{gray}{HTML}{7E7E7E}

\makeatletter
\gdef\@fpheader{}
\makeatother

\begin{document}
\title{Ribbon operators in the  Semidual 
lattice code model}

\author[a,b]{{\bf Fred Soglohu,}}\emailAdd{fred@aims.edu.gh}
\affiliation[a]{Department of Mathematics, University of Ghana, Legon, Ghana}
\author[b, c]{{\bf Prince K. Osei,}}\emailAdd{posei@aims.edu.gh}
\affiliation[b]{African Institute for Mathematical Sciences (AIMS), Accra,  Ghana}
\affiliation[c]{ Quantum Leap Africa (QLA), AIMS RIC, Kigali, Rwanda}
\author[c]{{\bf Abdulmajid Osumanu}}\emailAdd{aosumanu@quantumleapafrica.org}

\date{\small\today}

\medskip

\abstract{
In this work, we provide a rigorous  definition of ribbon operators in the Semidual Kitaev lattice model and study their properties. These operators are essential for understanding quasi-particle excitations within topologically ordered systems. We show that the ribbon operators generate quasi-particle excitations at the ends of the ribbon and reveal themselves as irreducible representations of the Bicrossproduct quantum group $M(H)=H^{\text{cop}}\lrbicross H$ or  $M(H)^{\text{op}}$ depending on their chirality or local orientation. 
}

\maketitle

\section{Introduction \& Motivation} 

Topological quantum computing is an approach to the realization of quantum computing using \\
non-Abelian anyons or quasi particles in certain two dimensional quantum systems. 
Information is encoded in topological degrees of freedom with geometrical features that are robust against small local perturbations.
Topological phases of matter, a quantum system with such robust properties that depend on the topology of the underlying material provides a medium for topological quantum computing.
This proposal for fault-tolerant quantum computing is due to Alexei Kitaev \cite{AYK}.
The main class of  topological phases of matter in two dimensions are realised by lattice models, the Kitaev quantum double model. This approach is based on quantum many body systems exhibiting topological order. The realization of these orders involve the construction of lattice models with solvable Hamiltonians.
\\

 Kitaev models are define for a finite group $G$, with quasi-particles given by the representations of the quantum double $D(G)$, and are realized on a two-dimensional lattice. The well-known of the Kitaev models is the toric code, which is  based on the cyclic group $\Z_2$ \cite{AYK}. We refer to \cite{Delgago2d3d} for a recent review. 
  Excitations in the toric code model are in the form of Abelian anyons, limiting its universality for quantum computation.  In the case where $G$ is a non-Abelian group, the associated excitations become non-Abelian as well.
 These models were  generalized to that based  on a finite-dimensional semi simple Hopf algebra $H$ \cite{BMCA}.
The generalized Kitaev quantum double model allow for a broader range of applications and theoretical explorations. The generalised model and its relation to various topological quantum field theories (TQFTs) has been well studied \cite{CM,yan2022ribbon,Cowtan:2022enl,jia2023boundary}. Here excitations are characterized by the representations of the Drinfeld double $D(H)=H\dcross H^{*\text{op}} $. 
\\

The Kitaev quantum double model have also been studied in terms of Weak Hopf algebras. We also refer to \cite{Jia:2023xar}  where their study delved into the mathematical foundations of weak Hopf algebras, focusing on the interplay between weak Hopf symmetry and weak Hopf quantum double.
\\

Other class of topologically ordered spin models  such as the Levin-Wen string-net 
models \cite{Levin04, Levin06}  which are based on a representation category of $H$ are  also related to the Kitaev models \cite{Oliver09, Kadar08, Kadar09,Buerschaper2010}. The structure of excitations for these models is also well established \cite{KitaevKong,Lan14,Hu18}. One defines the so called ribbon operators on the Hilbert space that generate the excitations. 
\\



The fundamental mathematical object to manipulate anyons in the non-Abelean lattice code models. 
The concept of ribbon operators is pivotal in describing the creation, annihilation, or manipulation of quasi-particle excitations within topologically ordered system like the Kitaev models. These operators are represented as thick graphs or ribbons on a two-dimensional graph.  In the toric code, these are characterized by the  so called string operators. In the non-Abelian context, these operators become entangled, requiring the consideration of a thicker or more complex string of operators. In the quantum double model, the initial step involves defining  fundamental triangles which form the building blocks of the ribbon operators. This definition is subsequently extended to encompass longer ribbons through a process of induction \cite{AYK,yan2022ribbon}. 
The ribbon operators for the quantum double  model based on Hopf algebras have been well studied in \cite{AYK,yan2022ribbon, cowtan2022quantum, jia2023boundary}. 
In this case, they are understood as representations of either $D(H)^{*}$ or $D(H)^{*\text{op}}$. In a manner akin to the quantum double model approach, ribbon operators have also been established within the weak Hopf quantum double framework \cite{Jia:2023xar}.
\\

The semidual lattice code model recently introduced in \cite{FOM}, proposed a different class of generalised  models based on  the mirror bicrossproduct quantum group $M(H)=H^{\text{cop}}\lrbicross H$.\\
  The mirror bicrossproduct quantum group  is a special case of general  bicrossproduct quantum groups that originally emerged in the theory of quantum gravity
\cite{Majid:1994cy,SM,Majid94}  and have since been studied in that context  \cite{cath-bernd, OseiSchroers1, prince-bernd,PKO1,MO}.  It tuns out that the  mirror bicrossproduct quantum groups are related to  the quantum doubles $D(H)$ through the idea of semidualisation, a kind of quantum group Fourier transform  where one can exchange position and momentum degrees of freedom in an algebraic framework \cite{MaSch}. 
 They are also known to be mathematically related by  a Drinfeld and module algebra twist \cite{MO}. 
 \\

 The innovative feature of the semidual  model construction lies in the utilisation of the  canonical covariant representation action of $M(H)$ provided by the dual Hopf algebra $H^*$. An exactly solvable Hamiltonian is obtained, and a representation of the ground state is given in terms of tensor networks. In the current work, we provide a thorough  definition of  ribbon operators in the semidual Kitaev model in detail and explore some of their algebraic properties. 
 \\
 
 In the case of the semidual model, the ribbon operators reveal themselves as representations of $M(H)=H^{\text{cop}}\lrbicross H$ or  $M(H)^{\text{op}}= H\lrbicross H^{\text{op}}$. We adopt the definition of ribbons, as outlined in \cite{jia2023boundary}, which classifies them as type-A and type-B based on the chirality (or local orientation) of the triangles that constitute their structure. The chirality of the triangles dictates two distinct ribbons (either left-handed or right-handed) for each type, and thus requiring a different  definition of ribbon operators for each type. Ribbon operators of type-A are found to generate a representation of $H^{\text{cop}}\otimes H\cong M(H)$, while those of type-B generate a representation of $H\otimes H^{\text{op}}\cong M(H)^{\text{op}}$.
\\

 We then establish that, when provided with a ribbon, the ribbon operators exhibit commutativity with all terms in the Hamiltonian, excluding those linked to the two ends of the ribbon. This demonstrates that ribbon operators are instrumental in generating excitations solely at the ends of the ribbon. Further, show that geometric operators of the semidual model commutes with the ribbon operators at sites (i.e., a vertex and adjacent face) within a ribbon. 
\\ 

The paper is organized as follows: Section \ref{background} provides  the ingredients required for constructing the generalised Kitaev lattice models, encompassing directed graphs and Hopf algebras. Additionally, it discuss the semidualization of the quantum double  $D(H)$. In section \ref{sec:semidual_kitaev_model}, we provide an overview of  the semidual lattice model of \cite{FOM} and introduce the right-module triangle operators pivotal for the construction of the ribbon operators. Section \ref{sec:ribbon_operators} is dedicated to the construction of the ribbon operators and the study of their properties. The appendicies contain these calculations. Finally, in section \ref{sec:conclusion}, we present our concluding remarks. 
\\

\section{Background }
\label{background}
There are two main ingredients of the Kitaev lattice code models. Directed graphs and Hopf algebras. 

\subsection{Directed graphs}
\label{subsec:basic_graph_structures}

  For simplicity, we work with  a square lattice  $\Gamma = (E,V,F)$ a  set of interconnected line segments,
  with edge set $E$, vertex set $V$  and face set $F$. When the edges are endowed with directions, the graph is said to be directed. Let $e\in E$ be an edge in a directed graph, we denote by $\partial_{0}e$ the initial vertex of $e$ and $\partial_{1}e$ the terminal vertex of $e$. We say a graph has no self loops if there is no edge $e\in E$ that starts and ends at the same vertex, that is, $\forall e\in E$,  $\partial_{0}e \neq \partial_{1}e$. If there is at most one edge joining any two vertex of a given graph, we say that the graph has no duplicate edges. In other words, we say a graph has no duplicate edges if for any two edges $e, e\prime \in E$ such that $\partial_{i}e = \partial_{i}e\prime$ for $i=0,1$, then $e = e\prime$. A directed graph that has no self loops and no duplicate edges is called a simple directed graph. An example of a directed simple graph is illustrated in Figure \ref{fig:graph_with_basic_structures}.
\\

  Given  a directed simple graph $\Gamma$, there is a graph $\Gamma^{*} = (E^{*}, F^{*}, V^{*})$ induced by $\Gamma$, with $F^{*}$ vertex set of $\Gamma^{*}$ and $V^{*}$ it's face set. By this definition, we mean that the vertices and faces of $\Gamma^{*}$ are respectively defined by the faces and vertices of $\Gamma$. If $p$ is a face in the direct graph, then $p^{*}$ is a vertex in the dual graph and if $v$ is a vertex in the direct graph, then $v^{*}$ is a face in the dual graph. We call $\Gamma$ the direct graph and $\Gamma^{*}$ the corresponding dual graph.  If $\Gamma$ is simple and directed, then $\Gamma^{*}$ is also simple and directed. 
\\

    Let $\Gamma $ be a directed graph, a site $s := (v_{s},p_{s})$ is  a vertex $v\in V$ and an adjacent face $p\in F$ containing $v$. If two sites share the same face, we say they are face adjacent and if they share the same vertex, we say they are vertex adjacent. We use  a dotted line segment connecting $v$ and the adjacent face $p$  to represent the site as shown in Figure \ref{fig:graph_with_basic_structures}.  We refer to  \cite{AYK,BombinDeldado} for more details.

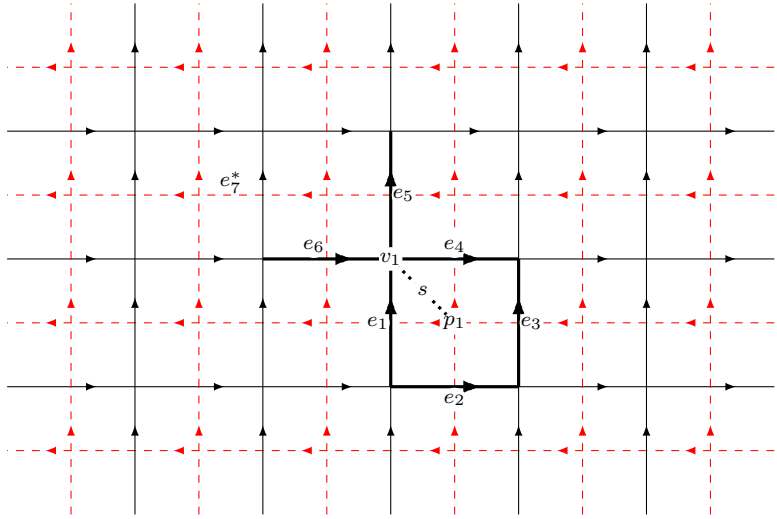
\begin{figure}[h!]
    \centering
    \newcommand\gghgs{3} 
\newcommand\ggvgs{5} 

\newcommand\gs{2}
\newcommand\gss{2*\gs}

\begin{tikzpicture}[scale = 0.85]
    	
	\begin{scope}[decoration={markings,mark=at position 0.7 with {\arrow{latex}}}]
    	\foreach \i in {-2,...,\gghgs}{
        	\foreach \j in {3,...,\ggvgs}{
        	    \draw[postaction={decorate}] (\i*\gs, \j*\gs) -- (\i*\gs + \gs, \j*\gs);
        		\draw[red, thin, dashed, postaction={decorate}] (\i*\gs + \gs, \j*\gs + 0.5*\gs) -- (\i*\gs, \j*\gs + 0.5*\gs);
            }
            \draw[red, thin, dashed, postaction={decorate}] (\i*\gs + \gs, 2.5*\gs) -- (\i*\gs, 2.5*\gs);
    	}
    		
    	\foreach \j in {2,...,\ggvgs}{
        	\foreach \i in {-1,...,\gghgs}{
        		\draw[postaction={decorate}] (\i*\gs, \j*\gs) -- (\i*\gs, \j*\gs + \gs);
        		\draw[red, thin, dashed, postaction={decorate}] (\i*\gs + 0.5*\gs, \j*\gs) -- (\i*\gs + 0.5*\gs, \j*\gs + \gs);
        	}
        	\draw[red, thin, dashed, postaction={decorate}] (-1.5*\gs, \j*\gs) -- (-1.5*\gs, \j*\gs + \gs);
    	}
    \end{scope}
    	
    \draw[very thick, dotted] (1*\gs, 4*\gs) -- (1.5*\gs, 3.5*\gs);
    \begin{scope}[very thick, decoration={markings,mark=at position 0.7 with {\arrow{latex}}}]
        \draw[postaction={decorate}] (1*\gs, 4*\gs) -- (1*\gs, 5*\gs);
        \draw[postaction={decorate}] (0*\gs, 4*\gs) -- (1*\gs, 4*\gs);
        \draw[postaction={decorate}] (1*\gs, 4*\gs) -- (2*\gs, 4*\gs);
        \draw[postaction={decorate}] (2*\gs, 3*\gs) -- (2*\gs, 4*\gs);
        \draw[postaction={decorate}] (1*\gs, 3*\gs) -- (2*\gs, 3*\gs);
        \draw[postaction={decorate}] (1*\gs, 3*\gs) -- (1*\gs, 4*\gs);
    \end{scope}
    	
    \draw[text=black, draw=white, fill=white] (1.25*\gs, 3.75*\gs) circle (5pt) node {\footnotesize $s$};
    	
    \draw[text=black, draw=white, fill=white] (0.9*\gs, 3.5*\gs) circle (5pt) node {\footnotesize $e_{1}$};
    \draw[text=black, draw=white, fill=white] (1.5*\gs, 2.9*\gs) circle (5pt) node {\footnotesize $e_{2}$};
    \draw[text=black, draw=white, fill=white]  (2.1*\gs, 3.5*\gs) circle (5pt) node {\footnotesize $e_{3}$};
    \draw[text=black, draw=white, fill=white] (1.5*\gs, 4.1*\gs) circle (5pt) node {\footnotesize $e_{4}$};
    \draw[text=black, draw=white, fill=white] (1.1*\gs, 4.5*\gs) circle (5pt) node {\footnotesize $e_{5}$};
    \draw[text=black, draw=white, fill=white] (0.4*\gs, 4.1*\gs) circle (5pt) node {\footnotesize $e_{6}$};

    \draw[text=black, draw=white, fill=white] (-0.25*\gs, 4.6*\gs) circle (5pt) node {\footnotesize $e^{*}_{7}$};
    	
    \draw[text=black, draw=white, fill=white] (1*\gs, 4*\gs) circle (5pt) node {\footnotesize $v_{1}$};
    \draw[text=black, draw=white, fill=white] (1.5*\gs, 3.5*\gs) circle (5pt) node {\footnotesize $p_{1}$};
    	
            
    	
\end{tikzpicture}
    \caption{An oriented 2d lattice with examples of direct edges ($e_{i}$), dual edges ($e_{i}^{*}$), vertices ($v_{i}$), faces ($p_{i}$) and sites ($s_{i}$).}
    \label{fig:graph_with_basic_structures}
\end{figure}

\subsection{Hopf algebra
}
\label{Hopf algebras}
Hopf algebras and their representation categories play a key role in the the construction of topological quantum computing models and topological field theories. In the following, we provide a brief review of main notations of Hopf algebras for the purpose of this work following  conventions in the book \cite{SM}.  
\\
A Hopf algebra over a field $\C$ is a vector space  $H$ 
which is an algebra, a coalgebra and has an antipode. As an algebra, it has a linear `product' $\mu:H\tens H \rightarrow H$ and unit $\eta:\C \rightarrow H$ satisfying 
$\mbox{associativity:  }\mu \circ (\id\tens\mu)=\mu\circ(\mu\tens\id)
$ and 
$\mbox{unitary: } \mu \circ(\eta\tens\id)=\id=\mu\circ(\id\tens\eta).
$
 As a coalgebra, it is equipped with a  linear `coproduct' $\Delta:H \rightarrow  H\tens H$  with notation $\Delta (h)= h_{\o}\otimes h_{\t}=h\sso\otimes h\sst$  and  counit $\eps:H \rightarrow \C$  which are algebra homomorphisms satisfying  the coassociativity condition $(\Delta\otimes \mbox{id})\circ \Delta=(\mbox{id}\otimes\Delta)\circ\Delta$ and counit condition  $(\epsilon \otimes \id)\circ\Delta =\id =(\id\otimes \epsilon)$ respectively. The maps $\mu$ and $\eta$ are morphisms of coalgebras.
The  antipode $S:H\rightarrow H$ is an anti-algebra map satisfying  $(Sh\o)h\t=h\o S h\t=\eps(h)$, for all $h\in H$. If $H$ is finite-dimensional, then $S^{-1}$ exist.
We denote by $H^{\tens n}$, $n\in \N$ the $n$-fold tensor product of $H$.  The composition of $n$ coproducts is the map $\Delta^{(n)}: H\rightarrow H^{\tens (n+1)}$  defined by $\Delta^{(n)}(h)=h_{\o} \tens h_{\t} \tens ... \tens h_{(n+1)}$. 
We denote by $H^*$ the dual Hopf algebra with dual pairing given by the non-degenerate bilinear map $ \langle \,,\,\rangle$ and by  $H^{\text{cop}}, H^{\text{op}}$ denote taking the opposite coproduct or opposite product in $H$ respectively. 
\\

A left action (or representation) of an algebra $H$ is a pair $(\la , V)$ where $V$ is a vector space and $\la$ is a linear map $H\tens V\rightarrow V$ such that 
\begin{equation}\label{module}
h\la v\in V,\quad (hg)\la v=h\la(g\la v), \quad 1\la v=v 
\end{equation}
In this case, we say that $V$ is a left $H$-module with respect to the action $\la$. There is an analogous notion for a right module algebra by the action $\ra$ so that for example,  if $\la$ is a left action, then 
\be \label{l to r action} v\ra h = (Sh)\la v\ee
 is a right action since $S$ is an antialgebra map. If $H$ acts on vector spaces $V,W$ then it also acts on $V\tens W $ by $h\la(v\tens w) = h\o\la v \tens h\t \la w$ for all $h\in H$, $v\in V$ and $w\in W$.
An algebra $A$ is said to be an $H$-module \textit{algebra} if $A$  is a left $H$-module and this action is covariant,  i.e. 
\begin{equation}\label{module algb}
h\la (ab)= (h\o\la a)(h\t\la b),\hspace{1cm} h\la 1=\eps(h),\;\; a\in A, \, h\in H.
\end{equation}
We say that  $(H,A)$ is a  \textit{covariant system}. 

\subsection{Quantum double and Semidualisation}

A  double crossproduct Hopf algebra $H_1\dcross H_2$ can be viewed  as a Hopf algebra $H$ which factorizes into two sub-Hopf algebras  and built on $H_1\tens H_2$ as a vector space.  One can then  extract the actions 
$
\la: H_2\tens H_1\rightarrow H_1 \quad \mbox{and} \quad \ra: H_2\tens H_1\rightarrow H_2
$
of each Hopf algebra on the vector space of the other satisfying certain  compatibility conditions.  The algebra in  $H_1\dcross H_2$ is constructed from the given actions as a double (both left and right) cross product.  The coproduct of $H_1\dcross H_2$ is the tensor one given by the coproduct of each factor and there is  a canonical  covariant left action  of $H_1\dcross H_2$ on $H_2^* $ as an algebra which leads to   a covariant system $(H_1\dcross H_2, H_2^*)$.
\\

The  {\em semidual} of the above  matched pair of data is constructed by dualising half of the  data. If the  $H_2$ data is dualised, then one obtains a left-right  bicrossproduct Hopf algebra  $H_2^*\lrbicross H_1$ which then acts covariantly on $H_2$ or a semidual covariant system $(H_2^*\lrbicross H_1,H_2)$. This is called the  $B$-model semidualisation \cite{MaSch}. There is also an $A$-model semidualisation where $H_1$  is dualised to obtain the right-left bicrossproduct $ H_2\rlbicross H_1^*$ acting on the left on $H_1$.
These ideas were  proposed in the context of  quantum gravity in a sense  that one can swap position and momentum generators in the algebraic setting \cite{Ma:pla, SM}. See also \cite{cath-bernd, OseiSchroers1, prince-bernd,PKO1,MO} where this has been well explored. 
\\

The Drinfeld quantum double can be viewed as a double crossproduct  $D(H)=H\dcross H^{*\text{op}}$ and  left covariant system $(D(H), H^{\text{cop}})$ \cite{SM}.  The $B$-model semidualisation maps this to a right covariant system $(H^{\text{cop}}\lrbicross H, H^{*\text{op}})$, which is $(H^{\text{cop}}\lrbicross H, H^*)$ as a left covariant system.
Here, we refer  to $H^{\text{cop}}\lrbicross H=M(H)$ as the `mirror product' bicrossproduct, similar to that in  \cite{MaSch}. We refer to \cite{SM} for a detailed account.
  The Hopf algebra  $M(H)$ has an  algebra 
\begin{eqnarray}
\label{sd4}
(a\otimes h)(b\otimes g) = a (h_{\o}b Sh_{\t})\otimes h_{\thr}g,\hspace{0.5cm}h,g\in H,\hspace{0.3cm}a,b\in H^{\text{cop}}.
\end{eqnarray}
in which  $H^{\text{cop}}\tens 1$ and $1\tens H$ appear as subalgebras but with commutation relation  fully determined by 
\be
\label{strengthing mirror}
hb :=  (1\otimes h)(b\otimes 1)= (h_{\o}b Sh_{\t})h_{\thr},
\ee
The coproduct is 
\be
\label{coproduct of M(H)}
\Delta(a\otimes h) =\, a_{\t}\otimes h_{\t}\otimes a_{\o}h_{\o}Sh_{\thr}\otimes h_{\fo}, \ee
and antipode is given by
\be
\label{antipode of M(H)}
S(a\tens h) =\, (1\tens Sh\t)(S(ah\o Sh\thr)\tens 1).
\ee
 This Hopf algebra acts covariantly on   $H^{*\text{op}}$ from the right according to 
  \be\label{mirroraction}
  \phi \ra (a \tens h)=\<a h\o,\phi\o\>\phi\t\<S h\t, \phi\thr\>,\; h\in H,\; a \in H^{\text{cop}},\;\phi \in H^{*}.
  \ee
  This gives rise to a covariant  left action on $H^*$   according to
\begin{equation}
\label{left covariant action on H*}
(a\otimes h)\triangleright \phi =\langle Sh_{\o}Sa ,\phi_{\o}\rangle\langle h_{\t},\phi_{\thr}\rangle \phi_{\t}
\end{equation}
leading to the covariant system $(H^{\text{cop}}\lrbicross H, H^*)$.
We refer to \cite{MO}  and to \cite{SM} the recent work for more details.

\section{Semidual lattice  code  model 
}
\label{sec:semidual_kitaev_model}
We provide a brief review of the semidual  lattice code model and refer to \cite{FOM} for details. Given the data on a directed graph $\Gamma$ of a 2d compact oriented surface $\Sigma$ and the Hopf algebras in the previous section, we describe the semidual lattice code model. For a finite dimensional semi-simple Hopf algebra $H$, this  is constructed from  the $M(H)$-module data in the covariant system $(M(H), H^{*})$ and  provides a graph representation of the mirror bicrossproduct quantum group $M(H)=H^{\text{cop}}\lrbicross H$.
 \\
 
Each edge of $\Gamma$ is assigned a Hilbert space $\mathcal{H}_{e} := H^{*}$. The total Hilbert space for $\Gamma$ is then given by 
\begin{equation}
  \calH = \mathcal{H}_{\Gamma} := \bigotimes_{e\in E} H^*. 
\end{equation}
 the $|E|$-fold tensor product of  $H^*$.

To each edge $e\in E$, we assign a family of linear operators $(L^{h}_{\pm})_e$, $(T^{a}_{\pm})_e$, indexed by elements of the Hopf algebras $h\in H$ and  $a\in H^{\text{cop}}$ respectively.  They only act on $H^*$ associated to the edge in question and trivially on the copies associated to other edges. They are obtained from  \eqref{left covariant action on H*}  as
\begin{equation}
    \label{triangle operators}
    \arraycolsep=1.2pt\def\arraystretch{1.2}
    \begin{array}{ccc}
        L^h_{+}(\phi) = \langle h,S\phi\o\phi\thr\rangle \phi\t, & \qquad & L^h_{-}(\phi) = \langle h,\phi\thr S^{-1}\phi\o\rangle\phi\t,\\
        T^{a}_{+}(\phi) = \langle Sa, \phi\o\rangle \phi\t, & \qquad & T^{a}_{-}(\phi) = \langle a, \phi\t\rangle \phi\o.
    \end{array}
\end{equation}
Here, the operators $L_+$ and $T_+$ are the canonical left action  \eqref{left covariant action on H*}  of the bicrossproduct $ M(H)$ on $H^*$. 
The $L_-$ and $T_-$ are also left actions obtained using the relations 
\begin{equation}\label{relations}
{S\circ L^{h}_{-}(\phi) =  L^{h}_{+}\circ S(\phi), \quad  S\circ T^{a}_{-}(\phi) = T^{a}_{+}\circ S(\phi).}
\end{equation}

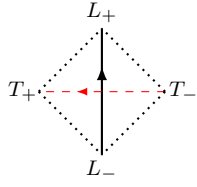
\begin{figure}[h!]
    \centering
    \begin{tikzpicture}[scale=0.7]
    
    \begin{scope}[decoration={markings,mark=at position 0.7 with {\arrow{latex}}}] 
        \draw[thick, postaction={decorate}] (4*\rcs, 1*\rcs) --  (4*\rcs, 5*\rcs);
        \draw[red, thin, dashed, postaction={decorate}] (6*\rcs, 3*\rcs) --  (2*\rcs, 3*\rcs);
    \end{scope}

    \draw[thick, dotted] (2*\rcs, 3*\rcs) --  (4*\rcs, 5*\rcs);
    \draw[thick, dotted] (2*\rcs, 3*\rcs) --  (4*\rcs, 1*\rcs);

    \draw[thick, dotted] (6*\rcs, 3*\rcs) --  (4*\rcs, 5*\rcs);
    \draw[thick, dotted] (6*\rcs, 3*\rcs) --  (4*\rcs, 1*\rcs);

    \node (LM) at (4*\rcs, 0.5*\rcs) {\footnotesize $L_{-}$};
    \node (LP) at (4*\rcs, 5.5*\rcs) {\footnotesize $L_{+}$};

    \node (TM) at (6.6*\rcs, 3*\rcs) {\footnotesize $T_{-}$};
    \node (TP) at (1.5*\rcs, 3*\rcs) {\footnotesize $T_{+}$};

            
            
            
    
\end{tikzpicture}
    \caption{An illustration of the action of the edge/triangle operators on an edge.}
    \label{fig:kitaev_convention}
\end{figure}

Since $H^*$ is finite-dimensional, $S^2 = \id$ and we can well identify $\phi\mapsto S(\phi)$, $\phi\in H^{*}$ with the reverse of the edge direction.
The $H^*$ is a left $H$-module with respect to the following two actions $L_\pm$ and also a left $H^{\text{cop}}$-module with respect to the actions $T_\pm$. 

A particular feature of this semidual model is that one requires another edge operator which is defined via the full  $M(H)$-module 
 as  $M^{a\tens h}_{\pm}= T^{a}_{\pm}\circ L^{h}_{\pm}$.  For a given directed edge $e$, we assign $L_-$ and $L_+$ to its starting and ending vertices respectively. Similarly, for the dual edge $\tilde{e}$, we assign $T_-$ and $T_+$ to its starting and ending faces respectively or
 \begin{equation}
        \arraycolsep=1.2pt\def\arraystretch{1.2}
        T^{a}(p, \psi) =
        \left\{\begin{array}{l}
            T^{a}_{+}(\psi), \quad \text {if $e$ points in the clockwise direction around $p$,} \\
            T^{a}_{-}(\psi), \quad \text {if $e$ points in the anti-clockwise direction around $p$.}
        \end{array}\right.
    \end{equation}
\\
Next, we define vertex and face operators  $A^h(v,p)$ and $B^a(v,p)$  for the semidual model on the Hilbert space $\mathcal{H}$. They require linear ordering of edges at each vertex and in each face which is specified by a site  $s=(v,p)$ \cite{AYK}.  
The vertex operators $A^{h}(v,p) : \mathcal{H} \rightarrow \mathcal{H}$ are restricted to the site $s=(v,p)$ and indexed by $h\in H$. They act as follows:
        \begin{equation}
            A^{h}(v,p)(\Gamma) = \cdots \otimes M^{h_{(\bar{1})}}(v, \psi^{1}) \otimes M^{h_{(\bar{2})}}(v, \psi^{2}) \otimes \cdots \otimes M^{h_{(\bar{n})}}(v, \psi^{n}) \otimes \cdots,
            \label{eqn:vertex_operator}
        \end{equation}
        where $\triangle^{(n-1)}_{M(H)}(1\otimes h) = h_{(\bar{1})} \otimes h_{(\bar{2})} \otimes \cdots h_{(\bar{n})}$ and $\{ \psi^{i} \}_{i=1}^{i=n}$ also called $star(v)$, is the set of edges that originate from $v$ or terminate at $v$ listed in a counter-clockwise order starting from and ending at $p$. 
        The vertex operator acts trivially on edges that do not belong to $star(v)$, this is represented with the ellipses on the extreme ends of the expression on the right hand side of equation (\ref{eqn:vertex_operator}). 
        An example of the action of the vertex operator is illustrated below,
        \begin{center}
            \begin{tikzpicture}

    \node at (1*\rcs, 5*\rcs) {$A^{h}(v,p)$};
    \node at (5*\rcs, 7.5*\rcs) {$\psi^{1}$};
    \node at (3*\rcs, 5.5*\rcs) {$\psi^{2}$};
    \node at (5*\rcs, 2.5*\rcs) {$\psi^{3}$};
    \node at (7*\rcs, 4.5*\rcs) {$\psi^{4}$};
    
    \begin{scope}[thick, decoration={markings, mark=at position 0.7 with {\arrow{latex}}}] 
        \draw[postaction={decorate}] (5*\rcs, 5*\rcs)--(3*\rcs, 5*\rcs);
        \draw[postaction={decorate}] (5*\rcs, 5*\rcs)--(7*\rcs, 5*\rcs);
        \draw[postaction={decorate}] (5*\rcs, 5*\rcs)--(5*\rcs, 3*\rcs);
        \draw[postaction={decorate}] (5*\rcs, 5*\rcs)--(5*\rcs, 7*\rcs);
    \end{scope}

    \node at (8*\rcs, 5*\rcs) {$:=$};
    \node at (12*\rcs, 7.5*\rcs) {$M_{-}^{h_{(\bar{1})}}(\psi^{1})$};
    \node at (10*\rcs, 5.5*\rcs) {$M_{-}^{h_{(\bar{2})}}(\psi^{2})$};
    \node at (12*\rcs, 2.5*\rcs) {$M_{-}^{h_{(\bar{3})}}(\psi^{3})$};
    \node at (14*\rcs, 4.5*\rcs) {$M_{-}^{h_{(\bar{4})}}(\psi^{4})$};
    
    \begin{scope}[thick, decoration={markings, mark=at position 0.7 with {\arrow{latex}}}] 
        \draw[postaction={decorate}] (12*\rcs, 5*\rcs)--(10*\rcs, 5*\rcs);
        \draw[postaction={decorate}] (12*\rcs, 5*\rcs)--(14*\rcs, 5*\rcs);
        \draw[postaction={decorate}] (12*\rcs, 5*\rcs)--(12*\rcs, 3*\rcs);
        \draw[postaction={decorate}] (12*\rcs, 5*\rcs)--(12*\rcs, 7*\rcs);
    \end{scope}

    

    \draw[thick, dotted] (5*\rcs, 5*\rcs)--(6*\rcs, 6*\rcs);
    \draw[thick, dotted] (12*\rcs, 5*\rcs)--(13*\rcs, 6*\rcs);

    \draw[white, fill=white] (5*\rcs, 5*\rcs) circle (3.5pt);
    \node at (5*\rcs, 5*\rcs) {$v$};
    \node at (6.3*\rcs, 6.3*\rcs) {$p$};

    \draw[white, fill=white] (12*\rcs, 5*\rcs) circle (3.5pt);
    \node at (12*\rcs, 5*\rcs) {$v$};
    \node at (13.3*\rcs, 6.3*\rcs) {$p$};






\end{tikzpicture}
        \end{center}

        The face operators $B^{a}(v,p) : \mathcal{H}_{\Gamma} \rightarrow \mathcal{H}_{\Gamma}$, also restricted to the site $s=(v,p)$ and indexed by $a\in H^{cop}$  act according to
        \begin{equation}
            B^{a}(v,p)(\Gamma) = \cdots \otimes T^{a_{(m)}}(p, \psi^{1}) \otimes T^{a_{(m-1)}}(p, \psi^{2}) \otimes \cdots \otimes T^{a_{(1)}}(p, \psi^{m}) \otimes \cdots,
            \label{eqn:face_operator}
        \end{equation}
        where $\triangle^{(m-1)}_{H^{cop}}(a) = a_{(m)} \otimes a_{(m-1)} \otimes \cdots \otimes a_{(1)}$ and $\{ \psi^{i} \}_{i=1}^{i=m}$ also called $bound(p)$, is the set of edges edges that bound $p$ listed in a counter-clockwise order starting from and ending at $v$. It also acts trivially on edges that do not belong to $bound(p)$ and this is also represented with the ellipses on the extreme ends of the expression on the right hand side on equation (\ref{eqn:face_operator}). An example of the action of the face operator is illustrated below:
        \begin{center}
            \begin{tikzpicture}

    \node at (0.5*\rcs, 5.5*\rcs) {$B^{a}(v,p)$};
    \node at (6.5*\rcs, 5.5*\rcs) {$\psi^{1}$};
    \node at (4.5*\rcs, 3.5*\rcs) {$\psi^{2}$};
    \node at (2.5*\rcs, 5.5*\rcs) {$\psi^{3}$};
    \node at (5*\rcs, 7.5*\rcs) {$\psi^{4}$};
    
    \begin{scope}[thick, decoration={markings, mark=at position 0.7 with {\arrow{latex}}}] 
        \draw[postaction={decorate}] (6*\rcs, 7*\rcs)--(6*\rcs, 4*\rcs);
        \draw[postaction={decorate}] (6*\rcs, 4*\rcs)--(3*\rcs, 4*\rcs);
        \draw[postaction={decorate}] (3*\rcs, 4*\rcs)--(3*\rcs, 7*\rcs);
        \draw[postaction={decorate}] (3*\rcs, 7*\rcs)--(6*\rcs, 7*\rcs);
    \end{scope}

    \node at (7.5*\rcs, 5.5*\rcs) {$:=$};
    \node at (15.5*\rcs, 5.5*\rcs) {$T_{+}^{a_{(4)}}(\psi^{1})$};
    \node at (12.5*\rcs, 3.5*\rcs) {$T_{+}^{a_{(3)}}(\psi^{2})$};
    \node at (9.5*\rcs, 5.5*\rcs) {$T_{+}^{a_{(2)}}(\psi^{3})$};
    \node at (12.5*\rcs, 7.5*\rcs) {$T_{+}^{a_{(1)}}(\psi^{4})$};
    
    \begin{scope}[thick, decoration={markings, mark=at position 0.7 with {\arrow{latex}}}] 
        \draw[postaction={decorate}] (14*\rcs, 7*\rcs)--(14*\rcs, 4*\rcs);
        \draw[postaction={decorate}] (14*\rcs, 4*\rcs)--(11*\rcs, 4*\rcs);
        \draw[postaction={decorate}] (11*\rcs, 4*\rcs)--(11*\rcs, 7*\rcs);
        \draw[postaction={decorate}] (11*\rcs, 7*\rcs)--(14*\rcs, 7*\rcs);
    \end{scope}

    

    \draw[thick, dotted] (4.5*\rcs, 5.5*\rcs)--(6*\rcs, 7*\rcs);
    \draw[thick, dotted] (12.5*\rcs, 5.5*\rcs)--(14*\rcs, 7*\rcs);

    \draw[white, fill=white] (6*\rcs, 7*\rcs) circle (3.5pt);
    \node at (6*\rcs, 7*\rcs) {$v$};
    \draw[white, fill=white] (4.5*\rcs, 5.5*\rcs) circle (3.5pt);
    \node at (4.5*\rcs, 5.5*\rcs) {$p$};

    \draw[white, fill=white] (14*\rcs, 7*\rcs) circle (3.5pt);
    \node at (14*\rcs, 7*\rcs) {$v$};
    \draw[white, fill=white] (12.5*\rcs, 5.5*\rcs) circle (3.5pt);
    \node at (12.5*\rcs, 5.5*\rcs) {$p$};






\end{tikzpicture}
        \end{center}
        
One can show that the directed graph $\Gamma$, equipped with these operators admits an $M(H)$-module  implemented by the 
relations \cite{FOM}
\begin{thm}{\cite{FOM}}
Let $H$ be a finite-dimensional Hopf algebra  satisfying $S^2 = \id$ with dual $H^*$ and the graph $\Gamma$  a square lattice as above.
 Then the operators  $A^h(v,p)$ and $ B^a(v,p) $ define an  $M(H)$ representation on $ H^{*\tens |E|}$ associated to each site $(v,p)$. Here $(a\tens h)$ acts by $A^h(v,p)\circ B^a(v,p) $, i.e. these enjoy the commutation relations 
\bes\label{AB-bicross}
 A^{h}(v,p)\circ B^{a} (v,p)&=& B^{(h\o aSh\t)}(v,p)\circ A^{h\thr}(v,p), \quad \forall h\in H,\;a\in H^{\text{cop}},\\[0.5\baselineskip]
 \label{AB-bicross1}
 A^{h}(v,p)\circ  A^{g}(v,p)&=& A^{hg}(v,p),\quad B^{a} (v,p)\circ  B^{b} (v,p) = B^{ab} (v,p).
\ees
\end{thm}

Let  $\ell\in H$, $k\in H^{\text{cop}}$ be normalized Haar integrals of the finite-dimensional Hopf algebras $H$ and $H^{\text{cop}}$ respectively. One can show that $A^{l}(v,p)$ depends only on $v$ and $B^{k}(v,p)$ depends only on $p$. Thus we set
\begin{equation}
    A(v) := A^{l}(v,p) \qquad \text{and} \qquad B(p) := B^{k}(v,p).
\end{equation}
The Hamiltonian for the system is defined by
\begin{equation}
    \mathfrak{H} = \sum_{v\in V} \big( \text{id} - A(v) \big) + \sum_{p\in F} \big( \text{id} - B(p) \big).
\end{equation}
The associated ground state is given as
\begin{equation}
    \mathcal{L} = \Big\{ \ket{\xi} \in \mathcal{H}_{\Gamma} : A(v)\ket{\xi} = \ket{\xi} ~ \text{and} ~ B(p)\ket{\xi} = \ket{\xi}, ~ \forall ~ v\in V ~ \text{and} ~ p\in F \Big\}.
\end{equation}

Similar to the left module actions, we  can introduce a right-module construction. To this end, we  define four sets of right-module edge operators. 
These are obtained using the relation \eqref{l to r action}, the right module edge operators are a family of linear operators $\tilde{L}^{h}_{\pm} : \mathcal{H}_{e} \rightarrow \mathcal{H}_{e}$ and  $\tilde{T}^{a}_{\pm} : \mathcal{H}_{e} \rightarrow \mathcal{H}_{e}$ with
\be
\tilde{L}^h_\pm = L_\pm^{Sh}, \quad \tilde{T}^a_\pm = T_ \pm^{S^{-1}(a)}, \quad  a\in H^{\text{cop}}, h\in H
\ee
More explicitly, we have 
  \begin{equation}\label{right triangle operators}
        \arraycolsep=1.2pt\def\arraystretch{1.2}
        \begin{array}{ccc}
            \tilde{L}_{+}^{h}(\psi) = \langle h, S\psi_{(3)}\psi_{(1)}\rangle \psi_{(2)}, &\qquad & \tilde{L}_{-}^{h}(\psi) = \langle h, \psi_{(1)}S^{-1}\psi_{(3)}\rangle \psi_{(2)}\\
            \tilde{T}_{+}^{a}(\psi) = \langle a, \psi_{(1)}\rangle \psi_{(2)}, &\qquad & \tilde{T}_{-}^{a}(\psi) = \langle S^{-1}a, \psi_{(2)}\rangle \psi_{(1)}.
        \end{array}
    \end{equation}
The $\tilde{L}^{h}$'s and $\tilde{T}^{a}$'s are composes to form another linear operator $\tilde{M}^{(a,h)} : \mathcal{H}_{e} \rightarrow \mathcal{H}_{e}$, which is given by
    \begin{equation}
        \tilde{M}^{(a,h)}_{\pm}(\psi) = \tilde{T}^{a}_{\pm} \circ \tilde{L}^{h}_{\pm}(\psi).
    \end{equation}
Here, $H^*$ is   right $H$-module with respect to  $\tilde{L}^{h}_{\pm}$  as well as a right $H^{\text{cop}}$-module with respect to  $\tilde{T}^{a}_{\pm} $.
Right module face and vertex operators $\tilde{A}^{h}(v,p)$ and $\tilde{B}^{a}(v,p)$  can be constructed in a similar way from the above edge operators.
One can easily check that they satisfy the relation 
\begin{equation}
    \arraycolsep=1.2pt\def\arraystretch{1.2}
    \begin{array}{c}
        \tilde{A}^{h}(v,p) \circ \tilde{B}^{a}(v,p) = \tilde{B}^{h_{(1)}aSh_{(2)}}(v,p) \circ \tilde{A}^{h_{(3)}}(v,p),
    \end{array}
    \label{eqn:relationship_between_right_module_vertex_and_face_operators_at_a_site}
\end{equation}
for all $a\in H^{\text{cop}}$ and $h\in H$. 
These generate  a representation of  $M(H)^{\text{cop}}$. 

\section{Ribbon Operators in the Semidual Kitaev Model}
\label{sec:ribbon_operators}
In this section, we will construct the ribbon operators for the semidual model on closed surfaces and explore some of their algebraic properties.  
In the case of the semidual model based on $M(H)$, the ribbon operators reveal themselves as representations of  $M(H)$.
The fundamental building blocks of a ribbon are direct and dual triangles. A direct triangle is a triangle on the original lattice, consisting of a directed edge and two adjacent sites such that the edge connects the two sites. The dual triangle is a direct triangle on the dual lattice.

\subsection{Directed ribbons and ribbon types}

   Let $\Gamma $ be a directed graph, a direct path $\calP$ of length $n$ is a list of successive direct edges $\calP := \left(e_{1}, e_{2}, \ldots e_{n}  \right)$ such that $\partial_{1}e_{i} = \partial_{0}e_{i+1}$ or $\partial_{1}e_{i} = \partial_{0}\bar{e}_{i+1}$ for $i=1, 2, \ldots, n-1$. Similarly, a dual path $\calP^{*}$ of length $n$ is a list of successive dual edges $\calP^{*} := \left(e_{1}^{*}, e_{2}^{*}, \ldots, e_{n}^{*}\right)$ such that $\partial_{1}e_{i}^{*} = \partial_{0}e_{i+1}^{*}$ or $\partial_{1}e_{i}^{*} = \partial_{0}\bar{e}_{i+1}^{*}$ for $i=1, 2, \ldots, n-1$.

\begin{defn}[Triangles]
    Let $\Gamma = (E,V,F)$ be a directed graph, a triangle is made up of two sites  which forms the legs of the triangle and either a dual edge or a direct edge which forms the base of the triangle. A direct triangle consists of a direct edge $e$ representing the base and two face adjacent sites $s_{0}$ and $s_{1}$ representing the legs such that $\partial_{0}e = v_{s_{0}}$ and $\partial_{1}e = v_{s_{1}}$ or $\partial_{0}\bar{e} = v_{s_{0}}$ and $\partial_{1}\bar{e} = v_{s_{1}}$. We denote a direct triangle by $\tau := (s_{0}, e, s_{1})$ or $\tau := (\partial_{0}\tau, e_{\tau}, \partial_{1}\tau)$. A dual triangle consists of a dual edge $e^{*}$ representing the base and two vertex adjacent sites $s_{0}$ and $s_{1}$ representing the legs such that $\partial_{0}e^{*} = p_{s_{0}}$ and $\partial_{1}e^{*} = p_{s_{1}}$ or $\partial_{0}\bar{e}^{*} = p_{s_{0}}$ and $\partial_{1}\bar{e}^{*} = p_{s_{1}}$; we will denote a dual triangle by $\tau := (s_{0}, e^{*}, s_{1})$ or $\tau := (\partial_{0}\tau, e^{*}_{\tau}, \partial_{1}\tau)$. Using this notation, we mean $s_{i}$ and $\partial_{i}\tau$ represent the initial   site for $i=0$ and the terminal site of $\tau$ $i=1$. This imposes a lateral orientation on a triangle, that is to say triangles point from their initial sites to their terminal sites. 
\end{defn}
Triangles  have local orientation or  chirality.  A direct or dual triangle is called a left-handed (right-handed) triangle if the edge of the triangle is on the left-hand (right-hand) side when one passes through the triangle along its positive direction. Left-handed (right-handed) triangles are denoted by $\tau_L$ and $\tilde{\tau}_L$  ($\tau_R$ and $\tilde{\tau}_R$). This is described  in  Figure \ref{fig:graph_with_triangles_and_ribbons}.

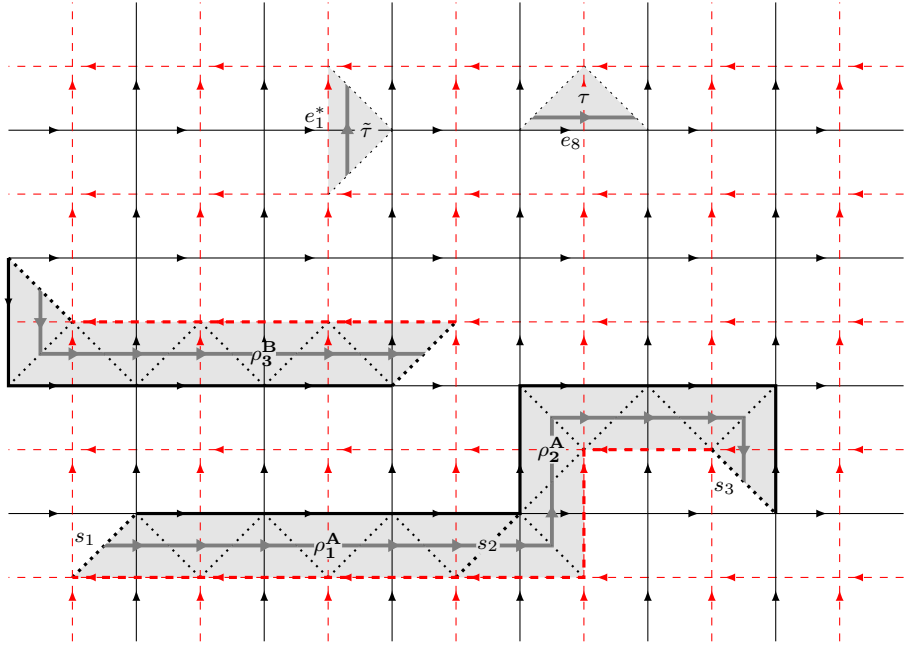
\begin{figure}[h!]
    \centering
    \newcommand\gghgs{6} 
\newcommand\ggvgs{4} 

\newcommand\gs{2}
\newcommand\gss{2*\gs}

\begin{tikzpicture}[scale = 0.85]
    	
	\draw[thick, dotted] (2.5*\gs, 3.5*\gs) -- (3*\gs, 4*\gs);
    \draw[thick, dotted] (3*\gs, 4*\gs) -- (2.5*\gs, 4.5*\gs);
    \fill [gray!20] (2.5*\gs, 3.5*\gs) -- (3*\gs, 4*\gs) -- (2.5*\gs, 4.5*\gs) -- (2.5*\gs, 3.5*\gs);
    	
    \draw[thick, dotted] (4*\gs, 4*\gs) -- (4.5*\gs, 4.5*\gs);
    \draw[thick, dotted] (4.5*\gs, 4.5*\gs) -- (5*\gs, 4*\gs);
    \fill [gray!20] (4*\gs, 4*\gs) -- (4.5*\gs, 4.5*\gs) -- (5*\gs, 4*\gs) -- (4*\gs, 4*\gs);
    	
    	
    	
    \fill [gray!20] (0*\gs, 3*\gs) -- (0.5*\gs, 2.5*\gs) -- (3.5*\gs, 2.5*\gs) -- (3*\gs, 2*\gs)
    	 -- (0*\gs, 2*\gs) -- (0*\gs, 3*\gs);
      
    \draw[very thick, dotted] (0*\gs, 3*\gs) -- (0.5*\gs, 2.5*\gs);
    \draw[thick, dotted] (0.5*\gs, 2.5*\gs) -- (0*\gs, 2*\gs);
    \draw[thick, dotted] (0.5*\gs, 2.5*\gs) -- (1*\gs, 2*\gs);
    \draw[thick, dotted] (1*\gs, 2*\gs) -- (1.5*\gs, 2.5*\gs);
    \draw[thick, dotted] (1.5*\gs, 2.5*\gs) -- (2*\gs, 2*\gs);
    \draw[thick, dotted] (2*\gs, 2*\gs) -- (2.5*\gs, 2.5*\gs);
    \draw[thick, dotted] (2.5*\gs, 2.5*\gs) -- (3*\gs, 2*\gs);
    \draw[very thick, dotted] (3*\gs, 2*\gs) -- (3.5*\gs, 2.5*\gs);
    	
    \fill [gray!20] (0.5*\gs, 0.5*\gs) -- (4.5*\gs, 0.5*\gs) -- (4.5*\gs, 1.5*\gs) -- (5.5*\gs, 1.5*\gs)
    	 -- (6*\gs, 1*\gs)  -- (6*\gs, 2*\gs) -- (4*\gs, 2*\gs) -- (4*\gs, 1*\gs) --  (1*\gs, 1*\gs) 
    	 -- (0.5*\gs, 0.5*\gs);
      
    \draw[very thick, dotted] (0.5*\gs, 0.5*\gs) -- (1*\gs, 1*\gs);
    \draw[thick, dotted] (1.5*\gs, 0.5*\gs) -- (1*\gs, 1*\gs);
    \draw[thick, dotted] (1.5*\gs, 0.5*\gs) -- (2*\gs, 1*\gs);
    \draw[thick, dotted] (2.5*\gs, 0.5*\gs) -- (2*\gs, 1*\gs);
    \draw[thick, dotted] (2.5*\gs, 0.5*\gs) -- (3*\gs, 1*\gs);
    \draw[thick, dotted] (3.5*\gs, 0.5*\gs) -- (3*\gs, 1*\gs);
    \draw[very thick, dotted] (3.5*\gs, 0.5*\gs) -- (4*\gs, 1*\gs);
    \draw[thick, dotted] (4.5*\gs, 0.5*\gs) -- (4*\gs, 1*\gs);
    \draw[thick, dotted] (4.5*\gs, 1.5*\gs) -- (4*\gs, 1*\gs);
    \draw[thick, dotted] (4.5*\gs, 1.5*\gs) -- (4*\gs, 2*\gs);
    \draw[thick, dotted] (4.5*\gs, 1.5*\gs) -- (5*\gs, 2*\gs);
    \draw[thick, dotted] (5.5*\gs, 1.5*\gs) -- (5*\gs, 2*\gs);
    \draw[thick, dotted] (5.5*\gs, 1.5*\gs) -- (6*\gs, 2*\gs);
    \draw[very thick, dotted] (5.5*\gs, 1.5*\gs) -- (6*\gs, 1*\gs);
    	
	\begin{scope}[decoration={markings,mark=at position 0.4 with {\arrow{latex}}}]
    	\foreach \i in {0,...,\gghgs}{
        	\foreach \j in {1,...,\ggvgs}{
        	    \draw[postaction={decorate}] (\i*\gs, \j*\gs) -- (\i*\gs + \gs, \j*\gs);
        		\draw[red, thin, dashed, postaction={decorate}] (\i*\gs + \gs, \j*\gs + 0.5*\gs) -- (\i*\gs, \j*\gs + 0.5*\gs);
            }
            \draw[red, thin, dashed, postaction={decorate}] (\i*\gs + \gs, 0.5*\gs) -- (\i*\gs, 0.5*\gs);
    	}
    		
    	\foreach \j in {0,...,\ggvgs}{
        	\foreach \i in {1,...,\gghgs}{
        		\draw[postaction={decorate}] (\i*\gs, \j*\gs) -- (\i*\gs, \j*\gs + \gs);
        		\draw[red, thin, dashed, postaction={decorate}] (\i*\gs + 0.5*\gs, \j*\gs) -- (\i*\gs + 0.5*\gs, \j*\gs + \gs);
        	}
        	\draw[red, thin, dashed, postaction={decorate}] (0.5*\gs, \j*\gs) -- (0.5*\gs, \j*\gs + \gs);
    	}

        \draw[postaction={decorate}] (0*\gs, 3*\gs) -- (0*\gs, 2*\gs);
    \end{scope}

    \begin{scope}
    	\draw[very thick] (1*\gs, 1*\gs) -- (4*\gs, 1*\gs) -- (4*\gs, 2*\gs) -- (6*\gs, 2*\gs) -- (6*\gs, 1*\gs);
        \draw[red, very thick, dashed] (5.5*\gs, 1.5*\gs) -- (4.5*\gs, 1.5*\gs) -- (4.5*\gs, 0.5*\gs) -- (3.5*\gs, 0.5*\gs) -- (2.5*\gs, 0.5*\gs) -- (1.5*\gs, 0.5*\gs) -- (0.5*\gs, 0.5*\gs);
    \end{scope}

    \begin{scope}
    	\draw[red, very thick, dashed] (3.5*\gs, 2.5*\gs) -- (0.5*\gs, 2.5*\gs);
        \draw[very thick] (0*\gs, 3*\gs) -- (0*\gs, 2*\gs) -- (3*\gs, 2*\gs);
    \end{scope}

    \varrow{\gs}{(2.65*\gs, 3.65*\gs)}{(2.65*\gs, 4.35*\gs)}{0.015}{1}
    \harrow{\gs}{(4.1*\gs, 4.1*\gs)}{(4.9*\gs, 4.1*\gs)}{0.015}{1}

    \varrow{\gs}{(0.25*\gs, 2.75*\gs)}{(0.25*\gs, 2.25*\gs)}{0.015}{0}
    \harrow{\gs}{(0.25*\gs, 2.25*\gs)}{(0.75*\gs, 2.25*\gs)}{0.015}{1}
    \harrow{\gs}{(0.75*\gs, 2.25*\gs)}{(1.25*\gs, 2.25*\gs)}{0.015}{0}
    \harrow{\gs}{(1.25*\gs, 2.25*\gs)}{(1.75*\gs, 2.25*\gs)}{0.015}{1}
    \harrow{\gs}{(1.75*\gs, 2.25*\gs)}{(2.25*\gs, 2.25*\gs)}{0.015}{0}
    \harrow{\gs}{(2.25*\gs, 2.25*\gs)}{(2.75*\gs, 2.25*\gs)}{0.015}{1}
    \harrow{\gs}{(2.75*\gs, 2.25*\gs)}{(3.25*\gs, 2.25*\gs)}{0.015}{0}

    \harrow{\gs}{(0.75*\gs, 0.75*\gs)}{(1.25*\gs, 0.75*\gs)}{0.015}{1}
    \harrow{\gs}{(1.25*\gs, 0.75*\gs)}{(1.75*\gs, 0.75*\gs)}{0.015}{0}
    \harrow{\gs}{(1.75*\gs, 0.75*\gs)}{(2.25*\gs, 0.75*\gs)}{0.015}{1}
    \harrow{\gs}{(2.25*\gs, 0.75*\gs)}{(2.75*\gs, 0.75*\gs)}{0.015}{0}
    \harrow{\gs}{(2.75*\gs, 0.75*\gs)}{(3.25*\gs, 0.75*\gs)}{0.015}{1}
    \harrow{\gs}{(3.25*\gs, 0.75*\gs)}{(3.75*\gs, 0.75*\gs)}{0.015}{0}
    \harrow{\gs}{(3.75*\gs, 0.75*\gs)}{(4.25*\gs, 0.75*\gs)}{0.015}{1}
    \varrow{\gs}{(4.25*\gs, 0.75*\gs)}{(4.25*\gs, 1.25*\gs)}{0.015}{0}
    \varrow{\gs}{(4.25*\gs, 1.25*\gs)}{(4.25*\gs, 1.75*\gs)}{0.015}{1}
    \harrow{\gs}{(4.25*\gs, 1.75*\gs)}{(4.75*\gs, 1.75*\gs)}{0.015}{0}
    \harrow{\gs}{(4.75*\gs, 1.75*\gs)}{(5.25*\gs, 1.75*\gs)}{0.015}{1}
    \harrow{\gs}{(5.25*\gs, 1.75*\gs)}{(5.75*\gs, 1.75*\gs)}{0.015}{0}
    \varrow{\gs}{(5.75*\gs, 1.75*\gs)}{(5.75*\gs, 1.25*\gs)}{0.015}{1}
    	
    \draw[text=black, draw=white, fill=white] (2.4*\gs, 4.1*\gs) circle (5pt) node {\footnotesize $e^{*}_{1}$};
    \draw[text=black, draw=gray!20, fill=gray!20] (2.8*\gs, 4*\gs) circle (5pt) node {\footnotesize $\tilde{\tau}$};
    	
    \draw[text=black, draw=gray!20, fill=gray!20] (4.5*\gs, 4.25*\gs) circle (5pt) node {\footnotesize $\tau$};
    \draw[text=black, draw=white, fill=white, line width=0mm] (4.4*\gs, 3.9*\gs) circle (5pt) node {\footnotesize $e_{8}$};
    	
    	
    \draw[text=black, draw=gray!20, fill=gray!20] (2*\gs, 2.25*\gs) circle (5pt) node {\footnotesize $\mathbf{\rho^{B}_{3}}$};
    	
    \draw[text=black, draw=white, fill=white] (0.6*\gs, 0.8*\gs) circle (5pt) node {\footnotesize $s_{1}$};
    \draw[text=black, draw=gray!20, fill=gray!20] (2.5*\gs, 0.75*\gs) circle (6pt) node {\footnotesize $\mathbf{\rho^{A}_{1}}$};
    \draw[text=black, draw=gray!20, fill=gray!20] (3.75*\gs, 0.75*\gs) circle (5pt) node {\footnotesize $s_{2}$};
    \draw[text=black, draw=gray!20, fill=gray!20] (4.25*\gs, 1.5*\gs) circle (6pt) node {\footnotesize $\mathbf{\rho^{A}_{2}}$};
    \draw[text=black, draw=white, fill=white] (5.6*\gs, 1.2*\gs) circle (5pt) node {\footnotesize $s_{3}$};

            
    	
\end{tikzpicture}
    \caption{An example of an arbitrary shaped oriented 2d-graph with examples of a direct (dual) triangle $\tau_{i}$($\tilde{\tau}_{i}$) and type-$A$ (type-$B$) ribbon $\rho^{A}_{i}$($\rho^{B}_{i}$). The dual ribbon paths are denoted by the thick dashed lines and the direct ribbon paths are denoted by the thick continuous lines.}
    \label{fig:graph_with_triangles_and_ribbons}
\end{figure}

A ribbon $\rho$ of length $n\in\mathbb{N}$ is a sequence of $n$ mutually non overlapping consecutive triangles such that the ending site of a triangle is the starting site of the next triangle. A ribbon is called direct (dual) if it is made up of only direct (dual) triangles, otherwise it is called proper. We will denote the starting site (end site) of $\rho$ by $\partial_{0}\rho$ ($\partial_{1}\rho$). A ribbon $\rho$ is closed if the starting and ending site are the same (that is $\partial_{0}\rho = \partial_{1}\rho$), else it is said to be open. There are two types of ribbons that are classified by the chirality or local orientation of the triangles composing them. A directed ribbon is called type-A (type-B) if all the direct triangles in it are left-handed (right-handed). Thus the type-A ribbons denoted by $\rho_{A}$, consist of left-handed direct triangles and right-handed dual triangles and the type-B ribbons denoted by $\rho_{B}$, consist of right-handed direct triangles and left-handed dual triangles. These are illustrated in \ref{fig:graph_with_triangles_and_ribbons}. The different types are considered separately.

\subsection{Ribbon operators for the semidual model}

For an element $a\otimes h \in H^{\text{cop}}\otimes H$ and a directed triangle $\tau$, we denote the ribbon operators by $F^{(a, h)}(\tau)$ or $F^{a\otimes h}(\tau)$.  To establish these operators in the semidual lattice code model, we begin by  defining the triangle operators and subsequently introduce a recursive relation for determining the ribbon operator.  We refer to the triangle operators as fundamental or elementary ribbons.
The operators will act on the whole Hilbert space $\calH$, but the action is non-trivial only on the edges contained in $\rho$.

Given the left and right module edge operators defined in  \eqref{triangle operators} and \eqref{right triangle operators}, we consider the different cases for triangle operators separately following conventions outlined in \cite{yan2022ribbon} to define the triangle operators as follows:
\begin{equation}
    \label{rdual_rib_plus}
    \begin{minipage}[c][2cm][c]{0.2\textwidth}
        \begin{tikzpicture}[scale=0.7]

    \fill [gray!40] (4*\rcs, 1*\rcs)--(1.5*\rcs, 3.5*\rcs)--(6.5*\rcs, 3.5*\rcs)--(4*\rcs,1*\rcs);
    
    \node (A) at (4.5*\rcs, 5*\rcs) {$\psi$};
    \begin{scope}[thick, decoration={markings, mark=at position 0.95 with {\arrow{latex}}}] 
        \draw[postaction={decorate}] (4*\rcs, 1*\rcs) -- (4*\rcs, 5*\rcs);
    \end{scope}
    
    \draw[thick, dotted] (4*\rcs,1*\rcs) -- (1.5*\rcs, 3.5*\rcs);
    \draw[thick, dotted] (4*\rcs,1*\rcs) -- (6.5*\rcs, 3.5*\rcs);
    \begin{scope}[red, thin, dashed, decoration={markings, mark=at position 0.4 with {\arrow{latex}}}] 
        \draw[postaction={decorate}] (6.5*\rcs, 3.5*\rcs) -- (1.5*\rcs, 3.5*\rcs);
    \end{scope}
    
    \draw[text=black, draw=white, fill=white] (2*\rcs, 2.25*\rcs) circle (5pt) node {\footnotesize $s_{0}$};
    \draw[text=black, draw=white, fill=white] (6*\rcs, 2.25*\rcs) circle (5pt) node {\footnotesize $s_{1}$};

    \draw[double, -latex] (3.5*\rcs, 2.5*\rcs) -- (4.5*\rcs, 2.5*\rcs);
    \harrow{\rcs}{(2.5*\rcs, 2.5*\rcs)}{(5.5*\rcs, 2.5*\rcs)}{0.1}{0}
    
            
            
            
    
\end{tikzpicture}
    \end{minipage}
    \hspace{0.02\textwidth}
    \begin{minipage}[c][2cm][t]{0.6\textwidth}
        $$F^{(a,h)}\left(\tilde{\tau}_{L}\right)\ket{\psi} =
        \epsilon(a)L^{h}_{+}\ket{\psi} 
        =\epsilon(a)\langle h, S\psi_{(1)}\psi_{(3)}\rangle\ket{\psi_{(2)}}
        $$
    \end{minipage}
\end{equation}
\begin{equation}
    \label{rdual_rib_minus}
    \begin{minipage}[c][2cm][c]{0.2\textwidth}
        \begin{tikzpicture}[scale=0.7]

    \fill [gray!20] (4*\rcs, 1*\rcs)--(1.5*\rcs, 3.5*\rcs)--(6.5*\rcs, 3.5*\rcs)--(4*\rcs,1*\rcs);
    
    \node (A) at (4.5*\rcs, 5*\rcs) {$\psi$};
    \begin{scope}[thick, decoration={markings, mark=at position 0.25 with {\arrow{latex}}}] 
        \draw[postaction={decorate}] (4*\rcs, 5*\rcs) -- (4*\rcs, 1*\rcs);
    \end{scope}
    
    \draw[thick, dotted] (4*\rcs,1*\rcs) -- (1.5*\rcs, 3.5*\rcs);
    \draw[thick, dotted] (4*\rcs,1*\rcs) -- (6.5*\rcs, 3.5*\rcs);
    \begin{scope}[red, thin, dashed, decoration={markings, mark=at position 0.4 with {\arrow{latex}}}] 
        \draw[postaction={decorate}] (1.5*\rcs, 3.5*\rcs) -- (6.5*\rcs, 3.5*\rcs);
    \end{scope}
    
    \draw[text=black, draw=white, fill=white] (2*\rcs, 2.25*\rcs) circle (5pt) node {\footnotesize $s_{0}$};
    \draw[text=black, draw=white, fill=white] (6*\rcs, 2.25*\rcs) circle (5pt) node {\footnotesize $s_{1}$};

    \harrow{\rcs}{(2.5*\rcs, 2.5*\rcs)}{(5.5*\rcs, 2.5*\rcs)}{0.1}{0}

            
            
            

\end{tikzpicture}
    \end{minipage}
    \hspace{0.02\textwidth}
    \begin{minipage}[c][2cm][t]{0.6\textwidth}
        $$F^{(a,h)}\left(\tilde{\tau}_{L}\right)\ket{\psi} =
        \epsilon(a)L^{h}_{-}\ket{\psi} =
        \epsilon(a)\langle h, \psi_{(3)}S^{-1}\psi_{(1)}\rangle\ket{\psi_{(2)}}$$
    \end{minipage}
\end{equation}
\begin{equation}
    \label{ldual_rib_plus}
    \begin{minipage}[c][2cm][c]{0.2\textwidth}
        \begin{tikzpicture}[scale=0.7]

    \fill [gray!20] (4*\rcs, 1*\rcs)--(1.5*\rcs, 3.5*\rcs)--(6.5*\rcs, 3.5*\rcs)--(4*\rcs,1*\rcs);
    
    \node (A) at (4.5*\rcs, 5*\rcs) {$\psi$};
    \begin{scope}[thick, decoration={markings, mark=at position 0.95 with {\arrow{latex}}}] 
        \draw[postaction={decorate}] (4*\rcs, 1*\rcs) -- (4*\rcs, 5*\rcs);
    \end{scope}
    
    \draw[thick, dotted] (4*\rcs,1*\rcs) -- (1.5*\rcs, 3.5*\rcs);
    \draw[thick, dotted] (4*\rcs,1*\rcs) -- (6.5*\rcs, 3.5*\rcs);
    \begin{scope}[red, thin, dashed, decoration={markings, mark=at position 0.4 with {\arrow{latex}}}] 
        \draw[postaction={decorate}] (6.5*\rcs, 3.5*\rcs) -- (1.5*\rcs, 3.5*\rcs);
    \end{scope}
    
    \draw[text=black, draw=white, fill=white] (2*\rcs, 2.25*\rcs) circle (5pt) node {\footnotesize $s_{1}$};
    \draw[text=black, draw=white, fill=white] (6*\rcs, 2.25*\rcs) circle (5pt) node {\footnotesize $s_{0}$};

    \harrow{\rcs}{(5.5*\rcs, 2.5*\rcs)}{(2.5*\rcs, 2.5*\rcs)}{0.1}{1}
    
            
            
            
    
\end{tikzpicture}
    \end{minipage}
    \hspace{0.02\textwidth}
    \begin{minipage}[c][2cm][t]{0.6\textwidth}
        $$F^{(a,h)}\left(\tilde{\tau}_{R}\right)\ket{\psi} =
        \epsilon(a)\tilde{L}^{h}_{+}\ket{\psi} =
        \epsilon(a) \langle h, S\psi_{(3)}\psi_{(1)}\rangle \ket{\psi_{(2)}}$$
    \end{minipage}
\end{equation}
\begin{equation}
    \label{ldual_rib_minus}
    \begin{minipage}[c][2cm][c]{0.2\textwidth}
        \begin{tikzpicture}[scale=0.7]

    \fill [gray!20] (4*\rcs, 1*\rcs)--(1.5*\rcs, 3.5*\rcs)--(6.5*\rcs, 3.5*\rcs)--(4*\rcs,1*\rcs);
    
    \node (A) at (4.5*\rcs, 5*\rcs) {$\psi$};
    \begin{scope}[thick, decoration={markings, mark=at position 0.25 with {\arrow{latex}}}] 
        \draw[postaction={decorate}] (4*\rcs, 5*\rcs) -- (4*\rcs, 1*\rcs);
    \end{scope}
    
    \draw[thick, dotted] (4*\rcs,1*\rcs) -- (1.5*\rcs, 3.5*\rcs);
    \draw[thick, dotted] (4*\rcs,1*\rcs) -- (6.5*\rcs, 3.5*\rcs);
    \begin{scope}[red, thin, dashed, decoration={markings, mark=at position 0.4 with {\arrow{latex}}}] 
        \draw[postaction={decorate}] (1.5*\rcs, 3.5*\rcs) -- (6.5*\rcs, 3.5*\rcs);
    \end{scope}

    \draw[text=black, draw=white, fill=white] (2*\rcs, 2.25*\rcs) circle (5pt) node {\footnotesize $s_{1}$};
    \draw[text=black, draw=white, fill=white] (6*\rcs, 2.25*\rcs) circle (5pt) node {\footnotesize $s_{0}$};

    \harrow{\rcs}{(5.5*\rcs, 2.5*\rcs)}{(2.5*\rcs, 2.5*\rcs)}{0.1}{1}
    
            
            
            
    
\end{tikzpicture}
    \end{minipage}
    \hspace{0.02\textwidth}
    \begin{minipage}[c][2cm][t]{0.6\textwidth}
        $$F^{(a,h)}\left(\tilde{\tau}_{R}\right)\ket{\psi} =
        \epsilon(a)\tilde{L}^{h}_{-}\ket{\psi} =
        \epsilon(a)\langle h, \psi_{(1)}S^{-1}\psi_{(3)}\rangle\ket{\psi_{(2)}}$$
    \end{minipage}
\end{equation}
\begin{equation}
    \label{rdir_rib_plus}
    \begin{minipage}[c][2cm][c]{0.2\textwidth}
        \begin{tikzpicture}[scale=0.7]

    \fill [gray!20] (4*\rcs, 1*\rcs)--(1.5*\rcs, 3.5*\rcs)--(6.5*\rcs, 3.5*\rcs)--(4*\rcs,1*\rcs);
    
    \draw[thick, dotted] (4*\rcs,1*\rcs) -- (1.5*\rcs, 3.5*\rcs);
    \draw[thick, dotted] (6.5*\rcs, 3.5*\rcs) -- (4*\rcs,1*\rcs);
    \begin{scope}[thick, decoration={markings,mark=at position 0.5 with {\arrow{latex}}}] 
        \draw[postaction={decorate}] (1.5*\rcs, 3.5*\rcs) to node[midway, above]{$\psi$} (6.5*\rcs, 3.5*\rcs) ;
    \end{scope}

    \draw[text=black, draw=white, fill=white] (2*\rcs, 2.25*\rcs) circle (5pt) node {\footnotesize $s_{1}$};
    \draw[text=black, draw=white, fill=white] (6*\rcs, 2.25*\rcs) circle (5pt) node {\footnotesize $s_{0}$};

    \harrow{\rcs}{(5.5*\rcs, 2.5*\rcs)}{(2.5*\rcs, 2.5*\rcs)}{0.1}{1}
    
            
            
            
    
\end{tikzpicture}
    \end{minipage}
    \hspace{0.02\textwidth}
    \begin{minipage}[c][2cm][c]{0.6\textwidth}
        $$F^{(a,h)}\left(\tau_{R}\right)\ket{\psi} =
        \epsilon(h)T^{a}_{+}\ket{\psi} =
        \epsilon(h)\langle Sa, \psi_{(1)} \rangle \ket{\psi_{(2)}}$$
    \end{minipage}
\end{equation}
\begin{equation}
    \label{rdir_rib_minus}
    \begin{minipage}[c][2cm][c]{0.2\textwidth}
        \begin{tikzpicture}[scale=0.7]

    \fill [gray!20] (4*\rcs, 1*\rcs)--(1.5*\rcs, 3.5*\rcs)--(6.5*\rcs, 3.5*\rcs)--(4*\rcs,1*\rcs);
    
    \draw[thick, dotted] (4*\rcs,1*\rcs) -- (1.5*\rcs, 3.5*\rcs);
    \draw[thick, dotted] (6.5*\rcs, 3.5*\rcs) -- (4*\rcs,1*\rcs);
    \begin{scope}[thick, decoration={markings,mark=at position 0.5 with {\arrow{latex}}}] 
        \draw[postaction={decorate}] (6.5*\rcs, 3.5*\rcs) to node[midway, above]{$\psi$} (1.5*\rcs, 3.5*\rcs);
    \end{scope}

    \draw[text=black, draw=white, fill=white] (2*\rcs, 2.25*\rcs) circle (5pt) node {\footnotesize $s_{1}$};
    \draw[text=black, draw=white, fill=white] (6*\rcs, 2.25*\rcs) circle (5pt) node {\footnotesize $s_{0}$};

    \draw[double,-latex] (4.5*\rcs, 2.5*\rcs) -- (3.5*\rcs, 2.5*\rcs);
    \harrow{\rcs}{(5.5*\rcs, 2.5*\rcs)}{(2.5*\rcs, 2.5*\rcs)}{0.1}{1}

            
            
            
    
\end{tikzpicture}
    \end{minipage}
    
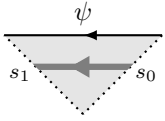
\hspace{0.02\textwidth}
    \begin{minipage}[c][2cm][c]{0.6\textwidth}
        $$F^{(a,h)}\left(\tau_{R}\right)\ket{\psi} =
        \epsilon(h)T^{a}_{-}\ket{\psi} = 
        \epsilon(h)\langle a, \psi_{(2)} \rangle \ket{\psi_{(1)}}$$
    \end{minipage}
\end{equation}
\begin{equation}
    \label{ldir_rib_plus}
    \begin{minipage}[c][2cm][c]{0.2\textwidth}
        \begin{tikzpicture}[scale=0.7]

    \fill [gray!20] (4*\rcs, 1*\rcs)--(1.5*\rcs, 3.5*\rcs)--(6.5*\rcs, 3.5*\rcs)--(4*\rcs,1*\rcs);
    
    \draw[thick, dotted] (4*\rcs,1*\rcs) -- (1.5*\rcs, 3.5*\rcs);
    \draw[thick, dotted] (6.5*\rcs, 3.5*\rcs) -- (4*\rcs,1*\rcs);
    \begin{scope}[thick, decoration={markings,mark=at position 0.5 with {\arrow{latex}}}] 
        \draw[postaction={decorate}] (1.5*\rcs, 3.5*\rcs) to node[midway, above]{$\psi$} (6.5*\rcs, 3.5*\rcs);
    \end{scope}

    \draw[text=black, draw=white, fill=white] (2*\rcs, 2.25*\rcs) circle (5pt) node {\footnotesize $s_{0}$};
    \draw[text=black, draw=white, fill=white] (6*\rcs, 2.25*\rcs) circle (5pt) node {\footnotesize $s_{1}$};

    \harrow{\rcs}{(2.5*\rcs, 2.5*\rcs)}{(5.5*\rcs, 2.5*\rcs)}{0.1}{0}

            
            
            
    
\end{tikzpicture}
    \end{minipage}
    \hspace{0.02\textwidth}
    \begin{minipage}[c][2cm][c]{0.6\textwidth}
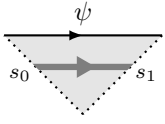
        $$F^{(a,h)}\left(\tau_{L}\right)\ket{\psi} =
        \epsilon(h)\tilde{T}^{a}_{+}\ket{\psi} = 
        \epsilon(h)\langle a, \psi_{(1)} \rangle \ket{\psi_{(2)}}$$
    \end{minipage}
\end{equation}
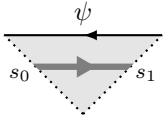
\begin{equation}
    \label{ldir_rib_minus}
    \begin{minipage}[c][2cm][c]{0.2\textwidth}
        \begin{tikzpicture}[scale=0.7]

    \fill [gray!20] (4*\rcs, 1*\rcs)--(1.5*\rcs, 3.5*\rcs)--(6.5*\rcs, 3.5*\rcs)--(4*\rcs,1*\rcs);
    
    \draw[thick, dotted] (4*\rcs,1*\rcs) -- (1.5*\rcs, 3.5*\rcs);
    \draw[thick, dotted] (4*\rcs,1*\rcs) -- (6.5*\rcs, 3.5*\rcs);
    \begin{scope}[thick, decoration={markings,mark=at position 0.5 with {\arrow{latex}}}] 
        \draw[postaction={decorate}] (6.5*\rcs, 3.5*\rcs) to node[midway, above]{$\psi$} (1.5*\rcs, 3.5*\rcs);
    \end{scope}

    \draw[text=black, draw=white, fill=white] (2*\rcs, 2.25*\rcs) circle (5pt) node {\footnotesize $s_{0}$};
    \draw[text=black, draw=white, fill=white] (6*\rcs, 2.25*\rcs) circle (5pt) node {\footnotesize $s_{1}$};

    \harrow{\rcs}{(2.5*\rcs, 2.5*\rcs)}{(5.5*\rcs, 2.5*\rcs)}{0.1}{0}

            
            
            
    
\end{tikzpicture}
    \end{minipage}
    \hspace{0.02\textwidth}
    \begin{minipage}[c][2cm][c]{0.6\textwidth}
        $$F^{(a,h)}\left(\tau_{L}\right)\ket{\psi} =
            \epsilon(h)\tilde{T}^{a}_{-}\ket{\psi} = 
            \epsilon(h)\langle S^{-1}a, \psi_{(2)} \rangle \ket{\psi_{(1)}}$$
    \end{minipage}
\end{equation}
Here,  the ribbons in \eqref{rdual_rib_plus} and \eqref{rdual_rib_minus}  are left dual, those in \eqref{ldual_rib_plus} and \eqref{ldual_rib_minus} are right dual, the ones in \eqref{rdir_rib_plus} and \eqref{rdir_rib_minus} are right direct and those in \eqref{ldir_rib_plus} and \eqref{ldir_rib_minus} are left direct. The convention for the local left and right orientation of the fundamental ribbons is based on \cite{jia2023boundary}.
For the ribbons other than the fundamental ones or triangles, the ribbon operators are define inductively.

\begin{prop} 
The ribbon operators satisfy the following multiplication relations
    \begin{equation}
        F^{(a,h)}(\tau_{R}) \circ F^{(a',h')}(\tau_{R}) = F^{(aa',hh')}(\tau_{R}), \quad  F^{(a,h)}(\tilde{\tau}_{L}) \circ F^{(a',h')}(\tilde{\tau}_{L}) = F^{(aa',hh')}(\tilde{\tau}_{L})
    \end{equation}
    \begin{equation}
        F^{(a,h)}(\tau_{L}) \circ F^{(a',h')}(\tau_{L}) = F^{(a'a,h'h)}(\tau_{L}), \quad F^{(a,h)}(\tilde{\tau}_{R}) \circ F^{(a',h')}(\tilde{\tau}_{R}) = F^{(a'a,h'h)}(\tilde{\tau}_{R})
    \end{equation}
\end{prop}
\begin{proof}
    See appendix \ref{appendix:multiplication_of_ribbon_operators}
\end{proof}
The above proposition means that, as algebras, the triangle operator algebras satisfy $\calB_{\tau_R}\cong \calB_{\tilde{\tau}_L}= H^{\text{cop}}\otimes H\cong H^{\text{cop}}\lrbicross H $ and $\calB_{\tau_L}\cong \calB_{\tilde{\tau}_R}\cong (H^{\text{cop}}\otimes H)^{\text{op}}= H\tens H^{\text{op}} \cong M(H)^{\text{op}}$. 
\\

Next, we define the  general ribbon $\rho$ in the semidual lattice code model. These are defined recursively for the different ribbon types.
To define $F^{(a,h)}(\rho)$, we consider the  decomposition  $\rho = \tau_{1}\cup \tau_{2}$ such that both $\tau_{1}$ and $\tau_{2}$ have the same direction with $\rho$ and the terminal site of $\tau_{1}$ is the initial site of  $\tau_{2}$.
 We begin with the definition of type-B ribbon operators. The ribbon operators for type-B ribbons  $\rho=\rho^B$ are
parameterized by $(a\tens h)\in H^{\text{cop}}\tens H $.  These operators form the ribbon operator algebra
$\calB_{\rho_B}=  \{ F^{(a,h)}(\rho^B) : a\in H^{\text{cop}} , h  \in H \}$.
Since $H^{\text{cop}}\otimes H\cong H^{\text{cop}}\lrbicross H $, 
 the ribbon operator $F^{(a,h)}(\rho)= F^{a\tens h }(\rho^B)$ for the composite ribbon  $\rho^B = \tau_{1}\cup \tau_{2}$   is defined recursively  via the coproduct \eqref{coproduct of M(H)} of $a\tens h\in H^{\text{cop}}\lrbicross H$ as
\begin{equation}
\label{gluing-operators}
    F^{(a,h)}(\rho^B) := \sum_{(a, h)} F^{a_{(2)}, h_{(2)}}(\tau_{1})F^{a_{(1)}h_{(1)}Sh_{(3)}, h_{(4)}}(\tau_{2}).
\end{equation}
It is independent of the choice of the decomposition $\rho =\tau_1 \cup \tau_2$ due to the coassociativity.  Similarly,  the type-A ribbon is  constructed  from the Hopf algebra $H\tens H^{\text{op}}=M(H)^{\text{op}}$.
 The construction uses the recursive relation  (\ref{gluing-operators}) since $M(H)^{\text{op}}$ has the same coproduct as $M(H)$.
For a closed ribbon, there is only one end $\partial \rho =\partial_0\rho  = \partial_1\rho,$  where the starting site of a ribbon $\rho$ is denoted by $\partial_{0}\rho$ and the ending site is denoted by $\partial_{1}\rho$. The vertex and face operators can therefore be considered  as special cases of closed ribbon operators.
\\

The properties of ribbon operators are essential for understanding  topological excitations of the model.
Since the topological excitations are generated at two ends of the ribbon, the commutation relations between the vertex and face operators and ribbon operators play a key role.

\begin{lemma}
Let $\rho=\rho_A$ or $\rho_B$ be an open ribbon with starting and ending sites $s_0,s_1$. Then the ribbon operator $F^{(a,h)}(\rho)$ do not commute  with the local vertex and face operators $A^h$ and $B^a$ operators.

\end{lemma}
\begin{proof}
See Appendix \ref{F commute AB relations} for commutation relations between ribbon operators and vertex and face operators.
\end{proof}

We can also show that ribbon operators commute with all terms in the Hamiltonian except for those associated with the ends of the ribbon.
\begin{prop}
\label{F commute AB}
Let $\rho$ be an open ribbon and $s$ be a site on $\rho$ such that $s$ has no overlap with $\partial_i \rho$. Then
\be
A(s)F^{(a,h)}\left(\rho\right)=F^{(a,h)}\left(\rho\right)A(s), \quad B(s)F^{(a,h)}\left(\rho\right)=F^{(a,h)}\left(\rho \right)B(s)
\ee
where  $A(s) = A^{l}(v,p)$ and $B(s) = B^{k}(v,p)$ are the  terms associated to $s$ in the Hamiltonian.
\end{prop}

\begin{proof}
See Appendix \ref{appendix:proof_cummutation_between_ribbon_and_geometric_operators}.
\end{proof}

The results in Proposition \ref{F commute AB} above follow for closed ribbons as well.  Consequently,  on a sufficiently long ribbon $\rho$, the ribbon operators  commute with all terms in the Hamiltonian except those associated with the ends of $\rho$, creating excitations only at the end of the ribbon. Since the action of the local operators on $\partial_i \rho$ preserves $M(H)$,  the space of elementary excitations for the semidual Kitaev lattice model is therefore given by $M(H)$.
When acting on the ground states,
the ribbon operators $F^{(a,h)}(\rho)$ define a representation of $M(H)$ and $M(H)^{\text{op}}$ depending on the local orientation.
\section{Conclusion}
\label{sec:conclusion}

In this work, we have systematically defined the ribbon operators in the semidual Kitaev lattice model  and derived some important properties. The ribbon  operators are essential for understanding quasi-particle excitations within topologically ordered systems. We established that the 
  space of ribbon operators  in the semidual Kitaev model naturally identifies with the bicrossproduct quantum groups $M(H)=H^{\text{cop}}\lrbicross H$ and $M(H)^{\text{op}}$ depending on the chirality
and corresponds to the two types of ribbons described as the type-$B$ (containing right-handed direct and left-handed dual triangles) and type-$A$ (containing left-handed direct and right-handed dual triangles) ribbons respectively. 
\\

It will be good to understand the physical relation between the generalized Kitaev quantum double model and the semidual model. Both are based on quantum groups related by a Drinfeld and module-algebra twist \cite{MO}. We expect the twist to play a role in the algebraic relation between these models. It also remains to understand the physical meaning of semidualisation in the context of these lattice code models.
 Another  interesting direction is  to generalize the semidual formulation to include  Hamiltonian and ribbon operators  with gapped boundaries. We expect that these generalized lattice code models based on Hopf algebras may find applications in topological quantum computing.

\bibliographystyle{JHEP}
\bibliography{biblio-kitaev}

\newpage
\appendix
\renewcommand{\thesection}{\Alph{section}}

\section{Multiplication of the Ribbon Operators}
\label{appendix:multiplication_of_ribbon_operators}

\begin{eqnarray}
    L^{h}_{\pm}L^{h'}_{\pm} &=& L^{hh'}_{\pm},\qquad \tilde{L}^{h}_{\pm}\tilde{L}^{h'}_{\pm} = \tilde{L}^{h'h}_{\pm}\\
    T^{a}_{\pm}T^{a'}_{\pm} &=& T^{aa'}_{\pm},\qquad \tilde{T}^{a}_{\pm}\tilde{T}^{a'}_{\pm} = \tilde{T}^{a'a}_{\pm}
\end{eqnarray}

Multiplication of the ribbon operator in (\ref{rdual_rib_plus})
\begin{equation}
    \begin{split}
        \ldur{a}{h}\ldur{a'}{h'}\ket{\psi}
        & = \eps(a')\< h', S\cop{\psi}{1}\cop{\psi}{3}\> \ldur{a}{h}\ket{\cop{\psi}{2}},\\
        & = \eps(a')\epsilon(a) \< h', S\cop{\psi}{1}\cop{\psi}{3}\> \< h, S\copp{\psi}{2}{1}\copp{\psi}{2}{3}\> \ket{\copp{\psi}{2}{2}},\\
        & = \eps(a'a)\< hh', S\cop{\psi}{1}\cop{\psi}{3} \> \ket{\cop{\psi}{2}},\\
        & = \ldur{aa'}{hh'}\ket{\psi}.
    \end{split}
\end{equation}

Multiplication of the ribbon operator in (\ref{rdual_rib_minus})
\begin{equation}
    \begin{split}
        \ldur{a}{h}\ldur{\pr{a}}{\pr{h}}\ket{\psi}
        & = \eps(a') \< h', \cop{\psi}{3}S^{-1}\cop{\psi}{1} \> \ldur{a}{b}\ket{\cop{\psi}{2}},\\
        & = \eps(a')\eps(a) \< h', \cop{\psi}{3}S^{-1}\cop{\psi}{1} \> \< h, \copp{\psi}{2}{3}S^{-1}\copp{\psi}{2}{1} \> \ket{\copp{\psi}{2}{2}},\\
        & = \eps(aa')\< h'h, \cop{\psi}{3}S^{-1}\cop{\psi}{1} \> \ket{\psi_{(2)}},\\
        & = \ldur{aa'}{hh'}\ket{\psi}.
    \end{split}
\end{equation}

Multiplication of the ribbon operator in (\ref{ldual_rib_plus})
\begin{equation}
    \begin{split}
        \rdur{a}{h}\rdur{a'}{h'}\ket{\psi}
        & = \eps(a')\< h', S\cop{\psi}{3}\cop{\psi}{1} \> \rdur{a}{h} \ket{\cop{\psi}{2}},\\
        & = \eps(a')\eps(a) \< h', S\cop{\psi}{3}\cop{\psi}{1} \> \< h, S\copp{\psi}{2}{3}\copp{\psi}{2}{1} \> \ket{\copp{\psi}{2}{2}},\\
        & = \eps(a'a) \< h'h, S\cop{\psi}{3}\cop{\psi}{1} \> \ket{\cop{\psi}{2}},\\
        & = \rdur{a'a}{h'h}\ket{\psi}.
    \end{split}
\end{equation}

Multiplication of the ribbon operator in (\ref{ldual_rib_minus})
\begin{equation}
    \begin{split}
        \rdur{a}{h}\rdur{a'}{h'}\ket{\psi}
        & = \eps(a')\< h', \cop{\psi}{1}S^{-1}\cop{\psi}{3} \> \rdur{a}{h} \ket{\psi_{(2)}},\\
        & = \eps(a')\eps(a) \< h', \cop{\psi}{1}S^{-1}\cop{\psi}{3} \> \< h, \copp{\psi}{2}{1}S^{-1}\copp{\psi}{2}{3} \> \ket{\copp{\psi}{2}{2}},\\
        & = \epsilon(a'a) \< h'h, \cop{\psi}{1}S^{-1}\cop{\psi}{3} \> \ket{\cop{\psi}{2}},\\
        & = \rdur{a'a}{h'h}\ket{\psi}.
    \end{split}
\end{equation}

Multiplication of the ribbon operator in (\ref{rdir_rib_plus})
\begin{equation}
    \begin{split}
        \rdir{a}{h}\rdir{a'}{h'}\ket{\psi}
        & = \eps(h') \< Sa', \cop{\psi}{1} \> \rdir{a}{h}\ket{\cop{\psi}{2}},\\
        & = \eps(h')\eps(h) \< Sa', \cop{\psi}{1} \> \< Sa, \copp{\psi}{2}{1} \> \ket{\copp{\psi}{2}{2}},\\
        & = \eps(hh') \< S(aa'), \cop{\psi}{1} \> \ket{\cop{\psi}{1}},\\
        & = \rdir{aa'}{hh'}\ket{\psi}.
    \end{split}
\end{equation}

Multiplication of the ribbon operator in (\ref{rdir_rib_minus})
\begin{equation}
    \begin{split}
        \rdir{a}{h}\rdir{a'}{h'}\ket{\psi}
        & = \eps(h')\< a', \cop{\psi}{2} \> \rdir{a}{h}\ket{\cop{\psi}{1}},\\
        & = \eps(h')\eps(h) \< a', \cop{\psi}{2} \> \< a, \copp{\psi}{1}{2} \> \ket{\copp{\psi}{1}{1}},\\
        & = \eps(hh')\< aa', \cop{\psi}{2} \> \ket{\cop{\psi}{2}},\\
        & = \rdir{aa'}{hh'}\left(\tau_{R}\right)\ket{\psi}.
    \end{split}
\end{equation}

Multiplication of the ribbon operator in (\ref{ldir_rib_plus})
\begin{equation}
     \begin{split}
        \ldir{a}{h}\ldir{a'}{h'}\ket{\psi}
        & = \eps(h')\< a', \cop{\psi}{1} \> \ldir{a}{h}\ket{\cop{\psi}{2}},\\
        & = \eps(h')\eps(h) \<a', \cop{\psi}{1} \> \< a, \copp{\psi}{2}{1} \> \ket{\copp{\psi}{2}{2}},\\
        & = \eps(h'h) \< a'a, \cop{\psi}{1} \> \ket{\cop{\psi}{2}},\\
        & = \ldir{a'a}{h'h}\ket{\psi}.
    \end{split}
\end{equation}

Multiplication of the ribbon operator in (\ref{ldir_rib_minus})
\begin{equation}
\begin{split}
        \ldir{a}{h}\ldir{a'}{h'}\ket{\psi}
        &=\eps(h') \< \Si a', \cop{\psi}{2} \> \ldir{a}{h}\ket{\cop{\psi}{1}},\\
        &=\eps(h')\eps(h) \< \Si a', \cop{\psi}{2} \>  \< \Si a, \copp{\psi}{1}{2} \> \ket{\copp{\psi}{1}{1}},\\
        &=\eps(h'h) \< \Si (a'a), \cop{\psi}{2}\> \ket{\cop{\psi}{1}},\\
        &=\ldir{a'a}{h'h}\ket{\psi}.
    \end{split} 
\end{equation}

\section{Relationship between the Ribbon and Geometric Operators}
\label{F commute AB relations}

In this section, we compute the commutation relationship between the geometric operators and the ribbon operators at the beginning and end sites of the types of ribbons (type-$A$ and type-$B$). We will consider the ribbon lengths that are sufficient to obtain these relations, this will help us to compute the relationship between the ribbon operators and the geometric operators at the sites within longer ribbons. If $\hat{G}$ is a geometric operator and $\hat{R}$ is a ribbon operator, the idea is to compute $\hat{G}\circ\hat{R}$ and $\hat{R}\circ\hat{G}$ and compare the results in other to establish the relationship.\\

We begin with the type-$B$ ribbons, the ribbon in \eqref{sigma_1:relation} is a type-$B$ ribbon made up of a single left dual triangle $\tau$, so we have $\rho^{B}=\tau$. This will enable us to compute the relationship between the face operators and the type-$B$ ribbon operators at the starting and end sites in of type-$B$ ribbons.

\begin{equation}
    \begin{minipage}[c]{0.05\textwidth}
        $\ket{\Sigma_{1}}$
    \end{minipage}
    \begin{minipage}[c]{0.05\textwidth}
        $=$
    \end{minipage}
    \begin{minipage}[c]{0.5\textwidth}
        \begin{tikzpicture}[scale = 0.8]

    \begin{scope}[thick, decoration={markings, mark=at position 0.7 with {\arrow{latex}}}]
        \draw[postaction={decorate}] (5*\rcs,5*\rcs) -- (5*\rcs,9*\rcs);
        \draw[postaction={decorate}] (1*\rcs,9*\rcs) -- (5*\rcs,9*\rcs);
        \draw[postaction={decorate}] (1*\rcs,5*\rcs) -- (1*\rcs,9*\rcs);
        \draw[postaction={decorate}] (5*\rcs,5*\rcs) -- (1*\rcs,5*\rcs);
        \draw[postaction={decorate}] (5*\rcs,9*\rcs) -- (9*\rcs,9*\rcs);
        \draw[postaction={decorate}] (9*\rcs,9*\rcs) -- (9*\rcs,5*\rcs);
        \draw[postaction={decorate}] (9*\rcs,5*\rcs) -- (5*\rcs,5*\rcs);
    \end{scope}

    \begin{scope}[thin, red, dashed, decoration={markings, mark=at position 0.3 with {\arrow{latex}}}]
        \draw[postaction={decorate}] (7*\rcs,7*\rcs) -- (3*\rcs, 7*\rcs);
    \end{scope}

    \draw[dotted, very thick] (5*\rcs, 5*\rcs) -- (3*\rcs, 7*\rcs);
    \draw[dotted, thick] (5*\rcs, 5*\rcs) -- (7*\rcs, 7*\rcs);

    \harrow{\rcs}{(4*\rcs, 6*\rcs)}{(6*\rcs, 6*\rcs)}{0.07}{0}
    
    \draw[text=black, draw=white, fill=white] (5.5*\rcs, 8.0*\rcs) circle (5pt) node {\footnotesize $\psi^{1}$};
    \draw[text=black, draw=white, fill=white] (2.5*\rcs, 9.5*\rcs) circle (5pt) node {\footnotesize $\psi^{2}$};
    \draw[text=black, draw=white, fill=white] (0.5*\rcs, 7*\rcs) circle (5pt) node {\footnotesize $\psi^{3}$};
    \draw[text=black, draw=white, fill=white] (3.0*\rcs, 4.5*\rcs) circle (5pt) node {\footnotesize $\psi^{4}$};
    \draw[text=black, draw=white, fill=white] (6.5*\rcs, 4.5*\rcs) circle (5pt) node {\footnotesize $\psi^{5}$};
    \draw[text=black, draw=white, fill=white] (9.5*\rcs, 7*\rcs) circle (5pt) node {\footnotesize $\psi^{6}$};
    \draw[text=black, draw=white, fill=white] (7*\rcs, 9.5*\rcs) circle (5pt) node {\footnotesize $\psi^{7}$};

    \draw[text=black, draw=white, fill=white] (4.0*\rcs, 6.0*\rcs) circle (5pt) node {\footnotesize $s_{0}$};
    \draw[text=black, draw=white, fill=white] (6.0*\rcs, 6.0*\rcs) circle (5pt) node {\footnotesize $s_{1}$};
    \draw[text=black, draw=white, fill=white] (5.0*\rcs, 6.5*\rcs) circle (5pt) node {\footnotesize $\tau$};

        



\end{tikzpicture}
    \end{minipage}
    \label{sigma_1:relation}
\end{equation}

At $s_{0}$, $B^{b}\(s_{0}\)F^{a\tens h}\(\rho^{B}\)$ gives
\begin{align}
    B^{b}\(s_{0}\)F^{a\tens h}\(\rho^{B}\)\ket{\Sigma_{1}} 
    &= B^{b}\(s_{0}\)\left[ \rop{a}{h}\(\psi^1\) \tens \psi^2 \tens \psi^3 \tens \psi^4 \right] \nonumber\\
    &= B^{b}\(s_{0}\)\left[ \eps(a)\<h, S\cop{\psi^1}{1}\cop{\psi^1}{3}\>\(\cop{\psi^1}{2}\) \tens \psi^2 \tens \psi^3 \tens \psi^4 \right] \nonumber\\
    &= \eps(a)\<h, S\cop{\psi^1}{1}\cop{\psi^1}{3}\>\TM{\cop{b}{4}}\(\cop{\psi^1}{2}\) \tens \TP{\cop{b}{3}}\(\psi^2\) \tens \TP{\cop{b}{2}}\(\psi^3\) \tens \TP{\cop{b}{1}}\(\psi^4\) \nonumber\\
    &= \eps(a)\<h, S\cop{\psi^1}{1}\cop{\psi^1}{3}\> \<\cop{b}{4}, \copp{\psi^1}{2}{2}\>\copp{\psi^1}{2}{1} \tens \<S\cop{b}{3}, \cop{\psi^2}{1}\>\cop{\psi^2}{2} \nonumber \\ 
    &\qquad \tens \<S\cop{b}{2}, \cop{\psi^3}{1}\>\cop{\psi^3}{2} \tens \<S\cop{b}{1}, \cop{\psi^4}{1}\>\cop{\psi^4}{2} \nonumber\\
    & = \eps(a)\<h, S\cop{\psi^1}{1}\cop{\psi^1}{4}\> \<\cop{b}{4}, \cop{\psi^1}{3}\>\cop{\psi^1}{2} \tens \<S\cop{b}{3}, \cop{\psi^2}{1}\>\cop{\psi^2}{2} \nonumber\\
    &\qquad \tens \<S\cop{b}{2}, \cop{\psi^3}{1}\>\cop{\psi^3}{2} \tens \<S\cop{b}{1}, \cop{\psi^4}{1}\>\cop{\psi^4}{2}
\end{align}
and $F^{a\tens h}\(\rho^{B}\)B^{b}\(s_{0}\)$ gives
\begin{align}
    F^{a\tens h}\(\rho^{B}\)B^{b}\(s_{0}\)\ket{\Sigma_{1}}
    &= F^{a\tens h}\(\rho^{B}\)\left[\TM{\cop{b}{4}}\(\psi^1\) \tens \TP{\cop{b}{3}}\(\psi^2\) \tens \TP{\cop{b}{2}}\(\psi^3\) \tens \TP{\cop{b}{1}}\(\psi^4\)\right] \nonumber \\
    &= F^{a\tens h}\(\rho^{B}\)\left[\<\cop{b}{4}, \cop{\psi^1}{2}\>\cop{\psi^1}{1} \tens \<S\cop{b}{3}, \cop{\psi^2}{1}\>\cop{\psi^2}{2} \tens \<S\cop{b}{2}, \cop{\psi^3}{1}\>\cop{\psi^3}{2} \tens \<S\cop{b}{1}, \cop{\psi^4}{1}\>\cop{\psi^4}{2}\right]\nonumber \\
    &=\<\cop{b}{4}, \cop{\psi^1}{2}\>\rop{a}{h}\(\cop{\psi^1}{1}\) \tens \<S\cop{b}{3}, \cop{\psi^2}{1}\>\cop{\psi^2}{2} \tens \<S\cop{b}{2}, \cop{\psi^3}{1}\>\cop{\psi^3}{2} \tens \<S\cop{b}{1}, \cop{\psi^4}{1}\>\cop{\psi^4}{2} \nonumber \\
    &=\<\cop{b}{4}, \cop{\psi^1}{2}\> \eps(a)\<h, S\copp{\psi^1}{1}{1}\copp{\psi^1}{1}{3}\>\copp{\psi^1}{1}{2} \tens \<S\cop{b}{3}, \cop{\psi^2}{1}\>\cop{\psi^2}{2} \nonumber \\
    &\qquad \tens \<S\cop{b}{2}, \cop{\psi^3}{1}\>\cop{\psi^3}{2} \tens \<S\cop{b}{1}, \cop{\psi^4}{1}\>\cop{\psi^4}{2} \nonumber \\
    &= \eps(a)\< h, S\cop{\psi^1}{1}\cop{\psi^1}{3}\>\<\cop{b}{4}, \cop{\psi^1}{4} \> \cop{\psi^1}{2} \tens \<S\cop{b}{3}, \cop{\psi^2}{1}\> \cop{\psi^2}{2} \nonumber\\
    & \qquad \tens \<S\cop{b}{2}, \cop{\psi^3}{1}\> \cop{\psi^3}{2} \tens \<S\cop{b}{1}, \cop{\psi^4}{1}\> \cop{\psi^4}{2}
\end{align}
At the site $s_{1}$, $B^{b}\(s_{1}\)F^{a\tens h}\(\rho^{A}\)$, gives
\begin{align}
    B^{b}\(s_{1}\)F^{a\tens h}\(\rho^{B}\)\ket{\Sigma_{1}} 
    &= B^{b}\(s_{0}\)\left[ \psi^{5} \tens \psi^{6} \tens \psi^{7} \tens \rop{a}{h}\(\psi^{1}\) \right] \nonumber \\
    &= B^{b}\(s_{0}\)\left[ \psi^{5} \tens \psi^{6} \tens \psi^{7} \tens \eps(a)\<h, S\cop{\psi^1}{1}\cop{\psi^1}{3} \>\cop{\psi^1}{2} \right] \nonumber \\
    &= \TP{\cop{b}{4}}\(\psi^{5}\) \tens \TP{\cop{b}{3}}\(\psi^{6}\) \tens \TP{\cop{b}{2}}\(\psi^{7}\) \tens \eps(a)\<h, S\cop{\psi^1}{1}\cop{\psi^1}{3} \>\TP{\cop{b}{1}}\(\cop{\psi^1}{2}\) \nonumber \\
    &= \<S\cop{b}{4}, \cop{\psi^5}{1}\>\cop{\psi^5}{2} \tens \<S\cop{b}{3}, \cop{\psi^6}{1}\>\cop{\psi^6}{2} \tens \<S\cop{b}{2}, \cop{\psi^7}{1}\>\cop{\psi^7}{2} \nonumber \\
    &\qquad \tens \eps(a)\<h, S\cop{\psi^1}{1}\cop{\psi^1}{3} \>\<S\cop{b}{1}, \copp{\psi^1}{2}{1}\>\copp{\psi^1}{2}{2} \nonumber \\
    &= \<S\cop{b}{4}, \cop{\psi^5}{1}\>\cop{\psi^5}{2} \tens \<S\cop{b}{3}, \cop{\psi^6}{1}\>\cop{\psi^6}{2} \tens \<S\cop{b}{2}, \cop{\psi^7}{1}\>\cop{\psi^7}{2} \nonumber \\
    &\qquad \tens \eps(a)\<h, S\cop{\psi^1}{1}\cop{\psi^1}{4} \>\<S\cop{b}{1}, \cop{\psi^1}{2}\>\cop{\psi^1}{3}
\end{align}
and $F^{a\tens h}\(\rho^{B}\)B^{b}\(s_{1}\)$
\begin{align}
    F^{a\tens h}\(\rho^{B}\)B^{b}\(s_{1}\)\ket{\Sigma_1}
    &= F^{a\tens h}\(\rho^{B}\)\left[\TP{\cop{b}{4}}\(\psi^5\) \tens \TP{\cop{b}{3}}\(\psi^6\) \tens \TP{\cop{b}{2}}\(\psi^7\) \tens \TP{\cop{b}{1}}\(\psi^1\)\right] \nonumber \\
    &= F^{a\tens h}\(\rho^{B}\)\left[\<S\cop{b}{4}, \cop{\psi^5}{1}\>\cop{\psi^5}{2} \tens \<S\cop{b}{3}, \cop{\psi^6}{1}\>\cop{\psi^6}{2} \tens \<S\cop{b}{2}, \cop{\psi^7}{1}\>\cop{\psi^7}{2} \tens \<S\cop{b}{1}, \cop{\psi^1}{1}\>\cop{\psi^1}{2} \right] \nonumber \\
    &= \<S\cop{b}{4}, \cop{\psi^5}{1}\>\cop{\psi^5}{2} \tens \<S\cop{b}{3}, \cop{\psi^6}{1}\>\cop{\psi^6}{2} \tens \<S\cop{b}{2}, \cop{\psi^7}{1}\>\cop{\psi^7}{2} \tens \<S\cop{b}{1}, \cop{\psi^1}{1}\>\rop{a}{h}\(\cop{\psi^1}{2}\) \nonumber \\
    &= \<S\cop{b}{4}, \cop{\psi^5}{1}\>\cop{\psi^5}{2} \tens \<S\cop{b}{3}, \cop{\psi^6}{1}\>\cop{\psi^6}{2} \tens \<S\cop{b}{2}, \cop{\psi^7}{1}\>\cop{\psi^7}{2} \nonumber \\
    &\qquad \tens \<S\cop{b}{1}, \cop{\psi^1}{1}\>\eps(a)\<h, S\copp{\psi^1}{2}{1}\copp{\psi^1}{2}{3}\>\copp{\psi^1}{2}{2} \nonumber \\
    &= \<S\cop{b}{4}, \cop{\psi^5}{1}\>\cop{\psi^5}{2} \tens \<S\cop{b}{3}, \cop{\psi^6}{1}\>\cop{\psi^6}{2} \tens \<S\cop{b}{2}, \cop{\psi^7}{1}\>\cop{\psi^7}{2} \nonumber \\
    &\qquad \tens \eps(a)\<S\cop{b}{1}, \cop{\psi^1}{1}\>\<h, S\cop{\psi^1}{2}\cop{\psi^1}{4}\>\cop{\psi^1}{3} 
\end{align}

We consider another type-$B$ in \eqref{sigma_2:relation}, it is made up of a left dual triangle ($\tau_1$) and a right direct triangle ($\tau_2$) so we have $\rho^{B} = \tau_{1} \cup \tau_{2}$. This will help us to compute the commutation relations between the vertex operator and the type-$B$ ribbon operator at the starting site of a type-$B$ ribbon.
\begin{equation}
    \begin{minipage}[c]{0.05\textwidth}
        $\ket{\Sigma_{2}}$
    \end{minipage}
    \begin{minipage}[c]{0.05\textwidth}
        $=$
    \end{minipage}
    \begin{minipage}[c]{0.5\textwidth}
        \begin{tikzpicture}[scale = 0.8]

    \begin{scope}[thick, decoration={markings, mark=at position 0.7 with {\arrow{latex}}}]
        \draw[postaction={decorate}] (5*\rcs,5*\rcs) -- (5*\rcs,9*\rcs);
        \draw[postaction={decorate}] (5*\rcs,5*\rcs) -- (1*\rcs,5*\rcs);
        \draw[postaction={decorate}] (9*\rcs,5*\rcs) -- (5*\rcs,5*\rcs);
        \draw[postaction={decorate}] (5*\rcs,5*\rcs) -- (5*\rcs,1*\rcs);
    \end{scope}

    \begin{scope}[thin, red, dashed, decoration={markings, mark=at position 0.3 with {\arrow{latex}}}]
        \draw[postaction={decorate}] (7*\rcs,7*\rcs) -- (3*\rcs, 7*\rcs);
    \end{scope}

    \draw[dotted, very thick] (5*\rcs, 5*\rcs) -- (3*\rcs, 7*\rcs);
    \draw[dotted, thick] (5*\rcs, 5*\rcs) -- (7*\rcs, 7*\rcs);
    \draw[dotted, thick] (7*\rcs, 7*\rcs) -- (9*\rcs, 5*\rcs);

    \harrow{\rcs}{(4*\rcs, 6*\rcs)}{(6*\rcs, 6*\rcs)}{0.07}{0}
    \harrow{\rcs}{(6*\rcs, 6*\rcs)}{(8*\rcs, 6*\rcs)}{0.07}{1}
    
    \draw[text=black, draw=white, fill=white] (5.5*\rcs, 8.0*\rcs) circle (5pt) node {\footnotesize $\psi^{1}$};
    \draw[text=black, draw=white, fill=white] (6.5*\rcs, 4.5*\rcs) circle (5pt) node {\footnotesize $\psi^{2}$};
    \draw[text=black, draw=white, fill=white] (5.5*\rcs, 3*\rcs) circle (5pt) node {\footnotesize $\psi^{3}$};
    \draw[text=black, draw=white, fill=white] (3.0*\rcs, 4.5*\rcs) circle (5pt) node {\footnotesize $\psi^{4}$};

    \draw[text=black, draw=white, fill=white] (4.0*\rcs, 6.0*\rcs) circle (5pt) node {\footnotesize $s_{0}$};
    \draw[text=black, draw=white, fill=white] (5.0*\rcs, 6.5*\rcs) circle (5pt) node {\footnotesize $\tau_{1}$};
    \draw[text=black, draw=white, fill=white] (7.0*\rcs, 5.5*\rcs) circle (5pt) node {\footnotesize $\tau_{2}$};

        



\end{tikzpicture}
    \end{minipage}
    \label{sigma_2:relation}
\end{equation}
At $s_{0}$, and $A^{g}\(s_{0}\)F^{a\tens h}\(\rho^{B}\)$ gives
\begin{align}
    A^{g}\(s_{0}\)F^{a\tens h}\(\rho^{B}\)\ket{\Sigma_{2}}
    &= A^{g}\(s_{0}\)\left[\rop[1]{\cop{a}{2}}{\cop{h}{2}}\(\psi^1\) \tens \rop[2]{\cop{a}{1}\cop{h}{1}S\cop{h}{3}}{\cop{h}{4}}\(\psi^2\) \tens \psi^3 \tens \psi^4 \right] \nonumber \\
    &= A^{g}\(s_{0}\)\left[\eps({\cop{a}{2}})\<\cop{h}{2}, S\cop{\psi^1}{1}\cop{\psi^1}{3}\>\cop{\psi^1}{2} \tens \eps(\cop{h}{4})\<S\(\cop{a}{1}\cop{h}{1}S\cop{h}{3}\), \cop{\psi^2}{1}\>\cop{\psi^2}{2} \otimes \psi^{3} \otimes \psi^{4}\right] \nonumber \\
    &= \eps({\cop{a}{2}})\<\cop{h}{2}, S\cop{\psi^1}{1}\cop{\psi^1}{3}\>\PM{1}{\cop{g}{4}}\(\cop{\psi^1}{2}\) \nonumber \\
    &\qquad \tens \eps(\cop{h}{4})\<S\(\cop{a}{1}\cop{h}{1}S\cop{h}{3}\), \cop{\psi^2}{1}\>\PP{\cop{g}{3}S\cop{g}{5}}{\cop{g}{6}}\(\cop{\psi^2}{2}\) \nonumber \\
    &\qquad \tens \PM{\cop{g}{2}S\cop{g}{7}}{\cop{g}{8}}\(\psi^{3}\) \tens \PM{\cop{g}{1}S\cop{g}{9}}{\cop{g}{10}}\(\psi^{4}\) \nonumber \\
    &= \eps({\cop{a}{2}})\<\cop{h}{2}, S\cop{\psi^1}{1}\cop{\psi^1}{3}\>\<\copp{g}{4}{1}, \copp{\psi^1}{2}{3}\>\<\copp{g}{4}{2}, S^{-1}\copp{\psi^1}{2}{1}\> \copp{\psi^1}{2}{2} \nonumber \\
    &\qquad \tens \eps(\cop{h}{4})\<S\(\cop{a}{1}\cop{h}{1}S\cop{h}{3}\), \cop{\psi^2}{1}\>\<S\copp{g}{6}{1}S\(\cop{g}{3}S\cop{g}{5}\), \copp{\psi^2}{2}{1}\>\<\copp{g}{6}{2}, \copp{\psi^2}{2}{3} \>\copp{\psi^2}{2}{2}  \nonumber \\
    &\qquad \tens \<\(\cop{g}{2}S\cop{g}{7}\)\copp{g}{8}{1}, \cop{\psi^3}{3}\>\<\copp{g}{8}{2}, S^{-1}\cop{\psi^3}{1}\>\cop{\psi^3}{2} \nonumber \\ 
    &\qquad \tens \<\(\cop{g}{1}S\cop{g}{9}\)\copp{g}{10}{1}, \cop{\psi^4}{3}\>\<\copp{g}{10}{2}, S^{-1}\cop{\psi^4}{1}\>\cop{\psi^4}{2} \nonumber \\
    &= \eps({\cop{a}{2}})\<\cop{h}{2}, S\cop{\psi^1}{1}\cop{\psi^1}{5}\>\<\cop{g}{4}, \cop{\psi^1}{4}\>\<\cop{g}{5}, S^{-1}\cop{\psi^1}{2}\> \cop{\psi^1}{3} \nonumber \\
    &\qquad \tens \eps(\cop{h}{4})\<S\(\cop{a}{1}\cop{h}{1}S\cop{h}{3}\), \cop{\psi^2}{1}\>\<S\cop{g}{3}, \cop{\psi^2}{2}\>\<\cop{g}{6}, \cop{\psi^2}{4}\>\cop{\psi^2}{3}  \nonumber \\
    &\qquad \tens \<\cop{g}{2}, \cop{\psi^3}{3}\>\<\cop{g}{7}, S^{-1}\cop{\psi^3}{1}\>\cop{\psi^3}{2} \tens \<\cop{g}{1}, \cop{\psi^4}{3}\>\<\cop{g}{8}, S^{-1}\cop{\psi^4}{1}\>\cop{\psi^4}{2} 
\end{align}
and $F^{a\tens h}\(\rho^{B}\)A^{g}\(s_{0}\)$ gives
\begin{align}
    F^{a\tens h}\(\rho^{B}\)A^{g}\(s_{0}\)\ket{\Sigma_{2}}
    &= F^{a\tens h}\(\rho^{B}\)\left[\PM{1}{\cop{g}{4}}\(\psi^1\) \tens \PP{\cop{g}{3}S\cop{g}{5}}{\cop{g}{6}}\(\psi^2\) \tens \PM{\cop{g}{2}S\cop{g}{7}}{\cop{g}{8}}\(\psi^3\) \right. \nonumber \\
    &\qquad \left. \tens \PM{\cop{g}{1}S\cop{g}{9}}{\cop{g}{10}}\(\psi^2\) \right] \nonumber \\
    &= F^{a\tens h}\(\rho^{B}\)\left[\<\copp{g}{4}{1}, \cop{\psi^1}{3}\>\<\copp{g}{4}{2}, S^{-1}\cop{\psi^1}{1}\>\cop{\psi^1}{2} \right. \nonumber \\
    &\qquad \left. \tens \<S\copp{g}{6}{1}S\(\cop{g}{3}S\cop{g}{5}\), \cop{\psi^2}{1}\>\<\copp{g}{6}{2}, \cop{\psi^2}{3}\>\cop{\psi^2}{2} \right. \nonumber\\
    &\qquad \left. \tens \<\(\cop{g}{2}S\cop{g}{7}\)\copp{g}{8}{1}, \cop{\psi^3}{3} \>\<\copp{g}{8}{2}, S^{-1}\cop{\psi^3}{1}\>\cop{\psi^3}{2} \right. \nonumber \\
    &\qquad \left. \tens \<\(\cop{g}{1}S\cop{g}{9}\)\copp{g}{10}{1}, \cop{\psi^4}{3}\>\<\copp{g}{10}{2}, S^{-1}\cop{\psi^4}{1}\>\cop{\psi^4}{2} \right] \nonumber \\
    &= \<\cop{g}{4}, \cop{\psi^1}{3}\>\<\cop{g}{5}, S^{-1}\cop{\psi^1}{1}\>\rop[1]{\cop{a}{2}}{\cop{h}{2}}\(\cop{\psi^1}{2}\) \nonumber \\
    &\qquad \tens \<S\(\cop{g}{3}\), \cop{\psi^2}{1}\>\<\cop{g}{6}, \cop{\psi^2}{3}\>\rop[2]{\cop{a}{1}\cop{h}{1}S\cop{h}{3}}{\cop{h}{4}}\(\cop{\psi^2}{2}\) \nonumber\\
    &\qquad \tens \<\cop{g}{2}, \cop{\psi^3}{3} \>\<\cop{g}{7}, S^{-1}\cop{\psi^3}{1}\>\cop{\psi^3}{2} \tens \<\cop{g}{1}, \cop{\psi^4}{3}\>\<\cop{g}{8}, S^{-1}\cop{\psi^4}{1}\>\cop{\psi^4}{2} \nonumber \\
    &= \<\cop{g}{4}, \cop{\psi^1}{3}\>\<\cop{g}{5}, S^{-1}\cop{\psi^1}{1}\>\eps\(\cop{a}{2}\)\<\cop{h}{2}, S\copp{\psi^1}{2}{1}\copp{\psi^1}{2}{3} \>\copp{\psi^1}{2}{2} \nonumber \\
    &\qquad \tens \<S\(\cop{g}{3}\), \cop{\psi^2}{1}\>\<\cop{g}{6}, \cop{\psi^2}{3}\>\eps\(\cop{h}{4}\)\<S\(\cop{a}{1}\cop{h}{1}S\cop{h}{3}\), \copp{\psi^2}{2}{1}\>\copp{\psi^2}{2}{2} \nonumber\\
    &\qquad \tens \<\cop{g}{2}, \cop{\psi^3}{3} \>\<\cop{g}{7}, S^{-1}\cop{\psi^3}{1}\>\cop{\psi^3}{2} \tens \<\cop{g}{1}, \cop{\psi^4}{3}\>\<\cop{g}{8}, S^{-1}\cop{\psi^4}{1}\>\cop{\psi^4}{2} \nonumber \\
    &= \eps\(\cop{a}{2}\)\<\cop{g}{4}, \cop{\psi^1}{5}\>\<\cop{g}{5}, S^{-1}\cop{\psi^1}{1}\>\<\cop{h}{2}, S\cop{\psi^1}{2}\cop{\psi^1}{4} \>\cop{\psi^1}{3} \nonumber \\
    &\qquad \tens \eps\(\cop{h}{4}\)\<S\(\cop{g}{3}\), \cop{\psi^2}{1}\>\<\cop{g}{6}, \cop{\psi^2}{4}\>\<S\(\cop{a}{1}\cop{h}{1}S\cop{h}{3}\), \cop{\psi^2}{2}\>\cop{\psi^2}{3} \nonumber\\
    &\qquad \tens \<\cop{g}{2}, \cop{\psi^3}{3} \>\<\cop{g}{7}, S^{-1}\cop{\psi^3}{1}\>\cop{\psi^3}{2} \tens \<\cop{g}{1}, \cop{\psi^4}{3}\>\<\cop{g}{8}, S^{-1}\cop{\psi^4}{1}\>\cop{\psi^4}{2} 
\end{align}

We consider another type-$B$ ribbon $\rho^{A} = \tau_{1} \cup \tau_{2}$ in \eqref{sigma_3:relation} which is also made up of a right direct triangle ($\tau_1$) and a left dual triangle ($\tau_2$). This will enable us to compute the commutation relation between the vertex operator and the type-$A$ ribbon operators at the end sites of type-$A$ ribbons.

\begin{equation}
    \begin{minipage}[c]{0.05\textwidth}
        $\ket{\Sigma_{3}}$
    \end{minipage}
    \begin{minipage}[c]{0.05\textwidth}
        $=$
    \end{minipage}
    \begin{minipage}[c]{0.5\textwidth}
        \begin{tikzpicture}[scale = 0.8]

    \begin{scope}[thick, decoration={markings, mark=at position 0.7 with {\arrow{latex}}}]
        \draw[postaction={decorate}] (5*\rcs,5*\rcs) -- (5*\rcs,9*\rcs);
        \draw[postaction={decorate}] (5*\rcs,5*\rcs) -- (1*\rcs,5*\rcs);
        \draw[postaction={decorate}] (9*\rcs,5*\rcs) -- (5*\rcs,5*\rcs);
        \draw[postaction={decorate}] (5*\rcs,5*\rcs) -- (5*\rcs,1*\rcs);
    \end{scope}

    \begin{scope}[thin, red, dashed, decoration={markings, mark=at position 0.3 with {\arrow{latex}}}]
        \draw[postaction={decorate}] (7*\rcs,7*\rcs) -- (3*\rcs, 7*\rcs);
    \end{scope}

    \draw[dotted, thick] (1*\rcs, 5*\rcs) -- (3*\rcs, 7*\rcs);
    \draw[dotted, thick] (5*\rcs, 5*\rcs) -- (3*\rcs, 7*\rcs);
    \draw[dotted, very thick] (5*\rcs, 5*\rcs) -- (7*\rcs, 7*\rcs);


    \harrow{\rcs}{(2*\rcs, 6*\rcs)}{(4*\rcs, 6*\rcs)}{0.07}{1}
    \harrow{\rcs}{(4*\rcs, 6*\rcs)}{(6*\rcs, 6*\rcs)}{0.07}{0}
    
    \draw[text=black, draw=white, fill=white] (6.5*\rcs, 4.5*\rcs) circle (5pt) node {\footnotesize $\psi^{1}$};
    \draw[text=black, draw=white, fill=white] (5.5*\rcs, 3*\rcs) circle (5pt) node {\footnotesize $\psi^{2}$};
    \draw[text=black, draw=white, fill=white] (3.0*\rcs, 4.5*\rcs) circle (5pt) node {\footnotesize $\psi^{3}$};
    \draw[text=black, draw=white, fill=white] (5.5*\rcs, 8.0*\rcs) circle (5pt) node {\footnotesize $\psi^{4}$};

    \draw[text=black, draw=white, fill=white] (6.0*\rcs, 6.0*\rcs) circle (5pt) node {\footnotesize $s_{1}$};
    \draw[text=black, draw=white, fill=white] (3.0*\rcs, 5.5*\rcs) circle (5pt) node {\footnotesize $\tau_{1}$};
    \draw[text=black, draw=white, fill=white] (5.0*\rcs, 6.5*\rcs) circle (5pt) node {\footnotesize $\tau_{2}$};

        



\end{tikzpicture}
    \end{minipage}
    \label{sigma_3:relation}
\end{equation}
At $s_{1}$, and $A^{g}\(s_{1}\)F^{a\tens h}\(\rho^{B}\)$ gives
\begin{align}
    A^{g}\(s_{1}\)F^{a\tens h}\(\rho^{B}\)\ket{\Sigma_3}
    &= A^{g}(s_{1})\left[\psi^1 \tens \psi^2 \tens \rop[1]{\cop{a}{2}}{\cop{h}{2}}\(\psi^3\) \tens \rop[2]{\cop{a}{1}\cop{h}{2}S\cop{h}{3}}{\cop{h}{4}}\(\psi^4\) \right] \nonumber\\
    &= A^{g}(s_{1})\left[ \psi^1 \tens \psi^2 \tens \eps\(\cop{h}{2}\)\<S\cop{a}{2}, \cop{\psi^3}{1}\>\cop{\psi^3}{2} \tens \eps\(\cop{a}{1}\cop{h}{1}S\cop{h}{3}\) \<\cop{h}{4}, S\cop{\psi^4}{1}\cop{\psi^4}{3} \>\cop{\psi^4}{2} \right] \nonumber\\
    &= \PP{1}{\cop{g}{4}}\(\psi^1\) \tens \PM{\cop{g}{3}S\cop{g}{5}}{\cop{g}{6}}\(\psi^2\) \tens \eps\(\cop{h}{2}\)\<S\cop{a}{2}, \cop{\psi^3}{1}\>\PM{\cop{g}{2}S\cop{g}{7}}{\cop{g}{8}}\(\cop{\psi^3}{2}\) \nonumber \\
    &\qquad \tens \eps\(\cop{a}{1}\cop{h}{1}S\cop{h}{3}\) \<\cop{h}{4}, S\cop{\psi^4}{1}\cop{\psi^4}{3}\>\PM{\cop{g}{1}S\cop{g}{9}}{\cop{g}{10}}\(\cop{\psi^4}{2}\) \nonumber\\
    &= \<S\copp{g}{4}{1}, \cop{\psi^1}{1}\>\<\copp{g}{4}{2}, \cop{\psi^1}{3}\>\cop{\psi^1}{2} \tens \<\(\cop{g}{3}S\cop{g}{5}\)\copp{g}{6}{1}, \cop{\psi^2}{3}\>\<\copp{g}{6}{2}, \cop{\psi^2}{1}\> \cop{\psi^2}{2} \nonumber \\
    &\qquad \tens \eps\(\cop{h}{2}\)\<S\cop{a}{2}, \cop{\psi^3}{1}\>\<\(\cop{g}{2}S\cop{g}{7}\)\copp{g}{8}{1}, \copp{\psi^3}{2}{3}\>\<\copp{g}{8}{2}, \copp{\psi^3}{2}{1}\>\copp{\psi^3}{2}{2} \nonumber \\          
    &\qquad \tens \eps\(\cop{a}{1}\cop{h}{1}S\cop{h}{3}\) \<\cop{h}{4}, S\cop{\psi^4}{1}\cop{\psi^4}{3}\>\<\(\cop{g}{1}S\cop{g}{9}\)\copp{g}{10}{1}, \copp{\psi^4}{2}{3}\>\<\copp{g}{10}{2}, \copp{\psi^4}{2}{1}\>\copp{\psi^4}{2}{2} \nonumber \\
    &= \<S\cop{g}{4}, \cop{\psi^1}{1}\>\<\cop{g}{5}, \cop{\psi^1}{3}\>\cop{\psi^1}{2} \tens \<\cop{g}{3}, \cop{\psi^2}{3}\>\<\cop{g}{6}, \cop{\psi^2}{1}\> \cop{\psi^2}{2} \nonumber \\
    &\qquad \tens \eps\(\cop{h}{2}\)\<S\cop{a}{2}, \cop{\psi^3}{1}\>\<\cop{g}{2}, \cop{\psi^3}{4}\>\<\cop{g}{7}, \cop{\psi^3}{2}\>\cop{\psi^3}{3} \nonumber \\          
    &\qquad \tens \eps\(\cop{a}{1}\cop{h}{1}S\cop{h}{3}\) \<\cop{h}{4}, S\cop{\psi^4}{1}\cop{\psi^4}{5}\>\<\cop{g}{1}, \cop{\psi^4}{4}\>\<\cop{g}{8}, \cop{\psi^4}{2}\>\cop{\psi^4}{3}
\end{align}
and $F^{a\tens h}\(\rho^{B}\)A^{g}\(s_{1}\)$ gives
\begin{align}
    F^{a\tens h}\(\rho^{B}\)A^{g}\(s_{1}\)\ket{\Sigma_{3}}
    &= F^{a\tens h}\(\rho^{B}\)\left[\PP{1}{\cop{g}{4}}\(\psi^1\) \tens \PM{\cop{g}{3}S\cop{g}{5}}{\cop{g}{6}}\(\psi^2\) \tens \PM{\cop{g}{2}S\cop{g}{7}}{\cop{g}{8}}\(\psi^3\) \right. \nonumber \\
    &\qquad \left. \tens \PM{\cop{g}{1}S\cop{g}{9}}{\cop{g}{10}}\(\psi^4\) \right] \nonumber \\
    &= F^{a\tens h}\(\rho^{B}\)\left[\<S\copp{g}{4}{1}, \cop{\psi^1}{1}\>\<\copp{g}{4}{2}, \cop{\psi^1}{3}\>\cop{\psi^1}{2} \right. \nonumber \\
    &\qquad \left. \tens \<\(\cop{g}{3}S\cop{g}{5}\)\copp{g}{6}{1}, \cop{\psi^2}{3}\>\<\copp{g}{6}{2}, \cop{\psi^2}{1}\>\cop{\psi^2}{2} \right. \nonumber \\
    &\qquad \left. \tens \<\(\cop{g}{2}S\cop{g}{7}\)\copp{g}{8}{1}, \cop{\psi^3}{3}\>\<\copp{g}{8}{2}, \cop{\psi^3}{1}\>\cop{\psi^3}{2} \right. \nonumber \\
    &\qquad \left. \tens \<\(\cop{g}{1}S\cop{g}{9}\)\copp{g}{10}{1}, \cop{\psi^4}{3}\>\<\copp{g}{10}{2}, \cop{\psi^4}{1}\>\cop{\psi^4}{2} \right] \nonumber \\
    &= \<S\cop{g}{4}, \cop{\psi^1}{1}\>\<\cop{g}{5}, \cop{\psi^1}{3}\>\cop{\psi^1}{2} \tens \<\cop{g}{3}, \cop{\psi^2}{3}\>\<\cop{g}{6}, \cop{\psi^2}{1}\>\cop{\psi^2}{2} \nonumber \\
    &\qquad \tens \<\cop{g}{2}, \cop{\psi^3}{3}\>\<\cop{g}{7}, \cop{\psi^3}{1}\>\rop[1]{\cop{a}{2}}{\cop{h}{2}}\(\cop{\psi^3}{2}\)  \nonumber \\ 
    &\qquad \tens \<\cop{g}{1}, \cop{\psi^4}{3}\>\<\cop{g}{8}, \cop{\psi^4}{1}\>\rop[2]{\cop{a}{1}\cop{h}{1}S\cop{h}{3}}{\cop{h}{4}}\(\cop{\psi^4}{2}\) \nonumber \\
    &= \<S\cop{g}{4}, \cop{\psi^1}{1}\>\<\cop{g}{5}, \cop{\psi^1}{3}\>\cop{\psi^1}{2} \tens \<\cop{g}{3}, \cop{\psi^2}{3}\>\<\cop{g}{6}, \cop{\psi^2}{1}\>\cop{\psi^2}{2} \nonumber \\
    &\qquad \tens \<\cop{g}{2}, \cop{\psi^3}{3}\>\<\cop{g}{7}, \cop{\psi^3}{1}\>\eps\(\cop{h}{2}\)\<S\cop{a}{2}, \copp{\psi^3}{2}{1}\>\copp{\psi^3}{2}{2}  \nonumber \\ 
    &\qquad \tens \<\cop{g}{1}, \cop{\psi^4}{3}\>\<\cop{g}{8}, \cop{\psi^4}{1}\>\eps\(\cop{a}{1}\cop{h}{1}S\cop{h}{3}\)\<\cop{h}{4}, S\copp{\psi^4}{2}{1}\copp{\psi^3}{2}{3}\>\copp{\psi^3}{2}{2} \nonumber \\
    &= \<S\cop{g}{4}, \cop{\psi^1}{1}\>\<\cop{g}{5}, \cop{\psi^1}{3}\>\cop{\psi^1}{2} \tens \<\cop{g}{3}, \cop{\psi^2}{3}\>\<\cop{g}{6}, \cop{\psi^2}{1}\>\cop{\psi^2}{2} \nonumber \\
    &\qquad \tens \eps\(\cop{h}{2}\)\<S\cop{a}{2}, \cop{\psi^3}{2}\>\<\cop{g}{2}, \cop{\psi^3}{4}\>\<\cop{g}{7}, \cop{\psi^3}{1}\>\cop{\psi^3}{3}  \nonumber \\ 
    &\qquad \tens \eps\(\cop{a}{1}\cop{h}{1}S\cop{h}{3}\)\<\cop{h}{4}, S\cop{\psi^4}{2}\cop{\psi^3}{4}\>\<\cop{g}{1}, \cop{\psi^4}{5}\>\<\cop{g}{8}, \cop{\psi^4}{1}\>\cop{\psi^3}{3}
\end{align}

Next we compute the commutation relations between the type-$A$ ribbon operators and the face operators at the starting and end sites of the type-$A$ ribbons. We begin by considering the type-$A$ ribbon illustrated in \eqref{sigma_4:relation}, it is made up of a single right dual triangle.

\begin{equation}
    \begin{minipage}[c]{0.05\textwidth}
        $\ket{\Sigma_{4}}$
    \end{minipage}
    \begin{minipage}[c]{0.05\textwidth}
        $=$
    \end{minipage}
    \begin{minipage}[c]{0.5\textwidth}
        \begin{tikzpicture}[scale = 0.8]

    \begin{scope}[thick, decoration={markings, mark=at position 0.3 with {\arrow{latex}}}]
        \draw[postaction={decorate}] (5*\rcs,1*\rcs) -- (5*\rcs,5*\rcs);
        \draw[postaction={decorate}] (5*\rcs,5*\rcs) -- (1*\rcs,5*\rcs);
        \draw[postaction={decorate}] (1*\rcs,5*\rcs) -- (1*\rcs,1*\rcs);
        \draw[postaction={decorate}] (1*\rcs,1*\rcs) -- (5*\rcs,1*\rcs);
        \draw[postaction={decorate}] (5*\rcs,1*\rcs) -- (9*\rcs,1*\rcs);
        \draw[postaction={decorate}] (9*\rcs,1*\rcs) -- (9*\rcs,5*\rcs);
        \draw[postaction={decorate}] (9*\rcs,5*\rcs) -- (5*\rcs,5*\rcs);
    \end{scope}

    \begin{scope}[thin, red, dashed, decoration={markings, mark=at position 0.3 with {\arrow{latex}}}]
        \draw[postaction={decorate}] (7*\rcs, 3*\rcs) -- (3*\rcs, 3*\rcs);
    \end{scope}

    \draw[dotted, thick] (5*\rcs, 5*\rcs) -- (3*\rcs, 3*\rcs);
    \draw[dotted, thick] (5*\rcs, 5*\rcs) -- (7*\rcs, 3*\rcs);


    \harrow{\rcs}{(4*\rcs, 4*\rcs)}{(6*\rcs, 4*\rcs)}{0.07}{1}

    \draw[text=black, draw=white, fill=white] (3.0*\rcs, 5.5*\rcs) circle (5pt) node {\footnotesize $\psi^{1}$};
    \draw[text=black, draw=white, fill=white] (0.5*\rcs, 3.0*\rcs) circle (5pt) node {\footnotesize $\psi^{2}$};
    \draw[text=black, draw=white, fill=white] (3.0*\rcs, 0.5*\rcs) circle (5pt) node {\footnotesize $\psi^{3}$};
    \draw[text=black, draw=white, fill=white] (5.5*\rcs, 2.0*\rcs) circle (5pt) node {\footnotesize $\psi^{4}$};
    \draw[text=black, draw=white, fill=white] (7.0*\rcs, 0.5*\rcs) circle (5pt) node {\footnotesize $\psi^{5}$};
    \draw[text=black, draw=white, fill=white] (9.5*\rcs, 3.0*\rcs) circle (5pt) node {\footnotesize $\psi^{6}$};
    \draw[text=black, draw=white, fill=white] (7.0*\rcs, 5.5*\rcs) circle (5pt) node {\footnotesize $\psi^{7}$};

    \draw[text=black, draw=white, fill=white] (4.0*\rcs, 4.0*\rcs) circle (5pt) node {\footnotesize $s_{0}$};
    \draw[text=black, draw=white, fill=white] (6.0*\rcs, 4.0*\rcs) circle (5pt) node {\footnotesize $s_{1}$};
    \draw[text=black, draw=white, fill=white] (5.0*\rcs, 3.5*\rcs) circle (5pt) node {\footnotesize $\tau$};

        



\end{tikzpicture}
    \end{minipage}
    \label{sigma_4:relation}
\end{equation}

At $s_{0}$, $F^{a\tens h}\(\rho^{A}\)B^{b}\(s_{0}\)$ gives
\begin{align}
    F^{a\tens h}\(\rho^{A}\)B^{b}\(s_{0}\)\ket{\Sigma_4} 
    &= F^{a\tens h}\(\rho^{A}\)\left[\TM{\cop{b}{4}}\(\psi^1\) \tens \TM{\cop{b}{3}}\(\psi^2\) \tens \TM{\cop{b}{2}}\(\psi^3\) \tens \TM{\cop{b}{1}}\(\psi^4\)\right] \nonumber \\
    &= F^{a\tens h}\(\rho^{A}\)\left[\<\cop{b}{4}, \cop{\psi^1}{2}\>\cop{\psi^1}{1} \tens \<\cop{b}{3}, \cop{\psi^2}{2}\>\cop{\psi^2}{1}  \tens \<\cop{b}{2}, \cop{\psi^3}{2}\>\cop{\psi^3}{1}  \tens \<\cop{b}{1}, \cop{\psi^4}{2}\>\cop{\psi^4}{1} \right] \nonumber \\
    &= \<\cop{b}{4}, \cop{\psi^1}{2}\>\cop{\psi^1}{1} \tens \<\cop{b}{3}, \cop{\psi^2}{2}\>\cop{\psi^2}{1}  \tens \<\cop{b}{2}, \cop{\psi^3}{2}\>\cop{\psi^3}{1} \tens \<\cop{b}{1}, \cop{\psi^4}{2}\>\rop{a}{h}\(\cop{\psi^4}{1}\)  \nonumber \\
    &= \<\cop{b}{4}, \cop{\psi^1}{2}\>\cop{\psi^1}{1} \tens \<\cop{b}{3}, \cop{\psi^2}{2}\>\cop{\psi^2}{1}  \tens \<\cop{b}{2}, \cop{\psi^3}{2}\>\cop{\psi^3}{1}  \nonumber \\
    &\qquad \tens \<\cop{b}{1}, \cop{\psi^4}{2}\>\eps(a)\<h, \copp{\psi^4}{1}{1}S^{-1}\copp{\psi^4}{1}{3}\>\copp{\psi^4}{1}{2} \nonumber \\
    &= \<\cop{b}{4}, \cop{\psi^1}{2}\>\cop{\psi^1}{1} \tens \<\cop{b}{3}, \cop{\psi^2}{2}\>\cop{\psi^2}{1}  \tens \<\cop{b}{2}, \cop{\psi^3}{2}\>\cop{\psi^3}{1}  \nonumber \\
    &\qquad \tens \eps(a)\<\cop{b}{1}, \cop{\psi^4}{4}\>\<h, \cop{\psi^4}{1}S^{-1}\cop{\psi^4}{3}\>\cop{\psi^4}{2}
\end{align}
and $B^{b}\(s_{0}\)F^{a\tens h}\(\rho^{A}\)$ gives
\begin{align}
    B^{b}\(s_{0}\)F^{a\tens h}\(\rho^{A}\)\ket{\Sigma_4}
    &= B^{b}\(s_{0}\)\left[ \psi^1 \tens \psi^2 \tens \psi^3 \tens F^{a\tens h}(\tau)\(\psi^4\)\right] \nonumber \\
    &= B^{b}\(s_{0}\)\left[ \psi^1 \tens \psi^2 \tens \psi^3 \tens \eps(a)\<h, \cop{\psi^4}{1}S^{-1}\cop{\psi^4}{3}\>\cop{\psi^4}{2}\right] \nonumber \\
    &= \TM{\cop{b}{4}}\(\psi^1\) \tens \TM{\cop{b}{3}}\(\psi^2\) \tens \TM{\cop{b}{2}}\(\psi^3\) \tens \eps(a)\<h, \cop{\psi^4}{1}S^{-1}\cop{\psi^4}{3}\>\TM{\cop{b}{1}}\(\cop{\psi^4}{2}\) \nonumber \\
    &= \<\cop{b}{4}, \cop{\psi^1}{2}\>\cop{\psi^1}{1} \tens \<\cop{b}{3}, \cop{\psi^3}{2}\>\cop{\psi^3}{1} \tens \<\cop{b}{2}, \cop{\psi^3}{2}\>\cop{\psi^3}{1} \nonumber \\
    &\qquad \tens \eps(a)\<h, \cop{\psi^4}{1}S^{-1}\cop{\psi^4}{3}\>\<\cop{b}{1}, \copp{\psi^4}{2}{2}\>\copp{\psi^4}{2}{1} \nonumber \\
    &= \<\cop{b}{4}, \cop{\psi^1}{2}\>\cop{\psi^1}{1} \tens \<\cop{b}{3}, \cop{\psi^3}{2}\>\cop{\psi^3}{1} \tens \<\cop{b}{2}, \cop{\psi^3}{2}\>\cop{\psi^3}{1} \nonumber \\
    &\qquad \tens \eps(a)\<h, \cop{\psi^4}{1}S^{-1}\cop{\psi^4}{4}\>\<\cop{b}{1}, \cop{\psi^4}{3}\>\cop{\psi^4}{2}
\end{align}

At $s_{1}$, $F^{a\tens h}\(\rho^{A}\)B^{b}\(s_{1}\)$ gives
\begin{align}
    F^{a\tens h}\(\rho^{A}\)B^{b}\(s_{1}\)\ket{\Sigma_4}
    &= F^{a\tens h}\(\rho^{A}\)\left[\TP{\cop{b}{4}}\(\psi^4\) \tens \TM{\cop{b}{3}}\(\psi^5\) \tens \TM{\cop{b}{2}}\(\psi^6\) \tens \TM{\cop{b}{1}}\(\psi^7\)\right] \nonumber \\
    &= F^{a\tens h}\(\rho^{A}\)\left[\<S\cop{b}{4}, \cop{\psi^4}{1}\>\cop{\psi^4}{2} \tens \<\cop{b}{3}, \cop{\psi^5}{2}\>\cop{\psi^5}{1} \tens \<\cop{b}{2}, \cop{\psi^6}{2}\>\cop{\psi^6}{1} \tens \<\cop{b}{1}, \cop{\psi^7}{2}\>\cop{\psi^7}{1}\right] \nonumber \\
    &= \<S\cop{b}{4}, \cop{\psi^4}{1}\>\rop{a}{h}\(\cop{\psi^4}{2}\) \tens \<\cop{b}{3}, \cop{\psi^5}{2}\>\cop{\psi^5}{1} \tens \<\cop{b}{2}, \cop{\psi^6}{2}\>\cop{\psi^6}{1} \tens \<\cop{b}{1}, \cop{\psi^7}{2}\>\cop{\psi^7}{1} \nonumber \\
    &= \eps(a)\<S\cop{b}{4}, \cop{\psi^4}{1}\>\<h, \copp{\psi^4}{2}{1}S^{-1}\copp{\psi^4}{2}{3}\>\copp{\psi^4}{2}{2} \nonumber \\
    &\qquad \tens \<\cop{b}{3}, \cop{\psi^5}{2}\>\cop{\psi^5}{1} \tens \<\cop{b}{2}, \cop{\psi^6}{2}\>\cop{\psi^6}{1} \tens \<\cop{b}{1}, \cop{\psi^7}{2}\>\cop{\psi^7}{1} \nonumber \\
    &= \eps(a)\<S\cop{b}{4}, \cop{\psi^4}{1}\>\<h, \cop{\psi^4}{2}S^{-1}\cop{\psi^4}{4}\>\cop{\psi^4}{3} \nonumber \\
    &\qquad \tens \<\cop{b}{3}, \cop{\psi^5}{2}\>\cop{\psi^5}{1} \tens \<\cop{b}{2}, \cop{\psi^6}{2}\>\cop{\psi^6}{1} \tens \<\cop{b}{1}, \cop{\psi^7}{2}\>\cop{\psi^7}{1}
\end{align}
and $B^{b}\(s_{1}\)F^{a\tens h}\(\rho^{A}\)$ gives
\begin{align}
    B^{b}\(s_{1}\)F^{a\tens h}\(\rho^{A}\)\ket{\Sigma_4}
    &= B^{b}\(s_{1}\)\left[\rop{a}{h}\(\psi^4\) \tens \psi^5 \tens \psi^6 \tens \psi^7 \right] \nonumber \\
    &= B^{b}\(s_{1}\)\left[\eps(a)\<h, \cop{\psi^4}{1}S^{-1}\cop{\psi^4}{3}\>\cop{\psi^4}{2} \tens \psi^5 \tens \psi^6 \tens \psi^7 \right] \nonumber \\
    &= \eps(a)\<h, \cop{\psi^4}{1}S^{-1}\cop{\psi^4}{3}\>\TP{\cop{b}{4}}\(\cop{\psi^4}{2}\) \tens \TM{\cop{b}{3}}\(\psi^5\) \tens \TM{\cop{b}{2}}\(\psi^6\) \tens \TM{\cop{b}{1}}\(\psi^7\) \nonumber \\
    &= \eps(a)\<h, \cop{\psi^4}{1}S^{-1}\cop{\psi^4}{3}\>\<S\cop{b}{4}, \copp{\psi^4}{2}{1}\>\copp{\psi^4}{2}{2} \nonumber \\
    &\qquad \tens \<\cop{b}{3}, \cop{\psi^5}{2}\>\cop{\psi^5}{1} \tens \<\cop{b}{2}, \cop{\psi^6}{2}\>\cop{\psi^6}{1} \tens \<\cop{b}{1}, \cop{\psi^7}{2}\>\cop{\psi^7}{1} \nonumber \\
    &= \eps(a)\<h, \cop{\psi^4}{1}S^{-1}\cop{\psi^4}{4}\>\<S\cop{b}{4}, \cop{\psi^4}{2}\>\cop{\psi^4}{3} \nonumber \\
    &\qquad \tens \<\cop{b}{3}, \cop{\psi^5}{2}\>\cop{\psi^5}{1} \tens \<\cop{b}{2}, \cop{\psi^6}{2}\>\cop{\psi^6}{1} \tens \<\cop{b}{1}, \cop{\psi^7}{2}\>\cop{\psi^7}{1} \nonumber \\
\end{align}

Finally, we compute the commutation relation between the type-$A$ ribbon operators and the vertex operators at the start and end sites of type-$A$ ribbons, we will do so by considering the type-$A$ ribbon illustrated in \eqref{sigma_5:relation}, it is made up of a single left direct triangle.

\begin{equation}
    \begin{minipage}[c]{0.05\textwidth}
        $\ket{\Sigma_{5}}$
    \end{minipage}
    \begin{minipage}[c]{0.05\textwidth}
        $=$
    \end{minipage}
    \begin{minipage}[c]{0.5\textwidth}
        \begin{tikzpicture}[scale = 0.8]

    \begin{scope}[thick, decoration={markings, mark=at position 0.3 with {\arrow{latex}}}]
        \draw[postaction={decorate}] (5*\rcs,5*\rcs) -- (5*\rcs,1*\rcs);
        \draw[postaction={decorate}] (5*\rcs,5*\rcs) -- (1*\rcs,5*\rcs);
        \draw[postaction={decorate}] (5*\rcs,5*\rcs) -- (5*\rcs,9*\rcs);
        \draw[postaction={decorate}] (9*\rcs,5*\rcs) -- (5*\rcs,5*\rcs);
        \draw[postaction={decorate}] (9*\rcs,5*\rcs) -- (13*\rcs,5*\rcs);
        \draw[postaction={decorate}] (9*\rcs,5*\rcs) -- (9*\rcs,9*\rcs);
        \draw[postaction={decorate}] (9*\rcs,5*\rcs) -- (9*\rcs,1*\rcs);
    \end{scope}


    \draw[dotted, thick] (5*\rcs, 5*\rcs) -- (7*\rcs, 3*\rcs);
    \draw[dotted, thick] (7*\rcs, 3*\rcs) -- (9*\rcs, 5*\rcs);


    \harrow{\rcs}{(6*\rcs, 4*\rcs)}{(8*\rcs, 4*\rcs)}{0.07}{1}
    
    \draw[text=black, draw=white, fill=white] (4.5*\rcs, 3.0*\rcs) circle (5pt) node {\footnotesize $\psi^{1}$};
    \draw[text=black, draw=white, fill=white] (3.0*\rcs, 5.5*\rcs) circle (5pt) node {\footnotesize $\psi^{2}$};
    \draw[text=black, draw=white, fill=white] (4.5*\rcs, 7.0*\rcs) circle (5pt) node {\footnotesize $\psi^{3}$};
    \draw[text=black, draw=white, fill=white] (7.0*\rcs, 5.5*\rcs) circle (5pt) node {\footnotesize $\psi^{4}$};
    \draw[text=black, draw=white, fill=white] (8.5*\rcs, 7.0*\rcs) circle (5pt) node {\footnotesize $\psi^{5}$};
    \draw[text=black, draw=white, fill=white] (11.0*\rcs, 5.5*\rcs) circle (5pt) node {\footnotesize $\psi^{6}$};
    \draw[text=black, draw=white, fill=white] (9.5*\rcs, 3.0*\rcs) circle (5pt) node {\footnotesize $\psi^{7}$};

    \draw[text=black, draw=white, fill=white] (6.0*\rcs, 4.0*\rcs) circle (5pt) node {\footnotesize $s_{0}$};
    \draw[text=black, draw=white, fill=white] (8.0*\rcs, 4.0*\rcs) circle (5pt) node {\footnotesize $s_{1}$};
    \draw[text=black, draw=white, fill=white] (7.0*\rcs, 4.5*\rcs) circle (5pt) node {\footnotesize $\tau$};

        



\end{tikzpicture}
    \end{minipage}
    \label{sigma_5:relation}
\end{equation}

At $s_{0}$, $F^{a\tens h}\(\rho^{A}\)A^{g}\(s_{0}\)$ gives
\begin{align}
    F^{a\tens h}\(\rho^{A}\)A^{g}\(s_{0}\)\ket{\Sigma_5}
    &= F^{a\tens h}\(\rho^{A}\)\left[\PM{1}{\cop{g}{4}}\(\psi^1\) \tens \PM{\cop{g}{3}S\cop{g}{5}}{\cop{g}{6}}\(\psi^2\) \right. \nonumber \\
    &\qquad \left. \tens \PM{\cop{g}{2}S\cop{g}{7}}{\cop{g}{8}}\(\psi^3\) \tens \PP{\cop{g}{1}S\cop{g}{9}}{\cop{g}{10}}\(\psi^4\)\right] \nonumber \\
    &= F^{a\tens h}\(\rho^{A}\)\left[\<\copp{g}{4}{1}, \cop{\psi^1}{3}\>\<\copp{g}{4}{2}, \cop{\psi^1}{1}\>\cop{\psi^1}{2} \right. \nonumber \\
    &\qquad \left. \tens \<\(\cop{g}{3}S\cop{g}{5}\)\copp{g}{6}{1}, \cop{\psi^2}{3}\>\<\copp{g}{6}{2}, \cop{\psi^2}{1}\>\cop{\psi^2}{2} \right. \nonumber \\
    &\qquad \left. \tens \<\(\cop{g}{2}S\cop{g}{7}\)\copp{g}{8}{1}, \cop{\psi^3}{3}\>\<\copp{g}{8}{2}, \cop{\psi^3}{1}\>\cop{\psi^3}{2} \right. \nonumber \\
    &\qquad \left. \tens \<S\copp{g}{10}{1}S\(\cop{g}{1}S\cop{g}{9}\), \cop{\psi^4}{1}\>\<\copp{g}{10}{2}, \cop{\psi^4}{3}\>\cop{\psi^4}{2} \right] \nonumber \\
    &= \<\cop{g}{4}, \cop{\psi^1}{3}\>\<\cop{g}{5}, \cop{\psi^1}{1}\>\cop{\psi^1}{2} \tens \<\cop{g}{3}, \cop{\psi^2}{3}\>\<\cop{g}{6}, \cop{\psi^2}{1}\>\cop{\psi^2}{2} \nonumber \\
    &\qquad \tens \<\cop{g}{2}, \cop{\psi^3}{3}\>\<\cop{g}{7}, \cop{\psi^3}{1}\>\cop{\psi^3}{2}  \tens \<S\cop{g}{1}, \cop{\psi^4}{1}\>\<\cop{g}{8}, \cop{\psi^4}{3}\>\rop{a}{h}\(\cop{\psi^4}{2}\) \nonumber \\
    &= \<\cop{g}{4}, \cop{\psi^1}{3}\>\<\cop{g}{5}, \cop{\psi^1}{1}\>\cop{\psi^1}{2} \tens \<\cop{g}{3}, \cop{\psi^2}{3}\>\<\cop{g}{6}, \cop{\psi^2}{1}\>\cop{\psi^2}{2} \nonumber \\
    &\qquad \tens \<\cop{g}{2}, \cop{\psi^3}{3}\>\<\cop{g}{7}, \cop{\psi^3}{1}\>\cop{\psi^3}{2}  \tens \<S\cop{g}{1}, \cop{\psi^4}{1}\>\<\cop{g}{8}, \cop{\psi^4}{3}\>\eps(h)\<\Si a, \copp{\psi^4}{2}{2}\>\copp{\psi^4}{2}{1} \nonumber \\
    &= \<\cop{g}{4}, \cop{\psi^1}{3}\>\<\cop{g}{5}, \cop{\psi^1}{1}\>\cop{\psi^1}{2} \tens \<\cop{g}{3}, \cop{\psi^2}{3}\>\<\cop{g}{6}, \cop{\psi^2}{1}\>\cop{\psi^2}{2} \nonumber \\
    &\qquad \tens \<\cop{g}{2}, \cop{\psi^3}{3}\>\<\cop{g}{7}, \cop{\psi^3}{1}\>\cop{\psi^3}{2}  \tens \eps(h)\<S\cop{g}{1}, \cop{\psi^4}{1}\>\<\cop{g}{8}, \cop{\psi^4}{4}\>\<\Si a, \cop{\psi^4}{3}\>\cop{\psi^4}{2}
\end{align}
and $A^{g}\(s_{0}\)F^{a\tens h}\(\rho^{A}\)$ gives
\begin{align}
    A^{g}\(s_{0}\)F^{a\tens h}\(\rho^{A}\)\ket{\Sigma_5}
    &= A^{g}\(s_{0}\)\left[ \psi^1 \tens \psi^2 \tens \psi^3 \tens \rop{a}{b}\(\psi^4\) \right] \nonumber \\
    &= A^{g}\(s_{0}\)\left[ \psi^1 \tens \psi^2 \tens \psi^3 \tens \eps(h)\<\Si a, \cop{\psi^4}{2}\>\cop{\psi^4}{1} \right] \nonumber \\
    &= \PM{1}{\cop{g}{4}}\(\psi^1\) \tens \PM{\cop{g}{3}S\cop{g}{5}}{\cop{g}{6}}\(\psi^2\) \nonumber \\
    &\qquad \tens \PM{\cop{g}{2}S\cop{g}{7}}{\cop{g}{8}}\(\psi^3\) \tens \eps(h)\<\Si a, \cop{\psi^4}{2}\>\PP{\cop{g}{1}S\cop{g}{9}}{\cop{g}{10}}\(\cop{\psi^4}{1}\) \nonumber \\
    &= \<\copp{g}{4}{1}, \cop{\psi^1}{3}\>\<\copp{g}{4}{2}, \cop{\psi^1}{1}\>\cop{\psi^1}{2} \tens \<\(\cop{g}{3}S\cop{g}{5}\)\copp{g}{6}{1}, \cop{\psi^2}{3}\>\<\copp{g}{6}{2}, \cop{\psi^2}{1}\>\cop{\psi^2}{2} \nonumber \\
    &\qquad \tens \<\(\cop{g}{2}S\cop{g}{7}\)\copp{g}{8}{1}, \cop{\psi^3}{3}\>\<\copp{g}{8}{2}, \cop{\psi^3}{1}\>\cop{\psi^3}{2} \nonumber \\
    &\qquad \tens \eps(h)\<\Si a, \cop{\psi^4}{2}\>\<S\copp{g}{10}{1}S\(\cop{g}{1}S\cop{g}{9}\), \copp{\psi^4}{1}{1}\>\<\copp{g}{10}{2}, \copp{\psi^4}{1}{3}\>\copp{\psi^4}{1}{2} \nonumber \\
    &= \<\cop{g}{4}, \cop{\psi^1}{3}\>\<\cop{g}{5}, \cop{\psi^1}{1}\>\cop{\psi^1}{2} \tens \<\cop{g}{3}, \cop{\psi^2}{3}\>\<\cop{g}{6}, \cop{\psi^2}{1}\>\cop{\psi^2}{2} \nonumber \\
    &\qquad \tens \<\cop{g}{2}, \cop{\psi^3}{3}\>\<\cop{g}{7}, \cop{\psi^3}{1}\>\cop{\psi^3}{2}  \tens \eps(h)\<S\cop{g}{1}, \cop{\psi^4}{1}\>\<\cop{g}{8}, \cop{\psi^4}{3}\>\<\Si a, \cop{\psi^4}{4}\>\cop{\psi^4}{2}
\end{align}

At $s_{1}$, $F^{a\tens h}\(\rho^{A}\)A^{g}\(s_{1}\)$ gives
\begin{align}
    F^{a\tens h}\(\rho^{A}\)A^{g}\(s_{1}\)\ket{\Sigma_5}
    &= F^{a\tens h}\(\rho^{A}\)\left[\PM{1}{\cop{g}{4}}\(\psi^4\) \tens \PM{\cop{g}{3}S\cop{g}{5}}{\cop{g}{6}}\(\psi^5\) \right. \nonumber \\
    &\qquad \left. \tens \PM{\cop{g}{2}S\cop{g}{7}}{\cop{g}{8}}\(\psi^6\) \tens \PM{\cop{g}{1}S\cop{g}{9}}{\cop{g}{10}}\(\psi^7\)\right] \nonumber \\
    &= F^{a\tens h}\(\rho^{A}\)\left[\<\copp{g}{4}{1}, \cop{\psi^4}{3}\>\<\copp{g}{4}{2}, \cop{\psi^4}{1}\>\cop{\psi^4}{2} \right. \nonumber \\
    &\qquad \left. \tens \<\(\cop{g}{3}S\cop{g}{5}\)\copp{g}{6}{1}, \cop{\psi^5}{3}\>\<\copp{g}{6}{2}, \cop{\psi^5}{1}\>\cop{\psi^5}{2} \right. \nonumber \\
    &\qquad \left. \tens \<\(\cop{g}{2}S\cop{g}{7}\)\copp{g}{8}{1}, \cop{\psi^6}{3}\>\<\copp{g}{8}{2}, \cop{\psi^6}{1}\>\cop{\psi^6}{2} \right. \nonumber \\
    &\qquad \left. \tens \<\(\cop{g}{1}S\cop{g}{9}\)\copp{g}{10}{1}, \cop{\psi^7}{1}\>\<\copp{g}{10}{2}, \cop{\psi^7}{3}\>\cop{\psi^7}{2} \right] \nonumber \\
    &= \<\cop{g}{4}, \cop{\psi^4}{3}\>\<\cop{g}{5}, \cop{\psi^4}{1}\>\rop{a}{h}\(\cop{\psi^4}{2}\) \tens \<\cop{g}{3}, \cop{\psi^5}{3}\>\<\cop{g}{6}, \cop{\psi^5}{1}\>\cop{\psi^5}{2} \nonumber \\
    &\qquad \tens \<\cop{g}{2}, \cop{\psi^6}{3}\>\<\cop{g}{7}, \cop{\psi^6}{1}\>\cop{\psi^6}{2}  \tens \<\cop{g}{1}, \cop{\psi^7}{1}\>\<\cop{g}{8}, \cop{\psi^7}{3}\>\cop{\psi^7}{2} \nonumber \\
    &= \<\cop{g}{4}, \cop{\psi^4}{3}\>\<\cop{g}{5}, \cop{\psi^4}{1}\>\eps(h)\<\Si a, \copp{\psi^4}{2}{2}\>\copp{\psi^4}{2}{1} \tens \<\cop{g}{3}, \cop{\psi^5}{3}\>\<\cop{g}{6}, \cop{\psi^5}{1}\>\cop{\psi^5}{2} \nonumber \\
    &\qquad \tens \<\cop{g}{2}, \cop{\psi^6}{3}\>\<\cop{g}{7}, \cop{\psi^6}{1}\>\cop{\psi^6}{2}  \tens \<\cop{g}{1}, \cop{\psi^7}{1}\>\<\cop{g}{8}, \cop{\psi^7}{3}\>\cop{\psi^7}{2} \nonumber \\
    &= \eps(h)\<\cop{g}{4}, \cop{\psi^4}{4}\>\<\cop{g}{5}, \cop{\psi^4}{1}\>\<\Si a, \cop{\psi^4}{3}\>\cop{\psi^4}{2} \tens \<\cop{g}{3}, \cop{\psi^5}{3}\>\<\cop{g}{6}, \cop{\psi^5}{1}\>\cop{\psi^5}{2} \nonumber \\
    &\qquad \tens \<\cop{g}{2}, \cop{\psi^6}{3}\>\<\cop{g}{7}, \cop{\psi^6}{1}\>\cop{\psi^6}{2}  \tens \<\cop{g}{1}, \cop{\psi^7}{1}\>\<\cop{g}{8}, \cop{\psi^7}{3}\>\cop{\psi^7}{2}
\end{align}
and $A^{g}\(s_{1}\)F^{a\tens h}\(\rho^{A}\)$ gives
\begin{align}
    A^{g}\(s_{1}\)F^{a\tens h}\(\rho^{A}\)\ket{\Sigma_5}
    &= A^{g}\(s_{1}\)\left[\rop{a}{b}\(\psi^4\) \tens \psi^5 \tens \psi^6 \tens \psi^7 \right] \nonumber \\
    &= A^{g}\(s_{1}\)\left[\eps(h)\<\Si a, \cop{\psi^4}{2}\>\cop{\psi^4}{1} \tens \psi^5 \tens \psi^6 \tens \psi^7 \right] \nonumber \\
    &= \eps(h)\<\Si a, \cop{\psi^4}{2}\>\PM{1}{\cop{g}{4}}\(\cop{\psi^4}{1}\) \tens \PM{\cop{g}{3}S\cop{g}{5}}{\cop{g}{6}}\(\psi^5\) \nonumber \\
    &\qquad \tens \PM{\cop{g}{2}S\cop{g}{7}}{\cop{g}{8}}\(\psi^6\) \tens \PM{\cop{g}{1}S\cop{g}{9}}{\cop{g}{10}}\(\psi^7\) \nonumber \\
    &= \eps(h)\<\Si a, \cop{\psi^4}{2}\>\<\copp{g}{4}{1}, \copp{\psi^4}{1}{3}\>\<\copp{g}{4}{2}, \copp{\psi^4}{1}{1}\>\copp{\psi^4}{1}{2} \nonumber \\
    &\qquad \tens \<\(\cop{g}{3}S\cop{g}{5}\)\copp{g}{6}{1}, \cop{\psi^5}{3}\>\<\copp{g}{6}{2}, \cop{\psi^5}{1}\>\cop{\psi^5}{2} \nonumber \\
    &\qquad \tens \<\(\cop{g}{2}S\cop{g}{7}\)\copp{g}{8}{1}, \cop{\psi^6}{3}\>\<\copp{g}{8}{2}, \cop{\psi^6}{1}\>\cop{\psi^6}{2} \nonumber \\
    &\qquad \tens \<\(\cop{g}{1}S\cop{g}{9}\)\copp{g}{10}{1}, \cop{\psi^7}{3}\>\<\copp{g}{10}{2}, \cop{\psi^7}{1}\>\cop{\psi^7}{2} \nonumber \\
    &= \eps(h)\<\Si a, \cop{\psi^4}{4}\>\<\cop{g}{4}, \cop{\psi^4}{3}\>\<\cop{g}{5}, \cop{\psi^4}{1}\>\cop{\psi^4}{2} \tens \<\cop{g}{3}, \cop{\psi^5}{3}\>\<\cop{g}{6}, \cop{\psi^5}{1}\>\cop{\psi^5}{2} \nonumber \\
    &\qquad \tens \<\cop{g}{2}, \cop{\psi^6}{3}\>\<\cop{g}{7}, \cop{\psi^6}{1}\>\cop{\psi^6}{2}  \tens \<\cop{g}{1}, \cop{\psi^7}{3}\>\<\cop{g}{8}, \cop{\psi^7}{1}\>\cop{\psi^7}{2}
\end{align}

\section{Proof that geometric operators commute with ribbon operators at sites within a ribbon}
\label{appendix:proof_cummutation_between_ribbon_and_geometric_operators}

Here we want to show that the geometric operators commute with the ribbon operators at sites within a ribbon. We will consider the two types of ribbons, that is type A and type B ribbons. The ribbons in \eqref{type_B_ribbon_with_face} and \eqref{type_B_ribbon_with_vertex} are of type B and the ribbons in \eqref{type_A_ribbon_with_face} and \eqref{type_A_ribbon_with_vertex} are of type A. For simplicity, we will consider ribbons of length three i.e $\rho^{\alpha} = \tau_{1} \cup \tau_{2} \cup \tau_{3}$, where $\alpha = A \text{ or } B$. For each case if $G(s)$ (i.e., $B^{b}(s)$ or $ A^{g}(s)$) is a geometric operator acting on a site in a ribbon, the aim is to compute $G(s)\circ \rop{a}{b}$ and compare it to $\rop{a}{b} \circ G(s)$.
\begin{equation}
    \begin{minipage}[c]{0.05\textwidth}
        $\ket{\Sigma}$
    \end{minipage}
    \begin{minipage}[c]{0.05\textwidth}
        $=$
    \end{minipage}
    \begin{minipage}[c]{0.5\textwidth}
        \begin{tikzpicture}[scale = 0.8]

    \begin{scope}[thick, decoration={markings, mark=at position 0.7 with {\arrow{latex}}}]
        \draw[postaction={decorate}] (5*\rcs, 4*\rcs) -- (9*\rcs, 4*\rcs);
        \draw[postaction={decorate}] (9*\rcs, 4*\rcs) -- (9*\rcs, 8*\rcs);
        \draw[postaction={decorate}] (9*\rcs, 8*\rcs) -- (5*\rcs, 8*\rcs);
        \draw[postaction={decorate}] (5*\rcs, 8*\rcs) -- (5*\rcs, 4*\rcs);
    \end{scope}

    \begin{scope}[thin, red, dashed, decoration={markings, mark=at position 0.3 with {\arrow{latex}}}]
        \draw[postaction={decorate}] (3*\rcs, 6*\rcs) -- (7*\rcs, 6*\rcs);
        \draw[postaction={decorate}] (11*\rcs, 6*\rcs) -- (7*\rcs, 6*\rcs);
    \end{scope}

    \draw[dotted, very thick] (5*\rcs, 4*\rcs) -- (3*\rcs, 6*\rcs);
    \draw[dotted, very thick] (5*\rcs, 4*\rcs) -- (7*\rcs, 6*\rcs);
    \draw[dotted, very thick] (7*\rcs, 6*\rcs) -- (9*\rcs, 4*\rcs);
    \draw[dotted, very thick] (9*\rcs, 4*\rcs) -- (11*\rcs, 6*\rcs);
    

    \harrow{\rcs}{(4*\rcs, 5*\rcs)}{(6*\rcs, 5*\rcs)}{0.06}{0}
    \harrow{\rcs}{(6*\rcs, 5*\rcs)}{(8*\rcs, 5*\rcs)}{0.06}{1}
    \harrow{\rcs}{(8*\rcs, 5*\rcs)}{(10*\rcs, 5*\rcs)}{0.06}{0}
    
    \draw[text=black, draw=white, fill=white] (7*\rcs, 3.5*\rcs) circle (5pt) node {\footnotesize $\psi^{1}$};
    \draw[text=black, draw=white, fill=white] (9.5*\rcs, 7*\rcs) circle (5pt) node {\footnotesize $\psi^{2}$};
    \draw[text=black, draw=white, fill=white] (7*\rcs, 8.5*\rcs) circle (5pt) node {\footnotesize $\psi^{3}$};
    \draw[text=black, draw=white, fill=white] (4.5*\rcs, 7*\rcs) circle (5pt) node {\footnotesize $\psi^{4}$};

    \draw[text=black, draw=white, fill=white] (6*\rcs, 5*\rcs) circle (5pt) node {\footnotesize $s$};
    \draw[text=black, draw=white, fill=white] (4.5*\rcs, 5.5*\rcs) circle (5pt) node {\footnotesize $\tau_{1}$};
    \draw[text=black, draw=white, fill=white] (7.0*\rcs, 4.5*\rcs) circle (5pt) node {\footnotesize $\tau_{2}$};
    \draw[text=black, draw=white, fill=white] (9.0*\rcs, 5.5*\rcs) circle (5pt) node {\footnotesize $\tau_{3}$};





\end{tikzpicture}
    \end{minipage}
    \label{type_B_ribbon_with_face}
\end{equation}

For the above equation \eqref{type_B_ribbon_with_face}, $\rho^{B} = \tau_{1} \cup \tau_{2} \cup \tau_{3}$ we compute 
\begin{align}
    B^{b}(s)F^{a\tens h}\(\rho^{B}\) &= B^{b}(s) \left[ \rop[2]{\cop{a}{2}\cop{h}{2}S\cop{h}{4}}{\cop{h}{5}}\(\psi^1\) \tens \rop[3]{\cop{a}{1}\cop{h}{1}S\cop{h}{6}}{\cop{h}{7}}\(\psi^2\) \tens \psi^3 \tens \rop[1]{\cop{a}{3}}{\cop{h}{3}}\(\psi^4\) \right] \nonumber\\
    & =  B^{b}(s) \left[\epsilon \(\cop{h}{5}\) T^{\cop{a}{2}\cop{h}{2}S\cop{h}{4}}_{-}\(\psi^1\) \tens \eps\(\cop{a}{1}\cop{h}{1}S\cop{h}{6}\) L^{\cop{h}{7}}_{+}(\psi^2)\tens \psi^3\tens\eps(\cop{a}{3}) L^{\cop{h}{3}}_{-}(\psi^4)\right] \nonumber\\
    &= B^{b}(s)\left[ \eps\(\cop{h}{5}\)\< \cop{a}{2}\cop{h}{2}S\cop{h}{4}, \cop{\psi^1}{2} \>\cop{\psi^1}{1} \tens \eps\(\cop{a}{1}\cop{h}{1}S\cop{h}{6}\)\< \cop{h}{7}, S\cop{\psi^2}{1}\cop{\psi^2}{3} \>\cop{\psi^2}{2} \right. \nonumber \\
    &\qquad \left. \tens \psi^3 \tens \eps\(\cop{a}{3}\)\<\cop{h}{3}, \cop{\psi^4}{3}S^{-1}\cop{\psi^4}{1}\> \cop{\psi^4}{2} \right] \nonumber\\
    &= \eps\(\cop{h}{5}\)\< \cop{a}{2}\cop{h}{2}S\cop{h}{4}, \cop{\psi^{1}}{2} \>\TM{\cop{b}{4}}\(\cop{\psi^1}{1}\) \tens  \eps\(\cop{a}{1}\cop{h}{1}S\cop{h}{6}\)\< \cop{h}{7}, S\cop{\psi^2}{1}\cop{\psi^2}{3} \>\TM{\cop{b}{3}}\(\cop{\psi^2}{2}\)  \nonumber \\
    &\qquad \tens \TM{\cop{b}{2}}\(\psi^3\) \tens \eps(\cop{a}{3}\< \cop{h}{3}, \cop{\psi^4}{3}S^{-1}\cop{\psi^4}{1}\>\TM{\cop{b}{1}}\(\cop{\psi^4}{2}\) \nonumber\\
    &= \eps\(\cop{h}{5}\)\< \cop{a}{2}\cop{h}{2}S\cop{h}{4}, \cop{\psi^{1}}{2} \> \<\cop{b}{4}, \copp{\psi^1}{1}{2}\>\copp{\psi^1}{1}{1} \nonumber\\
    &\qquad \tens \eps\(\cop{a}{1}\cop{h}{1}S\cop{h}{6}\)\<\cop{h}{7}, S\cop{\psi^2}{1}\cop{\psi^2}{3}\> \<\cop{b}{3}, \copp{\psi^2}{2}{2}\> \copp{\psi^2}{2}{1}   \nonumber \\
    &\qquad \tens \<\cop{b}{2}, \cop{\psi^3}{2}\>\cop{\psi^3}{1} \tens \eps\(\cop{a}{3}\)\<\cop{h}{3}, \cop{\psi^4}{3}S^{-1}\cop{\psi^4}{1}\> \<\cop{b}{1}, \copp{\psi^4}{2}{2}\>\copp{\psi^4}{2}{1} \nonumber\\
    &= \eps\(\cop{h}{5}\)\< \cop{a}{2}\cop{h}{2}S\cop{h}{4}, \cop{\psi^1}{3} \> \<\cop{b}{4}, \cop{\psi^1}{2}\>\cop{\psi^1}{1} \nonumber\\
    &\qquad \tens \eps\(\cop{a}{1}\cop{h}{1}S\cop{h}{6}\)\<\cop{h}{7}, S\cop{\psi^2}{1}\cop{\psi^2}{4}\> \<\cop{b}{3}, \cop{\psi^2}{3}\> \cop{\psi^2}{2}   \nonumber \\
    &\qquad \tens \<\cop{b}{2}, \cop{\psi^3}{2}\>\cop{\psi^3}{1} \tens \eps\(\cop{a}{3}\)\<\cop{h}{3}, \cop{\psi^4}{4}S^{-1}\cop{\psi^4}{1}\> \<\cop{b}{1}, \cop{\psi^4}{3}\>\cop{\psi^4}{2} \nonumber\\
    &= \<b, \cop{\psi^1}{2}\cop{\psi^2}{3}\cop{\psi^3}{2}\cop{\psi^4}{3}\> \eps\(\cop{h}{5}\)\< \cop{a}{2}\cop{h}{2}S\cop{h}{4}, \cop{\psi^1}{3} \> \cop{\psi^1}{1} \nonumber\\
    &\qquad \tens \eps\(\cop{a}{1}\cop{h}{1}S\cop{h}{6}\)\<\cop{h}{7}, S\cop{\psi^2}{1}\cop{\psi^2}{4}\> \cop{\psi^2}{2} 
    \tens \cop{\psi^3}{1} \tens \eps\(\cop{a}{3}\)\<\cop{h}{3}, \cop{\psi^4}{4}S^{-1}\cop{\psi^4}{1}\> \cop{\psi^4}{2} \nonumber \\
    &= \eps\(\cop{h}{5}\)\< \cop{a}{2}\cop{h}{2}S\cop{h}{4}, \cop{\psi^1}{2} \> \cop{\psi^1}{1}
    \tens \eps\(\cop{a}{1}\cop{h}{1}S\cop{h}{6}\)\<\cop{h}{7}, S\cop{\psi^2}{1}\cop{\psi^2}{3}\> \cop{\psi^2}{2} \nonumber\\
    &\qquad \tens \psi^3 \tens \eps\(\cop{a}{3}\)\<\cop{h}{3}, \cop{\psi^4}{3}S^{-1}\cop{\psi^4}{1}\> \cop{\psi^4}{2}
\end{align}
In the first equation, we employ the gluing expression \eqref{gluing-operators} to the three ribbons. In the second equation, we spelt out the individual ribbon operators expression for each ribbon (i.e., $\tau_1$, $\tau_2$, $\tau_3$). For the fourth equation, the $B(s)$ operator is applied to the square lattice, with edges labeled $\psi\o^{1}$, $\psi\t^2$, $\psi^3$ and $\psi\t^4$. In the seventh equation we applied the fact that if $b$ is the Haar element of $H^{cop}$, then $\<b, \psi\> = \eps\(\psi\)$ and then simplified to get the last equaion.

For the same equation \eqref{type_B_ribbon_with_face}, we compute the following
\begin{align}
    F^{a\tens h}\(\rho^{B}\)B^{b}(s)\ket{\Sigma} &= \rop{a}{h}\left[ \TM{\cop{b}{4}}\(\psi^1\) \tens \TM{\cop{b}{3}}\(\psi^2\) \tens \TM{\cop{b}{2}}\(\psi^3\)\ \tens \TM{\cop{b}{1}}\(\psi^4\) \right] \nonumber\\
    &= \rop{a}{h}\left[ \<\cop{b}{4}, \cop{\psi^1}{2}\>\cop{\psi^1}{1} \tens \<\cop{b}{3}, \cop{\psi^2}{2}\>\cop{\psi^2}{1} \tens \<\cop{b}{2}, \cop{\psi^3}{2}\>\cop{\psi^3}{1} \tens \<\cop{b}{1}, \cop{\psi^4}{2}\>\cop{\psi^4}{1} \right] \nonumber\\
    &= \<\cop{b}{4}, \cop{\psi^1}{2}\>\rop[2]{\cop{a}{2}\cop{h}{2}S\cop{h}{4}}{\cop{h}{5}}\(\cop{\psi^1}{1}\) \tens \<\cop{b}{3}, \cop{\psi^2}{2}\> \rop[3]{\cop{a}{1}\cop{h}{1}S\cop{h}{6}}{\cop{h}{7}}\(\cop{\psi^2}{1}\) \nonumber \\
    &\qquad \tens \<\cop{b}{2}, \cop{\psi^3}{2}\>\cop{\psi^3}{1} \tens \<\cop{b}{1}, \cop{\psi^4}{2}\>\rop[1]{\cop{a}{3}}{\cop{h}{3}}\(\cop{\psi^4}{1}\) \nonumber\\
    &= \<\cop{b}{4}, \cop{\psi^1}{2}\> \eps\(\cop{h}{5}\) T^{\cop{a}{2}\cop{h}{2}S\cop{h}{4}}_{-}\(\cop{\psi^1}{1}\) \tens \<\cop{b}{3}, \cop{\psi^2}{2}\> \eps\(\cop{a}{1}\cop{h}{1}S\cop{h}{6}\) L^{\cop{h}{7}}_{+}\(\cop{\psi^2}{1}\)\nonumber \\
    &\qquad \tens \<\cop{b}{2}, \cop{\psi^3}{2}\>\cop{\psi^3}{1} \tens \<\cop{b}{1}, \cop{\psi^4}{2}\> \eps\(\cop{a}{3}\)\ L^{\cop{h}{3}}_{-}\(\cop{\psi^4}{1}\) \nonumber \\
    &= \eps\(\cop{h}{5}\)\<\cop{a}{2}\cop{h}{2}S\cop{h}{4}, \copp{\psi^1}{1}{2}\>\<\cop{b}{4}, \cop{\psi^1}{2}\>\copp{\psi^1}{1}{1} \nonumber \\
    &\qquad \tens \eps\(\cop{a}{1}\cop{h}{1}S\cop{h}{6}\)\<\cop{h}{7},  S\copp{\psi^2}{1}{1}\copp{\psi^2}{1}{3}\>\<\cop{b}{3}, \cop{\psi^2}{2}\>\copp{\psi^2}{1}{2} \nonumber \\
    &\qquad \tens \<\cop{b}{2}, \cop{\psi^3}{2}\>\cop{\psi^3}{1} \tens \eps\(\cop{a}{3}\)\<\cop{h}{3}, \copp{\psi^4}{1}{3}S^{-1}\copp{\psi^4}{1}{1} \> \<\cop{b}{1}, \cop{\psi^4}{2}\>\copp{\psi^4}{1}{2}\nonumber \\
    &= \eps\(\cop{h}{5}\)\<\cop{a}{2}\cop{h}{2}S\cop{h}{4}, \cop{\psi^1}{2}\>\<\cop{b}{4}, \cop{\psi^1}{3}\>\cop{\psi^1}{1} \nonumber \\
    &\qquad \tens \eps\(\cop{a}{1}\cop{h}{1}S\cop{h}{6}\)\<\cop{h}{7},  S\cop{\psi^2}{1}\cop{\psi^2}{3}\>\<\cop{b}{3}, \cop{\psi^2}{4}\>\cop{\psi^2}{2} \nonumber \\
    &\qquad \tens \<\cop{b}{2}, \cop{\psi^3}{2}\>\cop{\psi^3}{1} \tens \eps\(\cop{a}{3}\)\<\cop{h}{3}, \cop{\psi^4}{3}S^{-1}\cop{\psi^4}{1} \> \<\cop{b}{1}, \cop{\psi^4}{4}\>\cop{\psi^4}{2} \nonumber \\
    &= \<b, \cop{\psi^1}{3}\cop{\psi^2}{4}\cop{\psi^3}{2}\cop{\psi^4}{4} \> \eps\(\cop{h}{5}\)\<\cop{a}{2}\cop{h}{2}S\cop{h}{4}, \cop{\psi^1}{2}\>\cop{\psi^1}{1} \nonumber \\
    &\qquad \tens \eps\(\cop{a}{1}\cop{h}{1}S\cop{h}{6}\)\<\cop{h}{7},  S\cop{\psi^2}{1}\cop{\psi^2}{3}\>\cop{\psi^2}{2} 
    \tens \cop{\psi^3}{1} \tens \eps\(\cop{a}{3}\)\<\cop{h}{3}, \cop{\psi^4}{3}S^{-1}\cop{\psi^4}{1} \> \cop{\psi^4}{2} \nonumber \\
    &= \eps\(\cop{h}{5}\)\<\cop{a}{2}\cop{h}{2}S\cop{h}{4}, \cop{\psi^1}{2}\>\cop{\psi^1}{1} 
    \tens \eps\(\cop{a}{1}\cop{h}{1}S\cop{h}{6}\)\<\cop{h}{7},  S\cop{\psi^2}{1}\cop{\psi^2}{3}\>\cop{\psi^2}{2} \nonumber \\
    &\qquad \tens \psi^3 \tens \eps\(\cop{a}{3}\)\<\cop{h}{3}, \cop{\psi^4}{3}S^{-1}\cop{\psi^4}{1} \> \cop{\psi^4}{2}
\end{align}
For the first equation, the $B(s)$ operator is applied to the square lattice with edges labeled $\psi^{1}$, $\psi^2$, $\psi^3$ and $\psi^4$. We utilize the gluing expression \eqref{gluing-operators} in the third equation. In the fourth equation, we spelt out the individual ribbon operators expression for each ribbon (i.e., $\tau_1$, $\tau_2$, $\tau_3$). 
In the final step of the calculation, we again applied the fact that if $b$ is the Haar element of $H^{cop}$, then $\<b, \psi\> = \eps\(\psi\)$ and then simplified to get the last equaion.

\begin{equation}
    \begin{minipage}[c]{0.05\textwidth}
        $\ket{\Sigma}$
    \end{minipage}
    \begin{minipage}[c]{0.05\textwidth}
        $=$
    \end{minipage}
    \begin{minipage}[c]{0.5\textwidth}
        \begin{tikzpicture}[scale = 0.8]

    \begin{scope}[thick, decoration={markings, mark=at position 0.4 with {\arrow{latex}}}]
        \draw[postaction={decorate}] (3*\rcs, 4*\rcs) -- (7*\rcs, 4*\rcs);
        \draw[postaction={decorate}] (7*\rcs, 0*\rcs) -- (7*\rcs, 4*\rcs);
        \draw[postaction={decorate}] (11*\rcs, 4*\rcs) -- (7*\rcs, 4*\rcs);
        \draw[postaction={decorate}] (7*\rcs, 8*\rcs) -- (7*\rcs, 4*\rcs);
    \end{scope}

    \begin{scope}[thin, red, dashed, decoration={markings, mark=at position 0.3 with {\arrow{latex}}}]
        \draw[postaction={decorate}] (5*\rcs, 6*\rcs) -- (9*\rcs, 6*\rcs);
    \end{scope}

    \draw[dotted, very thick] (3*\rcs, 4*\rcs) -- (5*\rcs, 6*\rcs);
    \draw[dotted, very thick] (7*\rcs, 4*\rcs) -- (5*\rcs, 6*\rcs);
    \draw[dotted, very thick] (7*\rcs, 4*\rcs) -- (9*\rcs, 6*\rcs);
    \draw[dotted, very thick] (9*\rcs, 6*\rcs) -- (11*\rcs, 4*\rcs);
    

    \harrow{\rcs}{(4*\rcs, 5*\rcs)}{(6*\rcs, 5*\rcs)}{0.06}{1}
    \harrow{\rcs}{(6*\rcs, 5*\rcs)}{(8*\rcs, 5*\rcs)}{0.06}{0}
    \harrow{\rcs}{(8*\rcs, 5*\rcs)}{(10*\rcs, 5*\rcs)}{0.06}{1}
    
    \draw[text=black, draw=white, fill=white] (5*\rcs, 3.5*\rcs) circle (5pt) node {\footnotesize $\psi^{1}$};
    \draw[text=black, draw=white, fill=white] (7.5*\rcs, 2*\rcs) circle (5pt) node {\footnotesize $\psi^{2}$};
    \draw[text=black, draw=white, fill=white] (9*\rcs, 3.5*\rcs) circle (5pt) node {\footnotesize $\psi^{3}$};
    \draw[text=black, draw=white, fill=white] (6.5*\rcs, 7*\rcs) circle (5pt) node {\footnotesize $\psi^{4}$};

    \draw[text=black, draw=white, fill=white] (6*\rcs, 5*\rcs) circle (5pt) node {\footnotesize $s$};
    \draw[text=black, draw=white, fill=white] (5*\rcs, 4.5*\rcs) circle (5pt) node {\footnotesize $\tau_{1}$};
    \draw[text=black, draw=white, fill=white] (7.0*\rcs, 5.5*\rcs) circle (5pt) node {\footnotesize $\tau_{2}$};
    \draw[text=black, draw=white, fill=white] (9.0*\rcs, 4.5*\rcs) circle (5pt) node {\footnotesize $\tau_{3}$};





\end{tikzpicture}
    \end{minipage}
    \label{type_B_ribbon_with_vertex}
\end{equation}
For the above equation \eqref{type_B_ribbon_with_vertex}, we compute 
\begin{align}
    A^{g}(s)F^{a\tens h}\(\rho^{B}\)\ket{\Sigma} &= A^{g}(s)\left[ \rop[1]{\cop{a}{3}}{\cop{h}{3}}\(\psi^1\) \tens \psi^2 \tens \rop[3]{\cop{a}{1}\cop{h}{1}S\cop{h}{6}}{\cop{h}{7}}\(\psi^3\) \tens \rop[2]{\cop{a}{2}\cop{h}{2}S\cop{h}{4}}{\cop{h}{5}}\(\psi^4\) \right] \nonumber \\
    &= A^{g}(s)\left[\eps\(\cop{h}{3}\)\ T^{\cop{a}{3}}_{-}\(\psi^1\) \tens \psi^2 \tens \eps\(\cop{h}{7}\) T^{\cop{a}{1}\cop{h}{1}S\cop{h}{6}}_{-} \(\psi^3\) \tens \eps\(\cop{a}{2}\cop{h}{2}S\cop{h}{4}\)\ L^{\cop{h}{5}}_{-}\(\psi^4\) \right] \nonumber \\
    &= A^{g}(s) \left[\eps\(\cop{h}{3}\)\<\cop{a}{3}, \cop{\psi^1}{2} \>\cop{\psi^1}{1} \tens \psi^2 \tens \eps\(\cop{h}{7}\)\<\cop{a}{1}\cop{h}{1}S\cop{h}{6}, \cop{\psi^3}{2}\>\cop{\psi^3}{1} \right. \nonumber \\
    &\qquad \tens \left. \eps\(\cop{a}{2}\cop{h}{2}S\cop{h}{4}\)\<\cop{h}{5}, \cop{\psi^4}{3}S^{-1}\cop{\psi^4}{1}\>\cop{\psi^4}{2} \right] \nonumber \\
    &= \eps\(\cop{h}{3}\)\<\cop{a}{3}, \cop{\psi^1}{2} \>\PP{1}{\cop{g}{4}}\(\cop{\psi^1}{1}\) \tens \PP{\cop{g}{3}S\cop{g}{5}}{\cop{g}{6}}\(\psi^2\) \nonumber \\
    &\qquad \tens \eps\(\cop{h}{7}\)\<\cop{a}{1}\cop{h}{1}S\cop{h}{6}, \cop{\psi^3}{2}\> \PP{\cop{g}{2}S\cop{g}{7}}{\cop{g}{8}}\(\cop{\psi^3}{1}\) \nonumber \\
    &\qquad \tens \eps\(\cop{a}{2}\cop{h}{2}S\cop{h}{4}\)\<\cop{h}{5}, \cop{\psi^4}{3}S^{-1}\cop{\psi^4}{1}\>\PP{\cop{g}{1}S\cop{g}{9}}{\cop{g}{10}} \(\cop{\psi^4}{2}\) \nonumber \\
    &= \eps\(\cop{h}{3}\)\<\cop{a}{3}, \cop{\psi^1}{2} \> \<S\copp{g}{4}{1}, \copp{\psi^1}{1}{1}\>\<\copp{g}{4}{2}, \copp{\psi^1}{1}{3}\>\copp{\psi^1}{1}{2} \nonumber\\
    &\qquad \tens \<S\copp{g}{6}{1}S\(\cop{g}{3}S\cop{g}{5}\), \cop{\psi^2}{1}\>\<\copp{g}{6}{2}, \cop{\psi^2}{3}\> \cop{\psi^2}{2} \nonumber \\
    &\qquad \tens \eps\(\cop{h}{7}\)\<\cop{a}{1}\cop{h}{1}S\cop{h}{6}, \cop{\psi^3}{2}\> \<S\copp{g}{8}{1}S\(\cop{g}{2}S\cop{g}{7}\), \copp{\psi^3}{1}{1}\> \<\copp{g}{8}{2}, \copp{\psi^3}{1}{3}\> \copp{\psi^3}{1}{2} \nonumber \\ 
    &\qquad \tens \eps\(\cop{a}{2}\cop{h}{2}S\cop{h}{4}\)\<\cop{h}{5}, \cop{\psi^4}{3}S^{-1}\cop{\psi^4}{1}\>\<S\copp{g}{10}{1}S\(\cop{g}{1}S\cop{g}{9}\), \copp{\psi^4}{2}{1} \>\<\copp{g}{10}{2}, \copp{\psi^4}{2}{3}\> \copp{\psi^4}{2}{2} \nonumber \\
    &= \eps\(\cop{h}{3}\)\<\cop{a}{3}, \cop{\psi^1}{4} \> \<S\cop{g}{4}, \cop{\psi^1}{1}\>\<\cop{g}{5}, \cop{\psi^1}{3}\>\cop{\psi^1}{2} \nonumber \\
    &\qquad \tens \<S\cop{g}{7}S\(\cop{g}{3}S\cop{g}{6}\), \cop{\psi^2}{1}\>\<\cop{g}{8}, \cop{\psi^2}{3}\> \cop{\psi^2}{2} \nonumber \\
    &\qquad \tens \eps\(\cop{h}{7}\)\<\cop{a}{1}\cop{h}{1}S\cop{h}{6}, \cop{\psi^3}{4}\> \<S\cop{g}{10}S\(\cop{g}{2}S\cop{g}{9}\), \cop{\psi^3}{1}\> \<\cop{g}{11}, \cop{\psi^3}{3}\> \cop{\psi^3}{2} \nonumber \\ 
    &\qquad \tens \eps\(\cop{a}{2}\cop{h}{2}S\cop{h}{4}\)\<\cop{h}{5}, \cop{\psi^4}{5}S^{-1}\cop{\psi^4}{1}\>\<S\cop{g}{13}S\(\cop{g}{1}S\cop{g}{12}\), \cop{\psi^4}{2} \>\<\cop{g}{14}, \cop{\psi^4}{4}\> \cop{\psi^4}{3} \nonumber \\
    & = \eps\(\cop{h}{3}\)\<\cop{a}{3}, \cop{\psi^1}{4}\> \<\cop{g}{4}, S\cop{\psi^1}{1}\cop{\psi^1}{3} \> \cop{\psi^1}{2} \tens \<S\cop{g}{3}, \cop{\psi^2}{1}\> \<\cop{g}{5}, \cop{\psi^2}{3}\> \cop{\psi^2}{2} \nonumber \\
    &\qquad \tens \eps\(\cop{h}{7}\) \<\cop{a}{1}\cop{h}{1}S\cop{h}{6}, \cop{\psi^3}{4}\> \<S\cop{g}{2}, \cop{\psi^3}{1}\> \<\cop{g}{6}, \cop{\psi^3}{3} \> \cop{\psi^3}{2} \nonumber \\
    &\qquad \tens \eps\(\cop{a}{2}\cop{h}{2}S\cop{h}{6}\)\<\cop{h}{5}, \cop{\psi^4}{5}S^{-1}\cop{\psi^4}{1}\> \<S\cop{g}{1}, \cop{\psi^4}{2}\> \<\cop{g}{7}, \cop{\psi^4}{4}\> \cop{\psi^4}{3} \nonumber \\
    & = \<g, S\cop{\psi^4}{2}S\cop{\psi^3}{1}S\cop{\psi^2}{1}S\cop{\psi^1}{1}\cop{\psi^1}{3}\cop{\psi^2}{3}\cop{\psi^3}{3}\cop{\psi^4}{4} \> \eps\(\cop{h}{3}\)\<\cop{a}{3}, \cop{\psi^1}{4}\>\cop{\psi^1}{2} \tens \cop{\psi^2}{2} \nonumber \\
    &\qquad \tens \eps\(\cop{h}{7}\) \<\cop{a}{1}\cop{h}{1}S\cop{h}{6}, \cop{\psi^3}{4}\> \cop{\psi^3}{2} \tens \eps\(\cop{a}{2}\cop{h}{2}S\cop{h}{6}\)\<\cop{h}{5}, \cop{\psi^4}{5}S^{-1}\cop{\psi^4}{1}\> \cop{\psi^4}{3} \nonumber \\
    & = \eps\(\cop{h}{3}\)\<\cop{a}{3}, \cop{\psi^1}{2}\>\cop{\psi^1}{1} \tens \psi^2  \tens \eps\(\cop{h}{7}\) \<\cop{a}{1}\cop{h}{1}S\cop{h}{6}, \cop{\psi^3}{2}\> \cop{\psi^3}{1} \nonumber \\
    &\qquad \tens \eps\(\cop{a}{2}\cop{h}{2}S\cop{h}{6}\)\<\cop{h}{5}, \cop{\psi^4}{3}S^{-1}\cop{\psi^4}{1}\> \cop{\psi^4}{2}
\end{align}
In the first equation, we employ the gluing expression \eqref{gluing-operators} to the three ribbons. In the second equation, we spelt out the individual ribbon operators expression for each ribbon (i.e., $\tau_1$, $\tau_2$, $\tau_3$). For the fourth equation, the $A(s)$ operator is applied to the square lattice, with edges labeled $\psi\o^{1}$, $\psi^2$, $\psi\o^3$ and $\psi\t^4$.

\begin{align}
    F^{a\tens h}\(\rho^{B}\)A^{g}(s)\ket{\Sigma} &= \rop{a}{h}\left[\PP{1}{\cop{g}{4}}\(\psi^1\) \tens \PP{\cop{g}{3}S\cop{g}{5}}{\cop{g}{6}}\(\psi^2\) \tens \PP{\cop{g}{2}S\cop{g}{7}}{\cop{g}{8}}\(\psi^3\) \tens \PP{\cop{g}{1}S\cop{g}{9}}{\cop{g}{10}}\(\psi^4\) \right] \nonumber \\
    &= \rop{a}{h}\left[\<S\copp{g}{4}{1}, \cop{\psi^1}{1}\>\<\copp{g}{4}{2}, \cop{\psi^1}{3}\>\cop{\psi^1}{2} \tens \<S\copp{g}{6}{1}S\(\cop{g}{3}S\cop{g}{5}\), \cop{\psi^2}{1}\>\<\copp{g}{6}{2}, \cop{\psi^2}{3}\>\cop{\psi^2}{2} \right. \nonumber \\
    &\qquad \tens \left. \<S\copp{g}{8}{1}S\(\cop{g}{2}S\cop{g}{7}\), \cop{\psi^3}{1}\>\<\copp{g}{8}{2}, \cop{\psi^3}{3}\>\cop{\psi^3}{2} \right. \nonumber \\
    &\qquad \tens \left. \<S\copp{g}{10}{1}S\(\cop{g}{1}S\cop{g}{9}\), \cop{\psi^4}{1}\>\<\copp{g}{10}{2}, \cop{\psi^4}{3}\>\cop{\psi^4}{2} \right] \nonumber \\
    &= \<S\cop{g}{4}, \cop{\psi^1}{1}\>\<\cop{g}{5}, \cop{\psi^1}{3}\>\rop[1]{\cop{a}{3}}{\cop{h}{3}}\(\cop{\psi^1}{2}\) \nonumber \\
    &\qquad \tens \<S\cop{g}{7}S\(\cop{g}{3}S\cop{g}{6}\), \cop{\psi^2}{1}\>\<\cop{g}{8}, \cop{\psi^2}{3}\>\cop{\psi^2}{2}\nonumber \\
    &\qquad \tens \<S\cop{g}{10}S\(\cop{g}{2}S\cop{g}{9}\), \cop{\psi^3}{1}\>\<\cop{g}{11}, \cop{\psi^3}{3}\>\rop[3]{\cop{a}{1}\cop{h}{1}S\cop{h}{6}}{\cop{h}{7}}\(\cop{\psi^3}{2}\) \nonumber \\
    &\qquad \tens \<S\cop{g}{13}S\(\cop{g}{1}S\cop{g}{12}\), \cop{\psi^4}{1}\>\<\cop{g}{14}, \cop{\psi^4}{3}\>\rop[3]{\cop{a}{2}\cop{h}{2}S\cop{h}{4}}{\cop{h}{5}}\(\cop{\psi^4}{2}\) \nonumber \\
     &= \<S\cop{g}{4}, \cop{\psi^1}{1}\>\<\cop{g}{5}, \cop{\psi^1}{3}\>\eps\(\cop{h}{3}\) T^{\cop{a}{3}}_{-}\(\cop{\psi^1}{2}\) \nonumber \\
    &\qquad \tens \<S\cop{g}{7}S\(\cop{g}{3}S\cop{g}{6}\), \cop{\psi^2}{1}\>\<\cop{g}{8}, \cop{\psi^2}{3}\>\cop{\psi^2}{2}\nonumber \\
    &\qquad \tens \<S\cop{g}{10}S\(\cop{g}{2}S\cop{g}{9}\), \cop{\psi^3}{1}\>\<\cop{g}{11}, \cop{\psi^3}{3}\>\eps\(\cop{h}{7}\)\ T^{\cop{a}{1}\cop{h}{1}S\cop{h}{6}}_{-}\(\cop{\psi^3}{2}\) \nonumber \\
    &\qquad \tens \<S\cop{g}{13}S\(\cop{g}{1}S\cop{g}{12}\), \cop{\psi^4}{1}\>\<\cop{g}{14}, \cop{\psi^4}{3}\> \eps\(\cop{a}{2}\cop{h}{2}S\cop{h}{4}\)\ L^{\cop{h}{5}}_{+}\(\cop{\psi^4}{2}\) \nonumber \\
    &= \eps\(\cop{h}{3}\)\<\cop{a}{3}, \copp{\psi^1}{2}{2}\>\<S\cop{g}{4}, \cop{\psi^1}{1}\>\<\cop{g}{5}, \cop{\psi^1}{3}\>\copp{\psi^1}{2}{1} \nonumber \\
    &\qquad \tens \<S\cop{g}{7}S\(\cop{g}{3}S\cop{g}{6}\), \cop{\psi^2}{1}\>\<\cop{g}{8}, \cop{\psi^2}{3}\>\cop{\psi^2}{2} \nonumber \\
    &\qquad \tens \eps\(\cop{h}{7}\)\<\cop{a}{1}\cop{h}{1}S\cop{h}{6}, \copp{\psi^3}{2}{2}\>\<S\cop{g}{10}S\(\cop{g}{2}S\cop{g}{9}\), \cop{\psi^3}{1}\>\<\cop{g}{11}, \cop{\psi^3}{3}\>\copp{\psi^3}{2}{1} \nonumber \\
    &\qquad \tens \eps\(\cop{a}{2}\cop{h}{2}S\cop{h}{4}\)\<\cop{h}{5}, S\copp{\psi^4}{2}{1}\copp{\psi^4}{2}{3}\>\<S\cop{g}{13}S\(\cop{g}{1}S\cop{g}{12}\), \cop{\psi^4}{1}\>\<\cop{g}{14}, \cop{\psi^4}{3}\>\copp{\psi^4}{2}{2} \nonumber \\
    &= \eps\(\cop{h}{3}\)\<\cop{a}{3}, \cop{\psi^1}{3}\>\<S\cop{g}{4}, \cop{\psi^1}{1}\>\<\cop{g}{5}, \cop{\psi^1}{4}\>\cop{\psi^1}{2} \nonumber \\
    &\qquad \tens \<S\cop{g}{7}S\(\cop{g}{3}S\cop{g}{6}\), \cop{\psi^2}{1}\>\<\cop{g}{8}, \cop{\psi^2}{3}\>\cop{\psi^2}{2} \nonumber \\
    &\qquad \tens \eps\(\cop{h}{7}\)\<\cop{a}{1}\cop{h}{1}S\cop{h}{6}, \cop{\psi^3}{3}\>\<S\cop{g}{10}S\(\cop{g}{2}S\cop{g}{9}\), \cop{\psi^3}{1}\>\<\cop{g}{11}, \cop{\psi^3}{4}\>\cop{\psi^3}{2} \nonumber \\
    &\qquad \tens \eps\(\cop{a}{2}\cop{h}{2}S\cop{h}{4}\)\<\cop{h}{5}, S\cop{\psi^4}{2}\cop{\psi^4}{4}\>\<S\cop{g}{13}S\(\cop{g}{1}S\cop{g}{12}\), \cop{\psi^4}{1}\>\<\cop{g}{14}, \cop{\psi^4}{5}\>\cop{\psi^4}{3} \nonumber\\
    &= \eps\(\cop{h}{3}\)\<\cop{a}{3}, \cop{\psi^1}{3}\> \<\cop{g}{4}, S\cop{\psi^1}{1}\cop{\psi^1}{4}\>\cop{\psi^1}{2} \tens \<S\cop{g}{3}, \cop{\psi^2}{1}\>\<\cop{g}{5}, \cop{\psi^2}{3}\> \cop{\psi^2}{2} \nonumber\\
    &\qquad \tens \eps\(\cop{h}{7}\)\<\cop{a}{1}\cop{h}{1}S\cop{h}{6}, \cop{\psi^3}{3}\>\<S\cop{g}{2}, \cop{\psi^3}{1}\>\<\cop{g}{6}, \cop{\psi^3}{4}\>\cop{\psi^3}{2} \nonumber \\
    &\qquad \tens \eps\(\cop{a}{2}\cop{h}{2}S\cop{h}{4}\)\< \cop{h}{5}, S\cop{\psi^4}{2}\cop{\psi^4}{4}\>\< S\cop{g}{1}, \cop{\psi^4}{1}\>\<\cop{g}{7}, \cop{\psi^4}{5}\>\cop{\psi^4}{3} \nonumber \\
    &= \<g, S\cop{\psi^4}{1}S\cop{\psi^3}{1}S\cop{\psi^2}{1}S\cop{\psi^1}{1}\cop{\psi^1}{4}\cop{\psi^2}{3}\cop{\psi^3}{4}\cop{\psi^4}{5}\> \eps\(\cop{h}{3}\)\<\cop{a}{3}, \cop{\psi^1}{3}\> \cop{\psi^1}{2}
    \tens \cop{\psi^2}{2} \nonumber \\
    & \qquad \tens \eps\(\cop{h}{7}\)\<\cop{a}{1}\cop{h}{1}S\cop{h}{6}, \cop{\psi^3}{3}\> \cop{\psi^3}{2} \tens  \eps\(\cop{a}{2}\cop{h}{2}S\cop{h}{4}\)\< \cop{h}{5}, S\cop{\psi^4}{2}\cop{\psi^4}{4}\> \cop{\psi^4}{3}\nonumber \\
    & = \eps\(\cop{h}{3}\)\<\cop{a}{3}, \cop{\psi^1}{2}\>\cop{\psi^1}{1} \tens \psi^2  \tens \eps\(\cop{h}{7}\) \<\cop{a}{1}\cop{h}{1}S\cop{h}{6}, \cop{\psi^3}{2}\> \cop{\psi^3}{1} \nonumber \\
    &\qquad \tens \eps\(\cop{a}{2}\cop{h}{2}S\cop{h}{6}\)\<\cop{h}{5}, \cop{\psi^4}{3}S^{-1}\cop{\psi^4}{1}\> \cop{\psi^4}{2}
\end{align}
For the first equation, the $A(s)$ operator is applied to the square lattice with edges labeled $\psi^{1}$, $\psi^2$, $\psi^3$ and $\psi^4$. We employ the gluing expression \eqref{gluing-operators} in the third equation. In the fourth equation, we spelt out the individual ribbon operators expression for each ribbon (i.e., $\tau_1$, $\tau_2$, $\tau_3$).

\begin{equation}
    \begin{minipage}[c]{0.05\textwidth}
        $\ket{\Sigma}$
    \end{minipage}
    \begin{minipage}[c]{0.05\textwidth}
        $=$
    \end{minipage}
    \begin{minipage}[c]{0.5\textwidth}
        \begin{tikzpicture}[scale = 0.8]

    \begin{scope}[thick, decoration={markings, mark=at position 0.7 with {\arrow{latex}}}]
        \draw[postaction={decorate}] (5*\rcs, 4*\rcs) -- (9*\rcs, 4*\rcs);
        \draw[postaction={decorate}] (9*\rcs, 4*\rcs) -- (9*\rcs, 8*\rcs);
        \draw[postaction={decorate}] (9*\rcs, 8*\rcs) -- (5*\rcs, 8*\rcs);
        \draw[postaction={decorate}] (5*\rcs, 8*\rcs) -- (5*\rcs, 4*\rcs);
    \end{scope}

    \begin{scope}[thin, red, dashed, decoration={markings, mark=at position 0.3 with {\arrow{latex}}}]
        \draw[postaction={decorate}] (3*\rcs, 6*\rcs) -- (7*\rcs, 6*\rcs);
        \draw[postaction={decorate}] (11*\rcs, 6*\rcs) -- (7*\rcs, 6*\rcs);
    \end{scope}

    \draw[dotted, very thick] (3*\rcs, 6*\rcs) -- (5*\rcs, 8*\rcs);
    \draw[dotted, very thick] (5*\rcs, 8*\rcs) -- (7*\rcs, 6*\rcs);
    \draw[dotted, very thick] (7*\rcs, 6*\rcs) -- (9*\rcs, 8*\rcs);
    \draw[dotted, very thick] (9*\rcs, 8*\rcs) -- (11*\rcs, 6*\rcs);
    

    \harrow{\rcs}{(4*\rcs, 7*\rcs)}{(6*\rcs, 7*\rcs)}{0.06}{1}
    \harrow{\rcs}{(6*\rcs, 7*\rcs)}{(8*\rcs, 7*\rcs)}{0.06}{0}
    \harrow{\rcs}{(8*\rcs, 7*\rcs)}{(10*\rcs, 7*\rcs)}{0.06}{1}

    \draw[text=black, draw=white, fill=white] (7*\rcs, 8.5*\rcs) circle (5pt) node {\footnotesize $\psi^{1}$};
    \draw[text=black, draw=white, fill=white] (4.5*\rcs, 5*\rcs) circle (5pt) node {\footnotesize $\psi^{2}$};
    \draw[text=black, draw=white, fill=white] (7*\rcs, 3.5*\rcs) circle (5pt) node {\footnotesize $\psi^{3}$};
    \draw[text=black, draw=white, fill=white] (9.5*\rcs, 5*\rcs) circle (5pt) node {\footnotesize $\psi^{4}$};

    \draw[text=black, draw=white, fill=white] (8*\rcs, 7*\rcs) circle (5pt) node {\footnotesize $s$};
    \draw[text=black, draw=white, fill=white] (5*\rcs, 6.5*\rcs) circle (5pt) node {\footnotesize $\tau_{1}$};
    \draw[text=black, draw=white, fill=white] (7.0*\rcs, 7.5*\rcs) circle (5pt) node {\footnotesize $\tau_{2}$};
    \draw[text=black, draw=white, fill=white] (9.0*\rcs, 6.5*\rcs) circle (5pt) node {\footnotesize $\tau_{3}$};





\end{tikzpicture}
    \end{minipage}
    \label{type_A_ribbon_with_face}
\end{equation}

\begin{align}
    B^{b}(s)F^{a\tens h}\(\rho^{A}\)\ket{\Sigma} &= B^{b}(s)\left[\rop[2]{\cop{a}{2}\cop{h}{2}S\cop{h}{4}}{\cop{h}{5}}\(\psi^1\) \tens \rop[1]{\cop{a}{3}}{\cop{h}{3}}\(\psi^2\) \tens \psi^3 \tens \rop[3]{\cop{a}{1}\cop{h}{1}S\cop{h}{6}}{\cop{h}{7}}\(\psi^4\)\right]\nonumber \\
    &= B^{b}(s)\left[ \eps\(\cop{h}{5}\)\ \TMt{\cop{a}{2}\cop{h}{2}S\cop{h}{4}}\(\psi^1\) \tens  \eps\(\cop{a}{3}\)\ \tilde{L}^{\cop{h}{3}}_{1} \(\psi^2\) \tens \psi^3 \tens \eps\(\cop{a}{1}\cop{h}{1}S\cop{h}{6}\)\tilde{L}^{\cop{h}{7}}_{+}\(\psi^4\)\right]\nonumber \\
    &= B^{b}(s)\left[ \eps\(\cop{h}{5}\)\<\Si\(\cop{a}{2}\cop{h}{2}S\cop{h}{4}\), \cop{\psi^1}{2}\>\cop{\psi^1}{1} \tens \eps\(\cop{a}{3}\)\<\cop{h}{3}, S\cop{\psi^2}{3}\cop{\psi^2}{1}\>\cop{\psi^2}{2} \right. \nonumber \\
    &\qquad \tens \left. \psi^3 \tens \eps\(\cop{a}{1}\cop{h}{1}S\cop{h}{6}\)\<\cop{h}{7}, \cop{\psi^4}{1}S^{-1}\cop{\psi^4}{3}\>\cop{\psi^4}{2} \right] \nonumber \\
    &= \eps\(\cop{h}{5}\)\<\Si\(\cop{a}{2}\cop{h}{2}S\cop{h}{4}\), \cop{\psi^1}{2}\>\TM{\cop{b}{4}}\(\cop{\psi^1}{1}\) \tens \eps\(\cop{a}{3}\)\<\cop{h}{3}, S\cop{\psi^2}{3}\cop{\psi^2}{1}\>\TM{\cop{b}{3}} \(\cop{\psi^2}{2}\) \nonumber \\
    &\qquad \tens \TM{\cop{b}{2}}\(\psi^3\) \tens \eps\(\cop{a}{1}\cop{h}{1}S\cop{h}{6}\)\<\cop{h}{7}, \cop{\psi^4}{1}S^{-1}\cop{\psi^4}{3}\>\TM{\cop{b}{1}}\(\cop{\psi^4}{2}\) \nonumber \\
    &= \eps\(\cop{h}{5}\)\<\Si\(\cop{a}{2}\cop{h}{2}S\cop{h}{4}\), \cop{\psi^1}{2}\>\<\cop{b}{4},  \copp{\psi^1}{2}{1}\> \copp{\psi^1}{1}{1} \nonumber\\
    &\qquad \tens \eps\(\cop{a}{3}\)\<\cop{h}{3}, S\cop{\psi^2}{3}\cop{\psi^2}{1}\>\<\cop{b}{3},  \copp{\psi^2}{2}{2}\> \copp{\psi^2}{2}{1} \nonumber \\
    &\qquad \tens \<\cop{b}{2},  \cop{\psi^3}{2}\> \cop{\psi^3}{1} \tens \eps\(\cop{a}{1}\cop{h}{1}S\cop{h}{6}\)\<\cop{h}{7}, \cop{\psi^4}{1}S^{-1}\cop{\psi^4}{3}\>\<\cop{b}{1},  \copp{\psi^4}{2}{2}\> \copp{\psi^4}{2}{1}\nonumber \\
    &= \eps\(\cop{h}{5}\)\<\Si\(\cop{a}{2}\cop{h}{2}S\cop{h}{4}\), \cop{\psi^1}{3}\>\<\cop{b}{4},  \cop{\psi^1}{2}\> \cop{\psi^1}{1} \nonumber \\
    &\qquad \tens \eps\(\cop{a}{3}\)\<\cop{h}{3}, S\cop{\psi^2}{4}\cop{\psi^2}{1}\>\<\cop{b}{3},  \cop{\psi^2}{3}\> \cop{\psi^2}{2} \nonumber \\
    &\qquad \tens \<\cop{b}{2},  \cop{\psi^3}{2}\> \cop{\psi^3}{1} \tens \eps\(\cop{a}{1}\cop{h}{1}S\cop{h}{6}\)\<\cop{h}{7}, \cop{\psi^4}{1}S^{-1}\cop{\psi^4}{4}\>\<\cop{b}{1},  \cop{\psi^4}{3}\> \cop{\psi^4}{2} \nonumber \\
    &= \<b, \cop{\psi^1}{2}\cop{\psi^2}{3}\cop{\psi^3}{2}\cop{\psi^4}{3}\> \eps\(\cop{h}{5}\)\<\Si\(\cop{a}{2}\cop{h}{2}S\cop{h}{4}\), \cop{\psi^1}{3}\> \cop{\psi^1}{1} \nonumber \\
    &\qquad \tens \eps\(\cop{a}{3}\)\<\cop{h}{3}, S\cop{\psi^2}{4}\cop{\psi^2}{1}\> \cop{\psi^2}{2} \tens  \cop{\psi^3}{1} \tens \eps\(\cop{a}{1}\cop{h}{1}S\cop{h}{6}\)\<\cop{h}{7}, \cop{\psi^4}{1}S^{-1}\cop{\psi^4}{4}\> \cop{\psi^4}{2} \nonumber \\
    &= \eps\(\cop{h}{5}\)\<\Si\(\cop{a}{2}\cop{h}{2}S\cop{h}{4}\), \cop{\psi^1}{2}\> \cop{\psi^1}{1} \tens \eps\(\cop{a}{3}\)\<\cop{h}{3}, S\cop{\psi^2}{3}\cop{\psi^2}{1}\> \cop{\psi^2}{2} \nonumber \\
    &\qquad  \tens \psi^3 \tens \eps\(\cop{a}{1}\cop{h}{1}S\cop{h}{6}\)\<\cop{h}{7}, \cop{\psi^4}{1}S^{-1}\cop{\psi^4}{3}\> \cop{\psi^4}{2}
\end{align}

\begin{align}
    F^{a\tens h}\(\rho^{A}\)B^{b}(s)\ket{\Sigma} &= \rop{a}{h}\left[\TM{\cop{b}{4}}\(\psi^1\) \tens \TM{\cop{b}{3}}\(\psi^2\) \tens \TM{\cop{b}{2}}\(\psi^3\) \tens \TM{\cop{b}{1}}\(\psi^4\)\right] \nonumber \\
    &= \rop{a}{h}\left[\<\cop{b}{4}, \cop{\psi^1}{2}\>\cop{\psi^1}{1} \tens \<\cop{b}{3}, \cop{\psi^2}{2}\>\cop{\psi^2}{1}\tens \<\cop{b}{2}, \cop{\psi^3}{2}\>\cop{\psi^3}{1}\tens \<\cop{b}{1}, \cop{\psi^4}{2}\>\cop{\psi^4}{1} \right] \nonumber \\
    &= \<\cop{b}{4}, \cop{\psi^1}{2}\>\rop[2]{\cop{a}{2}\cop{h}{2}S\cop{h}{4}}{\cop{h}{5}}\(\cop{\psi^1}{1}\) \tens \<\cop{b}{3}, \cop{\psi^2}{2}\>\rop[1]{\cop{a}{3}}{\cop{h}{3}}\(\cop{\psi^2}{1}\) \nonumber \\
    &\qquad \tens \<\cop{b}{2}, \cop{\psi^3}{2}\>\cop{\psi^3}{1} \tens \<\cop{b}{1}, \cop{\psi^4}{2}\>\rop[3]{\cop{a}{1}\cop{h}{1}S\cop{h}{6}}{\cop{h}{7}}\(\cop{\psi^4}{1}\)\nonumber \\
    &= \<\cop{b}{4}, \cop{\psi^1}{2}\> \eps\(\cop{h}{5}\)\ \tilde{T}^{\cop{a}{2}\cop{h}{2}S\cop{h}{4}}_{-}\(\cop{\psi^1}{1}\) \tens   \eps\(\cop{a}{3}\)\ \tilde{L}^{\cop{h}{3}}_{+} \(\cop{\psi^2}{1}\) \nonumber \\
    &\qquad \tens \<\cop{b}{2}, \cop{\psi^3}{2}\>\cop{\psi^3}{1} \tens \<\cop{b}{1}, \cop{\psi^4}{2}\> \eps\(\cop{a}{1}\cop{h}{1}S\cop{h}{6}\)\ \tilde{L}^{\cop{h}{7}}_{-} \(\cop{\psi^4}{1}\)\nonumber \\
    &= \eps\(\cop{h}{5}\)\<\Si\(\cop{a}{2}\cop{h}{2}S\cop{h}{4}\), \copp{\psi^1}{1}{2}\>\<\cop{b}{4}, \cop{\psi^1}{2}\>\copp{\psi^1}{1}{1} \nonumber \\
    &\qquad \tens \eps\(\cop{a}{3}\)\<\cop{h}{3},  S\copp{\psi^2}{1}{3}\copp{\psi^2}{1}{1}\>\<\cop{b}{3}, \cop{\psi^2}{2}\>\copp{\psi^2}{1}{2} \nonumber \\
    &\qquad \tens \<\cop{b}{2}, \cop{\psi^3}{2}\>\cop{\psi^3}{1} \tens \eps\(\cop{a}{1}\cop{h}{1}S\cop{h}{6}\)\<\cop{h}{7}, \copp{\psi^4}{1}{1}S^{-1}\copp{\psi^4}{1}{3} \>\<\cop{b}{1}, \cop{\psi^4}{2}\> \copp{\psi^4}{1}{2}\nonumber \\
    &= \eps\(\cop{h}{5}\)\<\Si\(\cop{a}{2}\cop{h}{2}S\cop{h}{4}\), \cop{\psi^1}{2}\>\<\cop{b}{4}, \cop{\psi^1}{3}\>\cop{\psi^1}{1} \nonumber \\
    &\qquad \tens \eps\(\cop{a}{3}\)\<\cop{h}{3},  S\cop{\psi^2}{3}\cop{\psi^2}{1}\>\<\cop{b}{3}, \cop{\psi^2}{4}\>\cop{\psi^2}{2} \nonumber \\
    &\qquad \tens \<\cop{b}{2}, \cop{\psi^3}{2}\>\cop{\psi^3}{1} \tens \eps\(\cop{a}{1}\cop{h}{1}S\cop{h}{6}\)\<\cop{h}{7}, \cop{\psi^4}{1}S^{-1}\cop{\psi^4}{3}\>\<\cop{b}{1}, \cop{\psi^4}{4}\> \cop{\psi^4}{2} \nonumber \\
    &= \<b, \cop{\psi^1}{3}\cop{\psi^2}{4}\cop{\psi^3}{2}\cop{\psi^4}{4}\> \eps\(\cop{h}{5}\)\<\Si\(\cop{a}{2}\cop{h}{2}S\cop{h}{4}\), \cop{\psi^1}{2}\>\cop{\psi^1}{1} \nonumber \\
    &\qquad \tens \eps\(\cop{a}{3}\)\<\cop{h}{3},  S\cop{\psi^2}{3}\cop{\psi^2}{1}\>\cop{\psi^2}{2} \tens \cop{\psi^3}{1} \tens \eps\(\cop{a}{1}\cop{h}{1}S\cop{h}{6}\)\<\cop{h}{7}, \cop{\psi^4}{1}S^{-1}\cop{\psi^4}{3}\> \cop{\psi^4}{2} \nonumber \\
    &= \eps\(\cop{h}{5}\)\<\Si\(\cop{a}{2}\cop{h}{2}S\cop{h}{4}\), \cop{\psi^1}{2}\> \cop{\psi^1}{1} \tens \eps\(\cop{a}{3}\)\<\cop{h}{3}, S\cop{\psi^2}{3}\cop{\psi^2}{1}\> \cop{\psi^2}{2} \nonumber \\
    &\qquad  \tens \psi^3 \tens \eps\(\cop{a}{1}\cop{h}{1}S\cop{h}{6}\)\<\cop{h}{7}, \cop{\psi^4}{1}S^{-1}\cop{\psi^4}{3}\> \cop{\psi^4}{2}
\end{align}

\begin{equation}
    \begin{minipage}[c]{0.05\textwidth}
        $\ket{\Sigma}$
    \end{minipage}
    \begin{minipage}[c]{0.05\textwidth}
        $=$
    \end{minipage}
    \begin{minipage}[c]{0.5\textwidth}
        \begin{tikzpicture}[scale = 0.8]

    \begin{scope}[thick, decoration={markings, mark=at position 0.4 with {\arrow{latex}}}]
        \draw[postaction={decorate}] (3*\rcs, 4*\rcs) -- (7*\rcs, 4*\rcs);
        \draw[postaction={decorate}] (7*\rcs, 0*\rcs) -- (7*\rcs, 4*\rcs);
        \draw[postaction={decorate}] (11*\rcs, 4*\rcs) -- (7*\rcs, 4*\rcs);
        \draw[postaction={decorate}] (7*\rcs, 8*\rcs) -- (7*\rcs, 4*\rcs);
    \end{scope}

    \begin{scope}[thin, red, dashed, decoration={markings, mark=at position 0.3 with {\arrow{latex}}}]
        \draw[postaction={decorate}] (9*\rcs, 2*\rcs) -- (5*\rcs, 2*\rcs);
    \end{scope}

    \draw[dotted, very thick] (3*\rcs, 4*\rcs) -- (5*\rcs, 2*\rcs);
    \draw[dotted, very thick] (7*\rcs, 4*\rcs) -- (5*\rcs, 2*\rcs);
    \draw[dotted, very thick] (7*\rcs, 4*\rcs) -- (9*\rcs, 2*\rcs);
    \draw[dotted, very thick] (11*\rcs, 4*\rcs) -- (9*\rcs, 2*\rcs);
    

    \harrow{\rcs}{(4*\rcs, 3*\rcs)}{(6*\rcs, 3*\rcs)}{0.06}{0}
    \harrow{\rcs}{(6*\rcs, 3*\rcs)}{(8*\rcs, 3*\rcs)}{0.06}{1}
    \harrow{\rcs}{(8*\rcs, 3*\rcs)}{(10*\rcs, 3*\rcs)}{0.06}{0}

    \draw[text=black, draw=white, fill=white] (9*\rcs, 4.5*\rcs) circle (5pt) node {\footnotesize $\psi^{1}$};
    \draw[text=black, draw=white, fill=white] (6.5*\rcs, 6*\rcs) circle (5pt) node {\footnotesize $\psi^{2}$};
    \draw[text=black, draw=white, fill=white] (5*\rcs, 4.5*\rcs) circle (5pt) node {\footnotesize $\psi^{3}$};
    \draw[text=black, draw=white, fill=white] (6.5*\rcs, 1.0*\rcs) circle (5pt) node {\footnotesize $\psi^{4}$};

    \draw[text=black, draw=white, fill=white] (8*\rcs, 3*\rcs) circle (5pt) node {\footnotesize $s$};
    \draw[text=black, draw=white, fill=white] (5*\rcs, 3.5*\rcs) circle (5pt) node {\footnotesize $\tau_{1}$};
    \draw[text=black, draw=white, fill=white] (7.0*\rcs, 2.5*\rcs) circle (5pt) node {\footnotesize $\tau_{2}$};
    \draw[text=black, draw=white, fill=white] (9.0*\rcs, 3.5*\rcs) circle (5pt) node {\footnotesize $\tau_{3}$};





\end{tikzpicture}
    \end{minipage}
    \label{type_A_ribbon_with_vertex}
\end{equation}

\begin{align}
    A^{g}(s)F^{a\tens h}\(\rho^{A}\)\ket{\Sigma} &= A^{g}(s)\left[\rop[3]{\cop{a}{1}\cop{h}{1}S\cop{h}{6}}{\cop{h}{7}}\(\psi^1\) \tens \psi^2 \tens \rop[1]{\cop{a}{3}}{\cop{h}{3}}\(\psi^3\) \tens \rop[2]{\cop{a}{2}\cop{h}{2}S\cop{h}{4}}{\cop{h}{5}}\(\psi^4\)\right] \nonumber \\
    &= A^{g}(s)\left[\eps\(\cop{h}{7}\)\ \tilde{T}^{\cop{a}{1}\cop{h}{1}S\cop{h}{6}}_{-}\(\psi^1\) \tens \psi^2 \tens \eps\(\cop{h}{3}\)\ \tilde{T}^{\cop{a}{3}}_{+}\(\psi^3\) \tens \eps\(\cop{a}{2}\cop{h}{2}S\cop{h}{4}\)\ \tilde{L}^{\cop{h}{5}}_{-} \(\psi^4\)\right] \nonumber \\
    &= A^{g}(s)\left[ \eps\(\cop{h}{7}\)\< S\(\cop{a}{1}\cop{h}{1}S\cop{h}{6}\), \cop{\psi^1}{2}\> \cop{\psi^1}{1} \tens \psi^2 \tens \eps\(\cop{h}{3}\)\<\cop{a}{3}, \cop{\psi^3}{1}\>\cop{\psi^3}{2}\right. \nonumber \\
    &\qquad \left. \tens \eps\(\cop{a}{2}\cop{h}{2}S\cop{h}{4}\)\<\cop{h}{5}, \cop{\psi^4}{1}S^{-1}\cop{\psi^4}{3}\>\cop{\psi^4}{2} \right] \nonumber \\
    &= \eps\(\cop{h}{7}\)\< \Si\(\cop{a}{1}\cop{h}{1}S\cop{h}{6}\), \cop{\psi^1}{2}\> \PP{1}{\cop{g}{4}}\(\cop{\psi^1}{1}\) \tens \PP{\cop{g}{3}S\cop{g}{5}}{\cop{g}{6}}\(\psi^2\) \nonumber \\
    &\qquad \tens \eps\(\cop{h}{3}\)\<\cop{a}{3}, \cop{\psi^3}{1}\>\PP{\cop{g}{2}S\cop{g}{7}}{\cop{g}{8}}\(\cop{\psi^3}{2}\) \nonumber \\
    &\qquad  \tens \eps\(\cop{a}{2}\cop{h}{2}S\cop{h}{4}\)\<\cop{h}{5}, \cop{\psi^4}{1}S^{-1}\cop{\psi^4}{3}\>\PP{\cop{g}{1}S\cop{g}{9}}{\cop{g}{10}}\(\cop{\psi^4}{2}\) \nonumber \\
    &= \eps\(\cop{h}{7}\)\< \Si\(\cop{a}{1}\cop{h}{1}S\cop{h}{6}\), \cop{\psi^1}{2}\> \<S\copp{g}{4}{1}, \copp{\psi^1}{1}{1}\>  \<\copp{g}{4}{2}, \copp{\psi^1}{1}{3}\> \copp{\psi^1}{1}{2} \nonumber \\
    &\qquad \tens \<S\copp{g}{6}{1}S\(\cop{g}{3}S\cop{g}{5}\), \cop{\psi^2}{1} \> \<\copp{g}{6}{2}, \cop{\psi^2}{3}\> \cop{\psi^2}{2} \nonumber \\
    &\qquad \tens \eps\(\cop{h}{3}\)\<\cop{a}{3}, \cop{\psi^3}{1}\> \< S\copp{g}{8}{1}S\(\cop{g}{2}S\cop{g}{7}\), \copp{\psi^3}{2}{1} \> \<\copp{g}{8}{2}, \copp{\psi^3}{2}{3}\> \copp{\psi^3}{2}{2} \nonumber \\
    &\qquad \tens \eps\(\cop{a}{2}\cop{h}{2}S\cop{h}{4}\)\<\cop{h}{5}, \cop{\psi^4}{1}S^{-1}\cop{\psi^4}{3}\> \<S\copp{g}{10}{1}S\(\cop{g}{1}S\cop{g}{9}\), \copp{\psi^4}{2}{1}\> \<\copp{g}{10}{2}, \copp{\psi^4}{2}{3} \> \copp{\psi^4}{2}{2} \nonumber \\
    &= \eps\(\cop{h}{7}\)\< \Si\(\cop{a}{1}\cop{h}{1}S\cop{h}{6}\), \cop{\psi^1}{4}\> \<S\cop{g}{4}, \cop{\psi^1}{1}\>  \<\cop{g}{5}, \cop{\psi^1}{3}\> \cop{\psi^1}{2} \nonumber \\
    &\qquad \tens \<S\cop{g}{7}S\(\cop{g}{3}S\cop{g}{6}\), \cop{\psi^2}{1} \> \<\cop{g}{8}, \cop{\psi^2}{3}\> \cop{\psi^2}{2} \nonumber \\
    &\qquad \tens \eps\(\cop{h}{3}\)\<\cop{a}{3}, \cop{\psi^3}{1}\> \< S\cop{g}{10}S\(\cop{g}{2}S\cop{g}{9}\), \cop{\psi^3}{2} \> \<\cop{g}{11}, \cop{\psi^3}{4}\> \cop{\psi^3}{3} \nonumber \\
    &\qquad \tens \eps\(\cop{a}{2}\cop{h}{2}S\cop{h}{4}\)\<\cop{h}{5}, \cop{\psi^4}{1}S^{-1}\cop{\psi^4}{5}\> \<S\cop{g}{13}S\(\cop{g}{1}S\cop{g}{12}\), \cop{\psi^4}{2}\> \<\cop{g}{14}, \cop{\psi^4}{4} \> \cop{\psi^4}{3} \nonumber \\
    &=\eps\(\cop{h}{7}\)\<\Si\(\cop{a}{1}\cop{h}{1}S\cop{h}{6}\), \cop{\psi^1}{4}\>\<\cop{g}{4}, S\cop{\psi^1}{1}\cop{\psi^1}{3}\>\cop{\psi^1}{2} \tens \<S\cop{g}{3}, \cop{\psi^2}{1}\>\<\cop{g}{5}, \cop{\psi^2}{3} \>\cop{\psi^2}{2} \nonumber \\
    &\qquad \tens \eps\(\cop{h}{3}\)\<\cop{a}{3}, \cop{\psi^3}{1}\>\<S\cop{g}{2}, \cop{\psi^3}{2}\>\<\cop{g}{6}, \cop{\psi^3}{4}\>\cop{\psi^3}{3} \nonumber \\
    &\qquad \tens \eps\(\cop{a}{2}\cop{h}{2}S\cop{h}{4}\)\<\cop{h}{5}, \cop{\psi^4}{1}S^{-1}\cop{\psi^4}{5}\>\<S\cop{g}{1}, \cop{\psi^4}{2}\>\<\cop{g}{7}, \cop{\psi^4}{4}\>\cop{\psi^4}{3} \nonumber \\
    &=\<g, S\cop{\psi^4}{2}S\cop{\psi^3}{2}S\cop{\psi^2}{1}S\cop{\psi^1}{1}\cop{\psi^1}{3}\cop{\psi^2}{3}\cop{\psi^3}{4}\cop{\psi^4}{4}\>\eps\(\cop{h}{7}\)\<\Si\(\cop{a}{1}\cop{h}{1}S\cop{h}{6}\), \cop{\psi^1}{4}\>\cop{\psi^1}{2} \tens \cop{\psi^2}{2} \nonumber \\
    &\qquad \tens \eps\(\cop{h}{3}\)\<\cop{a}{3}, \cop{\psi^3}{1}\>\cop{\psi^3}{3} \tens \eps\(\cop{a}{2}\cop{h}{2}S\cop{h}{4}\)\<\cop{h}{5}, \cop{\psi^4}{1}S^{-1}\cop{\psi^4}{5}\>\cop{\psi^4}{3} \nonumber \\
    &=\eps\(\cop{h}{7}\)\<\Si\(\cop{a}{1}\cop{h}{1}S\cop{h}{6}\), \cop{\psi^1}{2}\>\cop{\psi^1}{1} \tens \psi^2 \tens \eps\(\cop{h}{3}\)\<\cop{a}{3}, \cop{\psi^3}{1}\>\cop{\psi^3}{2} \nonumber \\
    &\qquad \tens \eps\(\cop{a}{2}\cop{h}{2}S\cop{h}{4}\)\<\cop{h}{5}, \cop{\psi^4}{1}S^{-1}\cop{\psi^4}{3}\>\cop{\psi^4}{2} 
\end{align}

\begin{align}
    F^{a\tens h}\(\rho^{A}\)A^{g}(s)\ket{\Sigma} &= \rop{a}{h}\left[\PP{1}{\cop{g}{4}}\(\psi^1\) \tens \PP{\cop{g}{3}S\cop{g}{5}}{\cop{g}{6}}\(\psi^2\) \tens \PP{\cop{g}{2}S\cop{g}{7}}{\cop{g}{8}}\(\psi^3\) \tens \PP{\cop{g}{1}S\cop{g}{9}}{\cop{g}{10}}\(\psi^4\) \right] \nonumber \\
    &= \rop{a}{h}\left[\<S\copp{g}{4}{1}, \cop{\psi^1}{1}\>\<\copp{g}{4}{2}, \cop{\psi^1}{3}\>\cop{\psi^1}{2} \tens \<S\copp{g}{6}{1}S\(\cop{g}{3}S\cop{g}{5}\), \cop{\psi^2}{1}\>\<\copp{g}{6}{2}, \cop{\psi^2}{3}\>\cop{\psi^2}{2} \right. \nonumber \\
    &\qquad \left. \tens \<S\copp{g}{8}{1}S\(\cop{g}{2}S\cop{g}{7}\), \cop{\psi^3}{1}\>\<\copp{g}{8}{2}, \cop{\psi^3}{3}\>\cop{\psi^3}{2} \right. \nonumber \\
    &\qquad \left. \tens \<S\copp{g}{10}{1}S\(\cop{g}{1}S\cop{g}{9}\), \cop{\psi^4}{1}\>\<\copp{g}{10}{2}, \cop{\psi^4}{3}\>\cop{\psi^4}{2} \right] \nonumber \\
    &= \rop{a}{h}\left[\<\cop{g}{4}, S\cop{\psi^1}{1}\cop{\psi^1}{3}\>\cop{\psi^1}{2} \tens \<S\cop{g}{3}, \cop{\psi^2}{1}\>\<\cop{g}{5}, \cop{\psi^2}{3}\>\cop{\psi^2}{2} \right. \nonumber \\
    &\qquad \left. \tens \<S\cop{g}{2}, \cop{\psi^3}{1}\>\<\cop{g}{6}, \cop{\psi^3}{3}\>\cop{\psi^3}{2} \tens \<S\cop{g}{1}, \cop{\psi^4}{1}\>\<\cop{g}{7}, \cop{\psi^4}{3}\>\cop{\psi^4}{2} \right] \nonumber \\
    &= \rop{a}{h}\left[\<\cop{g}{4}, S\cop{\psi^1}{1}\cop{\psi^1}{3}\> \rop[3]{\cop{a}{1}\cop{h}{1}S\cop{h}{6}}{\cop{h}{7}}\(\cop{\psi^1}{2}\) \right. \nonumber \\
    &\qquad \left. \tens \<S\cop{g}{3}, \cop{\psi^2}{1}\>\<\cop{g}{5}, \cop{\psi^2}{3}\>\cop{\psi^2}{2} \right. \nonumber \\
    &\qquad \left. \tens \<S\cop{g}{2}, \cop{\psi^3}{1}\>\<\cop{g}{6}, \cop{\psi^3}{3}\>\rop[2]{\cop{a}{3}}{\cop{h}{3}}\(\cop{\psi^3}{2}\) \right. \nonumber \\
    &\qquad \left. \tens \<S\cop{g}{1}, \cop{\psi^4}{1}\>\<\cop{g}{7}, \cop{\psi^4}{3}\>\rop[3]{\cop{a}{2}\cop{h}{2}S\cop{h}{4}}{\cop{h}{5}}\(\cop{\psi^4}{2}\) \right] \nonumber \\
    &= \<\cop{g}{4}, S\cop{\psi^1}{1}\cop{\psi^1}{3}\> \eps\(\cop{h}{7}\)\TMt{\cop{a}{1}\cop{h}{1}S\cop{h}{6}}\(\cop{\psi^1}{2}\) \nonumber \\
    &\qquad \tens \<S\cop{g}{3}, \cop{\psi^2}{1}\>\<\cop{g}{5}, \cop{\psi^2}{3}\>\cop{\psi^2}{2} \nonumber \\
    &\qquad \tens \<S\cop{g}{2}, \cop{\psi^3}{1}\>\<\cop{g}{6}, \cop{\psi^3}{3}\>\eps\(\cop{h}{3}\)\TPt{\cop{a}{3}}\(\cop{\psi^3}{2}\) \nonumber \\
    &\qquad \tens \<S\cop{g}{1}, \cop{\psi^4}{1}\>\<\cop{g}{7}, \cop{\psi^4}{3}\>\eps\(\cop{a}{2}\cop{h}{2}S\cop{h}{4}\)\LMt{\cop{h}{5}}\(\cop{\psi^4}{2}\) \nonumber \\
    &= \<\cop{g}{4}, S\cop{\psi^1}{1}\cop{\psi^1}{3}\> \eps\(\cop{h}{7}\)\<\Si\(\cop{a}{1}\cop{h}{1}S\cop{h}{6}\), \copp{\psi^1}{2}{1} \>\copp{\psi^1}{2}{2} \nonumber \\
    &\qquad \tens \<S\cop{g}{3}, \cop{\psi^2}{1}\>\<\cop{g}{5}, \cop{\psi^2}{3}\>\cop{\psi^2}{2} \nonumber \\
    &\qquad \tens \<S\cop{g}{2}, \cop{\psi^3}{1}\>\<\cop{g}{6}, \cop{\psi^3}{3}\>\eps\(\cop{h}{3}\)\<\cop{a}{3}, \copp{\psi^3}{2}{1}\>\copp{\psi^3}{2}{2} \nonumber \\
    &\qquad \tens \<S\cop{g}{1}, \cop{\psi^4}{1}\>\<\cop{g}{7}, \cop{\psi^4}{3}\>\eps\(\cop{a}{2}\cop{h}{2}S\cop{h}{4}\)\<\cop{h}{5}, \copp{\psi^4}{2}{1}\Si\copp{\psi^4}{2}{3}\> \copp{\psi^4}{2}{2} \nonumber \\
    &= \<g, S\cop{\psi^4}{1}S\cop{\psi^3}{1}S\cop{\psi^2}{1}S\cop{\psi^1}{1}\cop{\psi^1}{4}\cop{\psi^2}{3}\cop{\psi^3}{5}\cop{\psi^4}{5}\> \eps\(\cop{h}{7}\)\<\Si\(\cop{a}{1}\cop{h}{1}S\cop{h}{6}\), \cop{\psi^1}{2}\>\cop{\psi^1}{3} \tens \cop{\psi^2}{2} \nonumber \\
    &\qquad \tens \eps\(\cop{h}{3}\)\<\cop{a}{3}, \cop{\psi^3}{2}\>\cop{\psi^3}{3} \tens \eps\(\cop{a}{2}\cop{h}{2}S\cop{h}{4}\)\<\cop{h}{5}, \cop{\psi^4}{2}\Si\cop{\psi^4}{4}\> \cop{\psi^4}{3} \nonumber \\
    &=\eps\(\cop{h}{7}\)\<S\(\cop{a}{1}\cop{h}{1}S\cop{h}{6}\), \cop{\psi^1}{2}\>\cop{\psi^1}{1} \tens \psi^2 \tens \eps\(\cop{h}{3}\)\<\cop{a}{3}, \cop{\psi^3}{1}\>\cop{\psi^3}{2} \nonumber \\
    &\qquad \tens \eps\(\cop{a}{2}\cop{h}{2}S\cop{h}{4}\)\<\cop{h}{5}, \cop{\psi^4}{1}S^{-1}\cop{\psi^4}{3}\>\cop{\psi^4}{2}
\end{align}

\end{document}